\documentclass{nature}
\usepackage{hyperref}
\usepackage{graphicx}
\usepackage{amsmath}
\usepackage{amssymb}
\usepackage{comment}
\usepackage{braket}
\usepackage{caption}
\usepackage{color}
\usepackage{etex}
\usepackage{amsthm}
\usepackage{amsfonts}
\usepackage{dsfont}
\usepackage{appendix}
\newtheorem{observation}{Comment}
\RequirePackage{etex}
\newcommand{\blk}{\color{black}}

\newcommand{\adri}[1]{{\color{black} #1}} 
\newcommand{\dam}[1]{{\color{black} #1}}
\newcommand{\adr}[1]{{\color{black} #1}}
\newcommand{\yang}[1]{{\color{black} #1}}
\newcommand{\yangnew}[1]{{\color{black} #1}}
\newcommand{\dami}[1]{{\color{black} #1}}
\newcommand{\damian}[1]{{\color{black} #1}}

\newtheorem{theorem}{Theorem}
\newtheorem{lemma}{Lemma}
\newtheorem{proposition}{Proposition}
\newtheorem{definition}{Definition}

\usepackage{siunitx}
\usepackage{tikz}
\usepackage[keys,primitives]{cryptocode}
\usepackage{tabularx}
\usepackage{lipsum}
\usepackage[utf8]{inputenc}
\RequirePackage{etex}
\everymath{\displaystyle}

\makeatletter
\let\saved@includegraphics\includegraphics
\AtBeginDocument{\let\includegraphics\saved@includegraphics}
\renewenvironment*{figure}{\@float{figure}}{\end@float}

\makeatother

\title{\dam{Experimental practical quantum tokens with transaction time advantage}}

\author{Yang{-}Fan Jiang$^{1*}$, Adrian Kent$^{2,3*}$, Dami\'an Pital\'ua-Garc\'ia$^{2*}$, Xiaochen Yao$^{4}$, Xiao-Han Chen$^{1}$, Jia Huang$^{5}$, George Cowperthwaite$^{2}$, Qibin Zheng$^{4}$, Hao Li$^{5}$, Lixing You$^{5}$, Yang Liu$^{1}$, Qiang Zhang$^{1,6,7}$ and Jian-Wei Pan$^{6,7}$}

\begin{document}
\begin{spacing}{1.0}
\begin{sloppypar}
\maketitle

\begin{affiliations}
\item Jinan Institute of Quantum Technology and Hefei National Laboratory Jinan Branch, Jinan 250101, China
 \item Centre for Quantum Information and Foundations, DAMTP, Centre for Mathematical Sciences, University of Cambridge, Wilberforce Road, Cambridge, CB3 0WA, United Kingdom
 \item Perimeter Institute for Theoretical Physics, 31 Caroline Street North, Waterloo, ON N2L 2Y5, Canada
 \item Quantum Medical Sensing Laboratory and School of Health Science and Engineering, University of Shanghai for Science and Technology, Shanghai 200093, China
 \item Shanghai Key Laboratory of Superconductor Integrated Circuit Technology, Shanghai Institute of Microsystem and Information Technology, Chinese Academy of Sciences, Shanghai 200050, China
\item Hefei National Laboratory, University of Science and Technology of China, Hefei 230088, China
\item Hefei National Research Center for Physical Sciences at the Microscale and School of Physical Sciences, University of Science and Technology of China, Hefei 230026, China
 \item[] $^{*}$First co-authors in alphabetical order
\end{affiliations}

\maketitle

\begin{abstract}
Quantum money~\cite{W83} is the first \dam{invention in quantum information science}, 
promising advantages over classical money by simultaneously achieving unforgeability, user privacy, and instant validation. However, \dam{standard} quantum money~\cite{W83,G12,MVW12,PYLLC12,GK15,MP16,AA17,BDG19,K19,HS20} rel\dam{ies} on quantum memories \dam{and long-distance quantum communication, which are technologically extremely challenging}. \dam{Quantum ``S-money" tokens~\cite{KSmoney,quantumtokenspat,KPG20,KLPGR22} eliminate these technological requirements while preserving unforgeability, user privacy, and instant validation}. Here, we report the first full experimental demonstration of quantum S-tokens, \dam{proven secure despite errors, losses and experimental imperfections.} 
The heralded single-photon source with a high system efficiency of 88.24\% \dam{protects against arbitrary multi-photon attacks~\cite{BCDKPG21} arising from losses in the quantum token generation}. Following short-range quantum communication, the token is stored, transacted, and verified using classical bits. We demonstrate \adr{a transaction time advantage over} intra-city 2.77 km and inter-city 60.54 km optical fibre network\dam{s}, compared with \adr{optimal} classical cross-checking \dam{schemes}. \dam{Our implementation demonstrates the practicality of quantum S-tokens for applications requiring \adr{high security, privacy 
and minimal transaction times}, like financial trading~\cite{WF10} and network control}.
It is also the first demonstration of a \dam{quantitative} quantum \adr{time} advantage in relativistic cryptography, showing the enhanced cryptographic power of simultaneously considering quantum and relativistic physics. 
\end{abstract}

\maketitle

{\it Introduction.---}Money plays a pivotal role \dam{in society}. 
Quantum money tokens, \adr{first proposed by Wiesner~\cite{W83} in 1970, guarantee} information-theoretic unforgeability \dam{from} the no-cloning theorem~\cite{WZ82,D82} \adr{of quantum information}; and ensure user privacy (nobody else knows when or where the user will spend) and instant validation (can be validated locally without communicating with distant locations)~\cite{KSmoney}. It has only recently been recognized that the fundamental advantage of quantum tokens over standard classical alternatives is in satisfying \adr{all} three properties simultaneously with unconditional security~\cite{KSmoney,KLPGR22}. For example, purely classical token schemes that cross-check for multiple presentations
among distant locations can give user privacy and unforgeability, but not instant validation;
\adri{if the user announces the presentation point in advance, they can give unforgeability and instant validation, but not user privacy; without cross-checking, they can give instant validation and user privacy but not \damian{unconditionally secure} unforgeability}~\cite{KSmoney,KLPGR22}. \adri{Examples of quantum token applications are envisaged in \adri{future} high-speed financial trading~\cite{KLPGR22}, where instant validation\adri{, user privacy and unconditional unforgeability will be} crucial to avoid relativistic signalling delays~\cite{WF10}, keep trading plans secret and prevent fraud.}

\dam{Standard quantum token schemes~\cite{W83,G12,MVW12,PYLLC12,GK15,MP16,AA17,BDG19,K19,HS20} are technologically very challenging as they require quantum memories and long-distance quantum communication. Despite remarkable progress in quantum memories~\cite{yang2016efficient,jiang2019experimental,wang2021single} and long-range quantum communication~\cite{chen2021integrated,liu2023experimental}, we are still far from implementing quantum token systems over useful time and distance scales. Hence, experimental investigations of standard quantum tokens~\cite{BCGLMN17,BBP17,BOVZKD18,GAAZLYWZP18,JBCL19},
while valuable, have not yet demonstrated practically useful schemes. }

\dam{A new class of quantum tokens, “S-money” (S-tokens)~\cite{KSmoney,quantumtokenspat}, achieves the three properties above with unconditional security but \adr{needs neither} quantum memory \adr{nor} long-range quantum communication, and is thus practical with current technology. The tokens can be generated \adr{far} in advance of their use and can be transferred
among parties~\cite{KPG20}.} \adr{Previous work\cite{KLPGR22} 
reported the quantum stage of S-token generation
and analysed security against experimental imperfections, but did not achieve full security or include real-time 
token presentation and validation.}

Here we present a full experimental \adri{implementation} of the quantum S-token scheme of Refs.~\cite{KSmoney,KPG20,KLPGR22}, \adri{demonstrating for the first time quantum tokens with near-perfect security and user privacy and  near-instant validation} in a realistic photonic setup that considers losses, errors and experimental imperfections.
We \adr{present a new security analysis based on maximal confidence \dam{quantum} measurement bounds \dam{\cite{CABGJ06}} to} \dami{prove \adr{unforgeability}.} 
\adr{Due to a high detection heralding efficiency of 88.24\%, our implementation allows the correct presentation and validation of tokens while guaranteeing user privacy by not requiring the user to report losses when the token is generated.  This prevents multi-photon \dam{attacks \cite{BCDKPG21},} to which many previous experimental demonstrations of mistrustful quantum cryptography (e.g., \cite{NJMKW12,LKBHTKGWZ13,LCCLWCLLSLZZCPZCP14,ENGLWW14,PJLCLTKD14,KLPGR22}) were vulnerable.} 
\adr{Moreover, our} experiment demonstrates for the first time a secure quantum token scheme that achieves faster validation than classical cross-checking can, given relativistic signalling constraints.   

\yang{By developing a data transmission and processing board with a 10 Gigabaud rate}, we demonstrate quantum token transactions using a metropolitan fibre network, showing transaction time advantages over any \dam{classical cross-checking token schemes}. We also demonstrate these advantages \adr{in short-\adr{range} urban fibre optic networks}.

{\it The scheme.---}
\adr{Following Refs.~\cite{KSmoney,KPG20,KLPGR22}}, \dam{Bob \adr{(the bank)} sends Alice \adr{(the client)} the quantum token via a short-range quantum channel}. \adr{Alice} measures the quantum states immediately upon reception without using quantum memories.  The scheme then proceeds with classical communications, \adr{without} further quantum communications.  

\yang{\adr{We consider a two-node network scenario (see Fig.~\dami{S3 in SI}) involving small spatial regions $L_0, L_1$, defined in an agreed reference frame $F$.} Alice (Bob) comprises two collaborating and mutually trusting agents or laboratories A$_0$ and A$_1$ (B$_0$ and B$_1$), located in $L_0$ and $L_1$ respectively connected via secure and authenticated classical channels (implemented via pre-distributed secret keys, for instance).} 
\adri{In general, Alice and Bob may be companies or governments with many distributed trusted agents: the two-node two-agent scenario models the simplest case (see SI for extensions to larger networks).   Note that any money or token scheme applicable in time-critical scenarios on extended networks needs distributed trusted agents, since tokens can be transmitted at near light speed while individual agents cannot.}
\adr{We define the space-time region} $R_i$ to comprise the location \dam{$L_i$} within a time interval $\Delta T$ beginning at \dam{a time $T$ \adr{in $F$}. Let} $\Delta T_\text{comm}$ be the time taken for A$_0$ to communicate \dam{a bit} to A$_1$. The values of $\Delta T$  \dam{and} $\Delta T_\text{comm}$ are agreed in advance by Alice and Bob, \adr{with} $\Delta T$  \adr{large enough} to allow Alice to present the token to Bob. \dam{We} assume that $T$ is only communicated by Alice to Bob at step 3 \yangnew{in Table~1} via the relation
\begin{equation}
\label{time}
  T=T_\text{bit}+ \Delta T_\text{comm} \, ,
\end{equation}
where $T_\text{bit}$ is the time at which A$_0$ communicates the bit $c$ to B$_0$ in step 3.
\adr{We assume that B$_0$ can communicate \dam{a bit} to B$_1$ within time $\Delta T_\text{comm}$, so B$_1$ is ready to receive and verify the token if and when presented by A$_1$. Alice and Bob also agree on a maximum error rate $\gamma_{err}$.}

We describe the complete procedure for \adr{an} \emph{ideal scheme} \yangnew{in Table~1}, which is extended to allow for experimental imperfection in a \emph{practical scheme} (see Methods); these \adr{are minor variations} of the schemes $\mathcal{IQT}_2$ and $\mathcal{QT}_2$ of Ref.~\cite{KLPGR22}. \yangnew{The quantum token preparation phase can be performed arbitrarily in advance of the following stage. The transaction phase requires high-speed data transmission and processing to give an advantage over purely classical schemes.}

\begin{table}

\caption{The complete procedure for \adr{an} \emph{ideal scheme}.}
{\noindent}\rule[0pt]{16.5cm}{0.05em}
{\noindent}\rule[10pt]{16.5cm}{0.05em}
\begin{spacing}{1.0}
Below, $\textbf{x},\textbf{x}',\textbf{x}^0,\textbf{x}^1,\textbf{t},\textbf{u}$ are $N-$bit strings, $z,b,c,d_0,d_1$ are bits, and $s_k$ is \adr{the} $k$th bit of \adr{the string $\textbf{s}$}.

~

\emph{The quantum phase for the token preparation.}
\begin{itemize}
\item[1.] B$_0$ sends A$_0$ $N$ random states from the BB84 set $\{\lvert 0\rangle, \lvert 1\rangle,\lvert +\rangle,\lvert -\rangle\}$~\cite{BB84}. A$_0$ randomly chooses $z$ and measures all the states in the qubit orthonormal basis $\mathcal{D}_{z}$, where $\mathcal{D}_{0}=\{\lvert 0\rangle,\lvert 1\rangle\}$ and $\mathcal{D}_{1}=\{\lvert +\rangle,\lvert -\rangle\}$. Let $\textbf{t}$ denote B$_0$'s encoded bits and $\textbf{u}$ the preparation bases. Let $\textbf{x}$ denote A$_0$'s outcomes.
\item[2.] A$_0$ \dam{generates a random dummy token $\textbf{x}'$,} keeps copies of $\textbf{x}$ \dam{and $\textbf{x}'$}, and sends copies to A$_1$; while B$_0$ keeps copies of $\textbf{t}$ and $\textbf{u}$ and sends copies to B$_1$. 
\end{itemize}

\emph{The classical phase for the token transaction.}
\begin{itemize}
\item[3.] A$_0$ \dam{obtains} $b$, which labels the location $L_b$ for token presentation, and keeps a copy; she also sends $b$ to A$_1$ at time $T_\text{begin}$ and $c=b\oplus z$ to B$_0$ \adr{as soon as possible after, at time $T_\text{bit} > T_\text{begin}$.}

\item[4.] Upon receiving $c$, B$_0$ keeps a copy and sends another copy to B$_1$.

\item[5.] For $i=0,1$, A$_i$ sends $\textbf{x}^i$ to B$_i$ at location $L_i$, where $\textbf{x}^b=\textbf{x}$ and $\textbf{x}^{b\oplus 1}=\textbf{x}'$. 

\item[6.] For $i=0,1$, at $L_i$, B$_i$ calculates $d_i= c \oplus i$, and computes the number $N_{\text{errors},i}$ of entries $k\in\Delta_i$ that do not satisfy $x_k^i=t_k$, where $\Delta_i=\{k\in[N]\vert u_k=d_i\}$. B$_i$ locally validates the token if
\begin{equation}
\label{allowederrors}
\begin{aligned}
\frac{N_{\text{errors},i}}{N_i}\leq \gamma_{err} \, ,
\end{aligned}
\end{equation}
or rejects it otherwise, where $N_i=\lvert \Delta_i\rvert$. We define $T_\text{end}$ as the time at which token verification or rejection is completed at both $L_0$ and $L_1$.
We define the quantum token scheme \emph{transaction time} by
\begin{equation}
    \label{tran}
  \Delta T_\text{tran}= T_\text{end}-T_\text{begin} \, .   
\end{equation}

The parameters $N$ and $\gamma_{err}$ define the token size and the transaction error tolerance.   They are chosen \adr{to minimize $N$} (and hence transaction times) while ensuring the false rejection probability and the forging probability are suitably low.
\end{itemize}

\end{spacing}
{\noindent}\rule[0pt]{16.5cm}{0.05em}
{\noindent}\rule[25pt]{16.5cm}{0.05em}

\end{table}

 \dam{The practical schemes} extend straightforwardly to an arbitrary number of presentation spacetime 
 regions. Our security analysis applies to this general case \dami{(See SI)}.  \dam{The} unforgeability proof \adr{also} holds even if Alice is required to report losses and applies for arbitrarily powerful dishonest Alice who may detect all quantum states received from Bob and choose to report an arbitrary subset of states as lost \dami{(See Methods and SI)}. \dami{However, if the scheme requires Alice to report losses,  \adr{and she 
does so honestly, she}
 cannot perfectly protect against multiphoton attacks and future privacy is compromised.} \adri{Full security in our implementation -- i.e. unforgeability combined with user privacy -- thus requires the high detection heralding efficiency that avoids the need for loss reporting.} 

\adr{We use two scenarios to quantify our scheme's time advantage over classical protocols in which the agent B$_i$ who receives a valid token waits for a signal confirming the other agent B$_{\bar{i}}$ has not also received one before accepting.
First, we suppose both quantum and classical protocols transmit the signals via the same optical fibre channel. The \textit{quantum advantage} is the transaction time of the classical cross-checking protocol minus that of the quantum S-token scheme.  
Second, we consider a classical cross-checking protocol that transmits direct signals at light speed in free space. The \textit{comparative advantage} is the transaction time of this ideal classical cross-checking protocol minus that of the quantum S-token scheme using a practical (fibre optical and not \dami{straight-line}) channel.} \dami{(See Methods)}

{\it Experimental Implementation.---}In the quantum phase (steps 1-2), the quantum token is prepared using a high-efficiency heralded single-photon source. A schematic, depicted in Fig.~1, consists of two modules: A$_0$ and B$_0$. B$_0$ generates \adr{a} photon pair $\ket{\phi}=\ket{0}_1\otimes \ket{1}_2$ via spontaneous parametric down-conversion (SPDC). One photon is detected by a high-efficiency superconducting nanowire single-photon detector (SNSPD) to generate a trigger signal. The other is modulated with random basis $u_i\in \{0, 1\}$ and bit $t_i\in \{0, 1\}$ using two Pockels cells driven by real-time quantum random number generators (QRNGs)~\cite{abellan2015generation}. After encoding, B$_0$ sends the photon to A$_0$. A$_0$ randomly selects the measurement basis $z\in \{0, 1\}$ using a half-wave plate (HWP) and records the outcomes as $x_i\in \{0, 1\}$, \dam{comprising} the token.

\adr{We employed the only known perfect protection against general multiphoton attacks~\cite{BCDKPG21}, with A$_0$ not reporting losses and accepting all pulses transmitted by B$_0$.}  A$_0$ assigned random measurement outcomes $x_i\in \{0, 1\}$ for the pulses activating none or both of her detectors, introducing errors \adr{for these cases} with probability close to $50\%$. Following Ref~\cite{bennink2010optimal} and employing improved experimental techniques, we measured the final system efficiency of $88.24\%$, including the heralded single-photon source, the encoding by Pockels cells and the measurement by A$_0$. \adr{This high system efficiency gives a 7 standard deviation bound for the} \adr{overall} error rate of only \dami{$E=6.2550\%$}. By setting $\gamma_\text{err}=9.4\%$ and $N=10,048$ bits,
we guarantee our implementation to be \dami{$2.1\times 10^{-11}$}-correct \dam{(Bob rejects a valid token with a probability  $\dami{\leq}\dami{2.1}\times 10^{-11}$)}.
Our implementation is proved \adr{$\beta_E$-private} \dam{(Bob learns Alice's chosen presentation region before she presents with a probability} $\dami{\leq}\frac{1}{2}+\adr{\beta_E}$), \adr{given the perfect protection against multiphoton attacks.
Here the bias $\beta_\text{E}=10^{-5}$ of the bit $z$ encoding A$_0$'s measurement basis in our implementation and can be made arbitrarily small by pre-processing. 
This assumes} that B$_0$ \adr{cannot} exploit side-channels to obtain information about A$_0$'s measurement basis \adr{nor} implement clock synchronization attacks to obtain information about A$_0$'s chosen presentation location prematurely.   \adr{ 
Seven standard deviation \dami{upper} bounds on B$_0$'s biases in selecting the preparation basis and state and the proportion of multi-photon heralded pulses 
were respectively  $\beta_\text{PB} = 
\dami{0.001360}$, $\beta_\text{PS} = 
\dami{0.001120}$, $P_\text{noqub} = 
\dami{4.9\times 10^{-5}}$.} 

Based on the measurement, the uncertainty angle in Bob's state preparation is guaranteed to be $\leq \theta = 
\dami{5.115515^\circ}$ 
with a probability $\geq (1- P_\theta)$, where $P_\theta = 0.027$ (see SI). This shows our scheme to be \dami{$5.52\times 10^{-9}$}-unforgeable (Alice succeeds in getting Bob to validate tokens at both presentation regions with a probability \dami{$\leq 5.52\times 10^{-9}$}). This follows from a novel security analysis based on bounds on Alice's maximum confidence quantum measurement \cite{CABGJ06} for each pulse (\dami{see Methods}).


The classical phase (steps 3-6) was implemented with optical fibre communication links and high-speed electronic boards. To demonstrate the time advantage over classical cross-checking \adr{protocols}, \adr{each A$_i$ communicates to B$_i$ at 10 Gbps} and \adr{ the B$_i$ perform} real-time local validation using high-speed field-programmable gate arrays (FPGAs). The total duration of the classical processing for $N=10,048$ bits without considering the communication time between $L_0$ and $L_1$ is $\Delta T_\text{proc}\approx 1.5~\mu\text{s}$.

We \adr{demonstrated} \textit{quantum advantage} \dam{within the city of Jinan, Shandong Province, between two locations separated by 426 m and connected by 2.766 km of optical fibre, as shown in} Fig.~2(a). A$_0$ decides the presentation location 
and sends the information $b$ to A$_1$ via the fibre. During the transaction, A$_0$ transmits $c$ to B$_0$ using 
high-speed electrical signals; subsequently, B$_0$ \dam{communicates} $c$ to B$_1$ \dam{over} the fibre. 
\dam{A$_b$ sends the token $\textbf{x}$ to B$_b$ at $L_b$ } 
\adr{and A$_{\bar{b}}$ } 
sends \adr{\dam{the dummy token $\textbf{x}'$ to B$_{\bar{b}}$} 
at $L_{\bar{b}}$}\dam{, both of  $N$=10,048 bits.} Upon receiving \dam{$\textbf{x}$ and $\textbf{x}'$, B$_b$ and B$_{\bar{b}}$ process them simultaneously, validating or rejecting them}. The token \dam{was} tested 20 times with the \adr{presentation} location chosen randomly, \dam{obtaining respective average error rate and transaction time of $6.02\%$ and $15.34 \pm 0.01 ~\mu s$}. \dam{The achieved} quantum advantage \dam{was} $12.32\pm 0.01 ~\mu s$, shown in Fig.~3(a)\dam{, demonstrating a} significant time advantage in practical fibre networks, even at short distances.

We further \adr{demonstrated} \textit{comparative advantage} between Yiyuan ($36^{\circ}  10^{\prime}  
50.4^{\prime \prime} N, 118^{\circ}  12^{\prime}  10^{\prime \prime} E$) and Mazhan ($36^{\circ}  0^{\prime}  19^{\prime 
\prime} N, 118^{\circ}  42^{\prime}  35^{\prime \prime} E$) in Shandong Province of China, \adr{separated by 51.60 km and connected by a 60.54 km field-deployable optical fibre,} as shown in Fig.~2(b). We \adr{ran} the transactions 20 
times and measured the average error rate and transaction time were $6.00\%$ and $304.20\pm 0.01 ~\mu s$, achieving the comparative advantage \dam{of} $39.80\pm 0.01 ~\mu s$, 
shown in Fig.~3(b).

The relationship between transaction time and distance is illustrated in Fig.~3(a) and (b). \adr{Assume straight fibre optic channels, the quantum and comparative advantage} can be achieved at approximately 0.3 km and 0.9 km, respectively \dami{(see Methods)}.
\adr{Even though real} optical fibre channels are not straight, our experiment demonstrates \adr{quantum and comparative advantage respectively in an intra- and inter-city network.}

{\it Conclusion.---}\dami{We} have presented the first complete implementation  of provably unforgeable quantum money \adr{tokens with near-instant validation and with user privacy, and with a time advantage over classical schemes.} We implement\adr{ed} a high-efficiency heralded single-photon source, enabling the secure preparation of quantum token\adr{s} against \dam{arbitrary} multiphoton attacks \cite{BCDKPG21}. 
Furthermore, \adr{by using} high-speed data transmission and processing, we have demonstrated a \adr{quantified} time advantage over \adr{optimal} classical \adr{cross-checking} protocols, \adr{even for intra-city networks}. The implementation could ideally be improved further using secure timing and location techniques \dami{(see Methods)}.

\adr{The total transaction times were $\approx  15~\mu{\rm s}$ and $\approx 30\dami{4}~\mu{\rm s} $  for our intra-city and inter-city experiments.  
For comparison, a recent implementation~\cite{SERGTBW23} of a modified version of 
the schemes\cite{KPG20,KLPGR22} required tens of minutes from Alice's choice of presentation point to Bob's validation (see SI for further comparative discussion \dami{and extensions of our schemes}).}

\adri{Quantum S-tokens straightforwardly extend to arbitrarily many presentation regions~\cite{KLPGR22} and extrapolating our implementations shows that quantum and comparative advantage are attainable in real world conditions on multi-node financial and other large-scale networks while maintaining strong security (see SI).}
\adr{The work thus represents} a crucial step \adr{towards} the widespread adoption of \adr{secure }quantum token\adr{s}.

\adri{Our theoretical and experimental techniques apply more broadly to mistrustful quantum cryptography. To our knowledge, this is the first mistrustful relativistic quantum cryptography experiment perfectly closing the multiphoton attacks loophole. We have advanced the theory of quantum token security \cite{KLPGR22} by new results using bounds on Alice's maximum confidence quantum measurement \cite{CABGJ06} for each pulse, which also apply to other quantum token (e.g., \cite{BOVZKD18}) and mistrustful quantum cryptographic schemes (e.g.,\cite{LKBHTKGWZ13,PGK18,PG19}). These bounds imply near-perfect token
unforgeability allowing for experimentally quantified
error types and loss levels.   Unlike 
previous analyses (e.g., \cite{KLPGR22,LKBHTKGWZ13}), they do not assume that source qubit states belong to orthonormal bases. We characterized the deviation from BB84 states and used this in our security proof, going substantially beyond previous security analyses (e.g., \cite{TVUZ05,NFHM08,BBBGST11,LKBHTKGWZ13,LCCLWCLLSLZZCPZCP14,PJLCLTKD14,BOVZKD18,KLPGR22,SERGTBW23}) in mistrustful quantum cryptography. To our knowledge, no previous security analysis has allowed for general deviations from the set of states stipulated by an ideal protocol, measured these deviations experimentally, and based security bounds on these empirical data; without these results, claimed security bounds are not reliable. (See Methods and SI.)}
\adri{Our experiment is also the first demonstration of a quantitative quantum time advantage in relativistic cryptography, showing the enhanced cryptographic power of simultaneously considering quantum and relativistic physics.}

\textbf{Acknowledgement} This work was supported by the National Key Research and Development Program of China (2020YFA0309704), the Innovation Program for Quantum Science and Technology (2021ZD0300800, 2023ZD0300100), the National Natural Science Foundation of China (12374470), the Key R\&D Plan of Shandong Province (2021ZDPT01), the Key Area R\&D Program of Guangdong Province (2020B0303010001), Shandong Provincial Natural Science Foundation (ZR2021LLZ005), Y.-F.J and Q.Z acknowledge support from the Taishan Scholars Program. A.K. and D.P.-G. acknowledge support from the UK Quantum Communications Hub grant no. EP/T001011/1. G.C. was supported by a studentship from the Engineering and Physical Sciences Research Council. A.K. was supported in part by Perimeter Institute for Theoretical Physics. Research at Perimeter Institute is supported by the Government of Canada through the Department of Innovation, Science and Economic Development and by the Province of Ontario through the Ministry of Research, Innovation and Science.

\clearpage

\begin{figure*}[!t]
	\centering
	\includegraphics[width=0.8\textwidth]{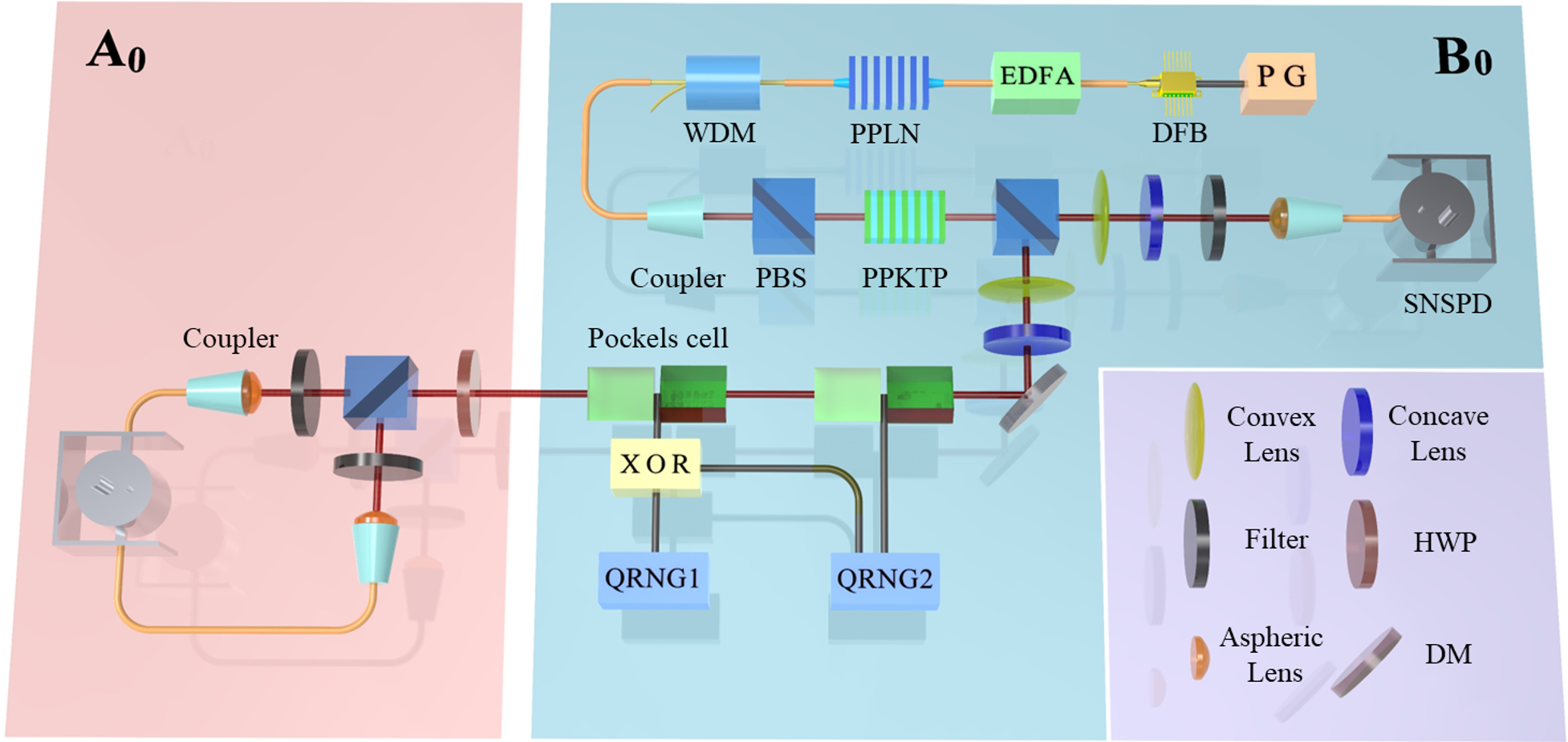}
	\caption{\dam{Diagram of the quantum phase of the quantum token generation using a} heralded single-photon source. A distributed feedback laser (DFB) with a central wavelength of 1560 nm is used as the pump, driven by a pulse generator (PG). The laser emits a pulse with a width of 5 ns and a repetition rate of 500 KHz. The pump is frequency-doubled with a periodically poled MgO-doped lithium niobate (PPLN) crystal and filtered by a wavelength division multiplex (WDM) to create the 780 nm pump. The photon pairs are generated through the spontaneous parametric down-conversion (SPDC) process, utilizing the Type-II periodically poled potassium titanyl phosphate (PPKTP) crystal. One photon is detected by the superconducting nanowire single-photon detector (SNSPD) to generate a trigger signal. The other is encoded with the basis $u_i$ and bit $t_i$ using the two Pockels cells respectively, which are driven by quantum random number generators (QRNGs). The XOR is an Exclusive OR operation implemented with electronic. A$_0$ \dam{measures} the received qubits using a half-wave plate (HWP) and a polarizing beam splitter (PBS), and the photons are detected by SNSPDs.
	} \label{Fig:concepts}
\end{figure*} 
\clearpage
\begin{figure*}[!t]
	\centering
	\includegraphics[width=1\textwidth]{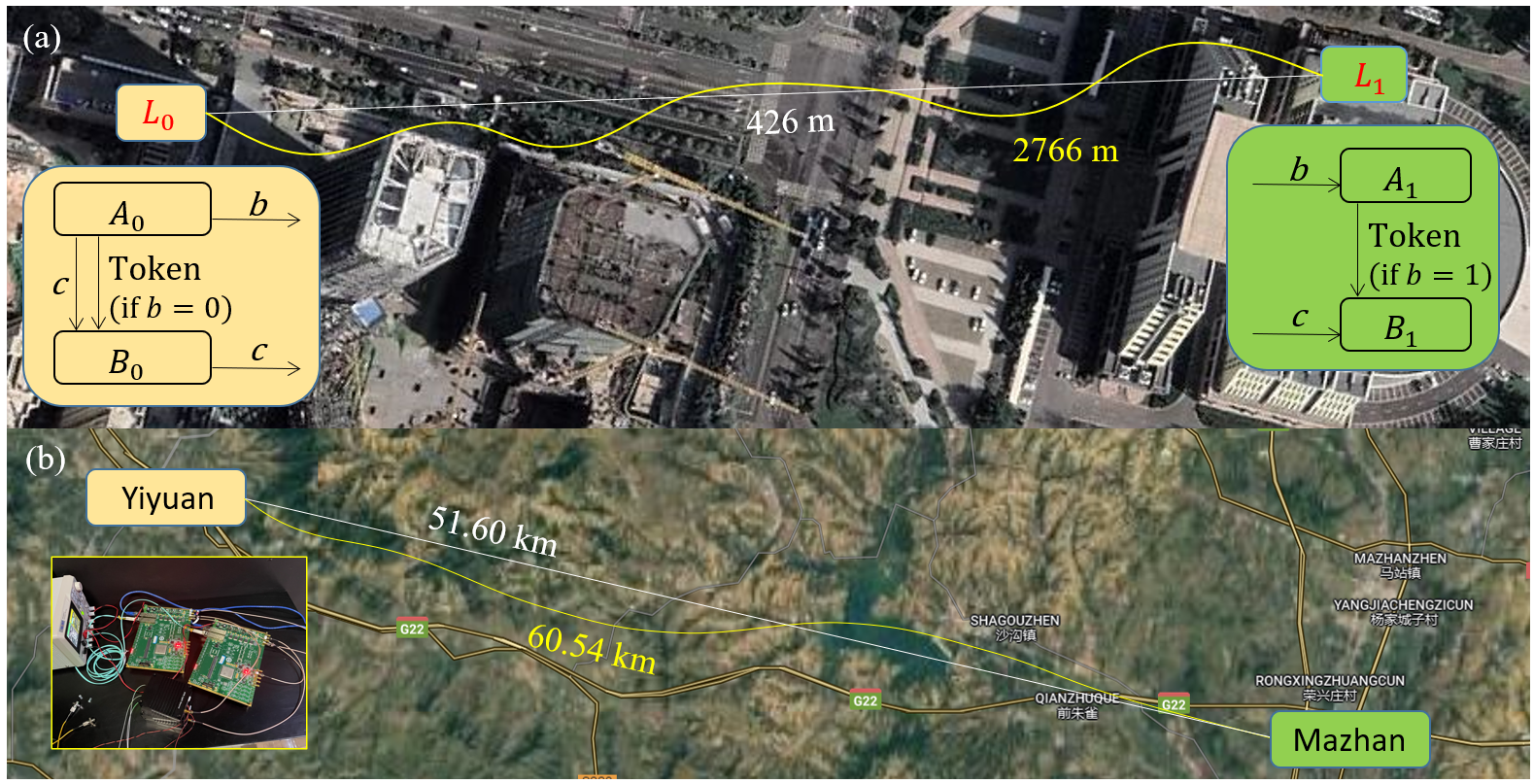}
	\caption{Field deployment of the S-token. (a) A satellite image shows \adr{the S-token setup} in the fibre optic network within the city of Jinan, Shandong Province, China, for demonstrating the quantum advantage. The fibre length is 2.766 km, whereas the corresponding direct free space distance is around 425 m. (b) A satellite image displays the S-token \adr{setup} with field-deployed fibre between Yiyuan ($36^{\circ} 10^{\prime} 50.4^{\prime \prime} N, 118^{\circ} 12^{\prime} 10^{\prime \prime} E$) and Mazhan ($36^{\circ} 0^{\prime} 19^{\prime \prime} N, 118^{\circ} 42^{\prime} 35^{\prime \prime} E$) in Shandong Province of China, for demonstrating the comparative advantage. The fibre length is 60.54 km, while the direct free space distance between them is about 51.60 km.
    }
	\label{fig:Field}
\end{figure*}
\clearpage

\begin{figure*}[!t]
\centering
\includegraphics[width=0.98\textwidth]{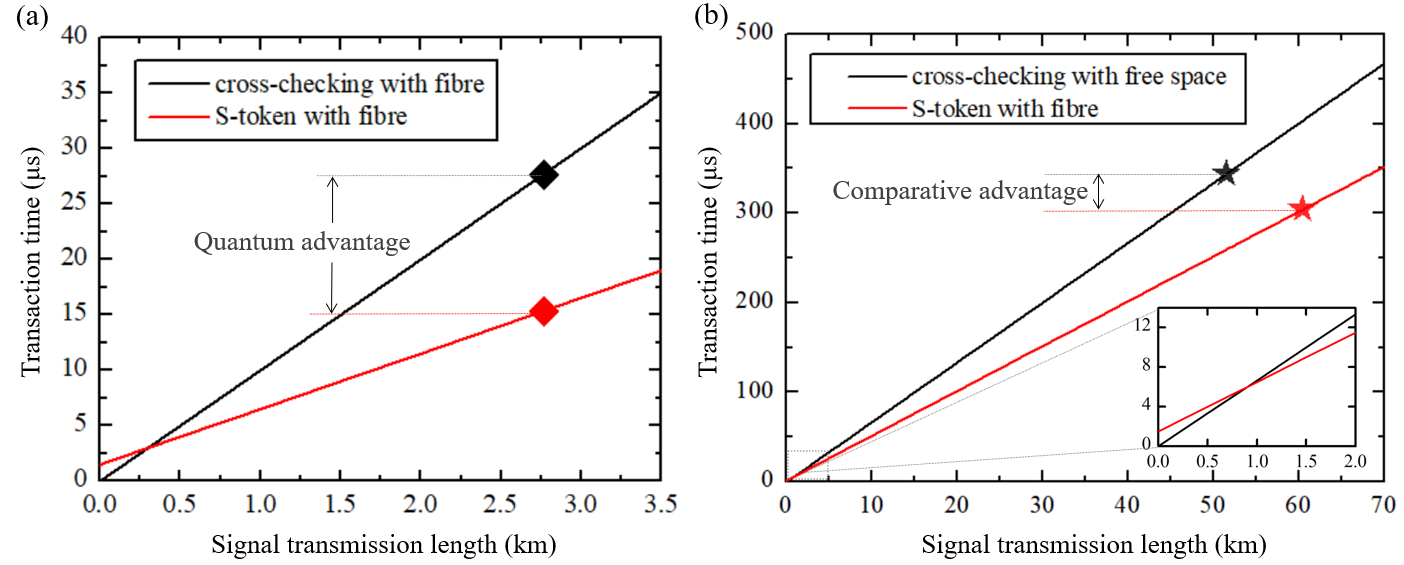}
\caption{The results of the time advantage. (a) The red and black lines are the transaction time based on the S-token and classical cross-checking through fibre. The red square is the measured transaction time of the S-token through 2.766 km fibre link, the black square is the ideal transaction time of classical cross-checking via the same fibre length. (b) The red and black lines are the transaction time based on the S-token and classical cross-checking through fibre and free space channel, respectively. The red star is the measured transaction time of the S-token through 60.54 km fibre link, the black star is the ideal transaction time of classical cross-checking through 51.60 km free space link. Insert: the comparative advantage is achieved for distances longer than 0.9 km. }
	\label{fig:result}
\end{figure*}

\clearpage
\end{sloppypar}

\section*{Methods}

\subsection{The practical quantum token scheme}

The \emph{practical scheme} deviates from the ideal scheme above by allowing the experimental imperfections described in Table 5, and \dam{making} the assumptions of Table 6, of Ref.~\cite{KLPGR22}. However, \dam{here we enhance the security analysis of Ref.~\cite{KLPGR22} by allowing Bob to prepare qubit states that do not belong to orthonormal bases, hence, we do not need to make assumption A of Ref.~\cite{KLPGR22}. We further improve the security analysis of Ref.~\cite{KLPGR22} by introducing the probability $P_\theta>0$ defined below. Moreover, Alice does not report any losses to Bob in our experimental implementation.} Thus, we do not need to make the assumptions C, D and F of Ref.~\cite{KLPGR22} \dam{either}. The  experimental imperfections \adr{ considered are defined by the parameters $\dam{E,} \gamma_\text{err},P_\text{noqub},\dam{P_{\theta},P_{\text{noqub},\theta}}\in(0,1)$, $\theta\in\Bigl(0,\frac{\pi}{4}\Bigr)$ and $\beta_{\text{PB}},\beta_{\text{PS}},\beta_\text{E}\in\Bigl(0,\frac{1}{2}\Bigr)$.   Here \adr{$E=\max_{t,u}\{E_{tu}\}$, where $E_{tu}$ is an upper bound on the probability that Alice obtains a wrong measurement outcome when she attempts to measure a quantum state $t$ }
in its preparation basis $u$;} $\gamma_\text{err}$ is the maximum error rate allowed by Bob for token validation as given by (\ref{allowederrors})\dam{;} $P_\text{noqub}$ is an upper bound on the probability that each quantum state transmitted by Bob has dimension greater than two (by comprising two or more qubits, for instance), which arises due to an imperfect single-photon source\dam{; $\theta$ is an uncertainty angle in the Bloch sphere;} \dam{$P_{\theta}$ is \adr{an upper bound on} the probability that a prepared quantum state has uncertainty angle greater than $\theta$ in the Bloch sphere; $P_{\text{noqub},\theta}$ is \adr{an upper bound on} the probability that a prepared quantum state has dimension greater than two or its uncertainty angle in the Bloch sphere is greater than $\theta$, given by
\begin{equation}
\label{Pnoqubtheta}
    P_{\text{noqub},\theta}=1-(1-P_{\text{noqub}})(1-P_{\theta});
\end{equation}}
and where $\beta_{\text{PB}}$, $\beta_{\text{PS}}$ and $\beta_\text{E}$ are upper bounds on the biases for the respective probabilities of basis preparation $u_k$, state preparation $t_k$ and the bit $z$. 

\dam{Although not needed in our experimental implementation, it is useful to mention that the scheme can be straightforwardly extended to allow Alice to report losses to Bob.  This requires the following extra steps in the quantum phase for the token generation~\cite{KLPGR22}. A$_0$ reports to B$_0$ the set $\Lambda$ of indices $k$ of quantum states $\lvert \psi_k\rangle$ sent by B$_0$ that produce unsuccessful measurements. Let $n=\lvert \Lambda\rvert$. B$_0$ does not abort if and only if $n\geq \gamma_\text{det} N$, where the threshold $\gamma_\text{det}\in(0,1)$ is agreed in advanced by Alice and Bob. The scheme continues as above but with the $n-$bit strings that restrict $\textbf{x},\textbf{x}',\textbf{x}^0,\textbf{x}^1,\textbf{t},\textbf{u}$ to entries with indices $k\in\Lambda$.} \dam{$P_{\text{det}}\in(0,1)$ is the probability that a quantum state transmitted by Bob is reported by Alice as being successfully measured. \adr{The strategy used by A$_0$ to report (un)successful measurements} must be chosen carefully to counter multi-photon attacks by B$_0$ \cite{KLPGR22,BCDKPG21}.} \dam{The analysis for our implementation reduces straightforwardly to the case $P_{\text{det}}=\gamma_\text{det}=1$ and $n=N$.}

\subsection{Security definitions}

A token scheme using $N$ transmitted quantum states is 
$\epsilon_{\text{cor}}-$correct \cite{KLPGR22} if the probability that Bob does not accept Alice's token as valid when Alice and Bob follow the scheme honestly is not greater than $\epsilon_{\text{cor}}$, for any $b\in\{0,1\}$;
$\epsilon_{\text{priv}}-$private if the probability that Bob guesses Alice's bit $b$ before she presents her token is not greater than $\frac{1}{2}+\epsilon_{\text{priv}}$, if Alice follows the scheme honestly
and chooses $b\in\{0,1\}$ randomly from a uniform distribution;
$\epsilon_{\text{unf}}-$unforgeable, if the probability that Bob accepts Alice's tokens as valid at the two presentation locations is not greater than $\epsilon_{\text{unf}}$, if Bob follows the scheme honestly\dam{;
\emph{$\epsilon_\text{rob}-$robust} if the probability that Bob aborts when Alice and Bob follow the token scheme honestly is not greater than $\epsilon_\text{rob}$, for any $b\in\{0, 1\}$.} 
It is correct, unforgeable \dam{and robust} if the respective $\epsilon$-parameters decrease exponentially with $N$, and private if 
$\epsilon_{\text{priv}}$ can be made arbitrarily small by increasing security parameters.

\subsection{Quantum \dam{and comparative} advantage}
\label{sec:adv}

To compare our quantum scheme to classical cross-checking schemes we provide the following definitions. 
Let $\Delta T_\text{tran,C}$ and $\Delta T_\text{tran,CF}$ be the transaction times of a classical cross-checking scheme when it uses the same classical communication channel as our quantum scheme, and when it uses an ideal free-space communication channel at light-speed, respectively. The transaction time $\Delta T_\text{tran}$ of our quantum token scheme is defined by (\ref{tran}). We say our quantum scheme has \emph{quantum advantage} if
\begin{equation}
\label{qadv}
    QA\equiv \Delta T_\text{tran,C}- \Delta T_\text{tran}>0,
\end{equation}
and that it has \emph{comparative advantage} if
\begin{equation}
\label{cadv}
    CA\equiv \Delta T_\text{tran,CF}-\Delta T_\text{tran}>0.
\end{equation}
We note that $\Delta T_\text{tran,C}\geq \Delta T_\text{tran,CF}$. Thus, $QA\geq CA$ and comparative advantage ($CA>0$) implies quantum advantage ($QA>0$).

We define $\Delta T_\text{tran,C}$ and $\Delta T_\text{tran,CF}$ precisely by considering the simplest type of classical cross-checking scheme implemented by Alice and Bob between $L_0$ and $L_1$, under assumptions that minimize transaction time. As in the quantum scheme, Alice and Bob agree in advance on a spacetime reference frame F.   They agree that the token may be presented by Alice in one of two spacetime presentation regions $R_0$ and $R_1$, where $R_i$ comprises location $L_i$ with a time interval $\delta T$ beginning at a time $T$ given by (\ref{time}) in the frame F, and where the value of $T_\text{bit}$ is only indicated by Alice to Bob during the protocol (in step 3). We first assume that the communications between $L_0$ and $L_1$ by Alice and Bob use the same classical communication channel as in our quantum scheme. Step 1 is performed arbitrarily in advance of the following steps. Below, $\textbf{y}$ is a $n-$bit string, and $b,r_0,r_1$ are bits.

\begin{enumerate}

\item B$_0$ gives a classical password $\textbf{y}$ to A$_0$ at $L_0$, keeps a copy of $\textbf{y}$ and
sends a copy to B$_1$; A$_0$ keeps a copy of $\textbf{y}$ and sends a copy to A$_1$.

\item A$_0$ obtains $b$, indicating that she wishes to present the token at $L_b$. She keeps a copy of $b$ and sends a copy to A$_1$ at the time $T_\text{begin,C}$.

\item At the time $T_\text{bit}$, A$_0$ indicates to B$_0$ that the token will be presented at time $T$ given by (\ref{time}) at either $L_0$ or $L_1$.

\item A$_b$ gives $\textbf{y}$ to B$_b$ at $L_b$ within the time interval $[T,T+\delta T]$.

\item At the time $T+\delta T$, B$_i$ sends $r_i$ to B$_{i\oplus 1}$, where $r_i=1$ ($r_i=0$) indicates that B$_i$ received (did not receive) a token at $L_i$ within the time interval $[T,T+\delta T]$.

\item At the time $T_\text{C,end}=T+\delta T+\Delta T_\text{comm}$, B$_i$ validates a token received at $L_i$ within $[T,T+\delta T]$ if it is equal to $\textbf{y}$ and if $r_{i\oplus 1}=0$.
\end{enumerate}

We assume that the communication time between $L_0$ and $L_1$ is the same for Alice and Bob. We define the \emph{transaction time} of the cross-checking scheme by
\begin{equation}
    \label{tranC0}
  \Delta T_\text{tran,C}= T_\text{end,C}-T_\text{begin,C}.  
\end{equation}
Assuming that all local classical processing and local communications at $L_0$ and $L_1$ can be made instantaneous, that $\delta T=0$ and that $T_\text{bit}=T_\text{begin,C}$, we obtain
\begin{equation}
    \label{tranC}
  \Delta T_\text{tran,C}= 2 \Delta T_\text{comm}.  
\end{equation}
If now we assume that Alice and Bob use free-space channels at light-speed for their communications between $L_0$ and $L_1$, then 
\begin{equation}
    \label{tranCF}
  \Delta T_\text{tran,CF}= \frac{2D}{c},  
\end{equation}
where $D$ is the distance between $L_0$ and $L_1$, and $c$ is the speed of light through \dam{a} vacuum.

\subsection{\dam{Experimental quantum and comparative advantage}}

Our experiment demonstrates quantum advantage (\ref{qadv}) and comparative advantage (\ref{cadv}). The transaction time $\Delta T_\text{tran}$ in our scheme is 
\begin{equation}
\label{eqtran}
    \Delta T_\text{tran}=\Delta T_\text{proc}+\Delta T_\text{comm},
\end{equation}
where
\begin{equation}
\label{eqcomm}
    \Delta T_\text{comm}=\frac{L_\text{fibre}}{c_\text{fibre}}
\end{equation}
is the communication time between $L_0$ and $L_1$\dami{, where $L_\text{fibre}$ is the length of the optical fibre between $L_0$ and $L_1$, and where $c_\text{fibre}\approx 2 \times 10^{8}~\text{ms}^{-1}$ is the speed of light through the fibre}. Thus, from (\ref{qadv}), (\ref{tranC}), (\ref{eqtran}) and (\ref{eqcomm}), we obtain
\begin{equation}
\label{qadvc}
    QA=\frac{L_\text{fibre}}{c_\text{fibre}}-\Delta T_\text{proc}.
\end{equation}
Thus, quantum advantage ($QA>0$) can be demonstrated for lengths
\begin{equation}
  L_\text{fibre}\geq \Delta T_\text{proc}c_\text{fibre} \approx 0.3 \text{~km},
\end{equation}
as \dam{mentioned} in the main text\dami{, since $\Delta T_\text{proc} \approx 1.5~\mu\text{s}$}. We observed the quantum advantage in the fibre optic network within the city of Jinan, Shandong Province, as depicted in Fig.~2(a). The fibre length connecting these locations is $L_\text{fibre}=2.766$ km, while the direct free-space distance is $D\approx 426$ m. The token is tested 20 times with randomness in choosing the presentation location $b$, resulting in an average transaction time of $\dam{\Delta} T_\text{tran}=15.34 \pm 0.01 ~\mu \text{s}$ and an average error rate of $6.02\%$. Notably, the quantum advantage in this scenario amounts to $QA = 12.32\pm 0.01 ~\mu \text{s}$, shown in Fig.~3(a).  

Furthermore, from (\ref{cadv}) and (\ref{tranCF}) -- (\ref{eqcomm}), we have
\begin{equation}
    \label{cadvc}
    CA=\frac{2D}{c}-\Delta T_\text{proc}-\frac{L_\text{fibre}}{c_\text{fibre}}.
\end{equation}
In general, $L_\text{fibre}\geq D$, hence, comparative advantage ($CA>0$) requires
\begin{equation}
    D>\frac{\Delta T_\text{proc}}{\frac{2}{c}-\frac{1}{c_\text{fibre}}}\approx 0.9 {\rm~ km},
\end{equation}
as \dam{mentioned} in the main text. We realize the comparative advantage in field-deployable fibre between Yiyuan ($36^{\circ}  10^{\prime}  50.4^{\prime \prime} N, 118^{\circ}  12^{\prime}  10^{\prime \prime} E$) and Mazhan ($36^{\circ}  0^{\prime}  19^{\prime \prime} N, 118^{\circ}  42^{\prime}  35^{\prime \prime} E$) in Shandong Province of China, which are connected by an optical fibre channel of length $L_\text{fibre}= 60.54$ km and separated by a physical distance of $D=51.60$ km, as shown in Fig.~2(b). We ran the transactions 20 times. The error rates of all trials are below $\gamma_{err}=9.4\%$, with an average error rate of $6.00\%$. Additionally, the average transaction time is $\dam{\Delta} T_\text{tran}=304.20\pm 0.01 ~\mu \text{s}$. Thus, we obtained the comparative advantage $CA= 39.80\pm 0.011 ~\mu \text{s}$ shown in Fig.~3(b). \adr{See SI for further details.} 

\subsection{Security analysis}

\dam{The security analysis of our experimental demonstration is based on the following lemmas and theorem, which apply to the general case of the practical scheme discussed above in which Alice reports losses to Bob.} \dam{The analysis for our implementation, in which Alice does not report losses, reduces straightforwardly to the case $P_\text{det}=\gamma_\text{det}=1$ and $n=N$.}

\begin{lemma}
\label{privacylemma}
We assume that the quantum token scheme is perfectly protected against arbitrary multi-photon attacks. We also assume that Bob does not obtain any information about Alice's measurement basis labelled by the bit $z$ via side-channel attacks, clock synchronization attacks or by any other means. We assume that Alice chooses the bit $b$ denoting her chosen presentation location $L_b$ randomly and securely. Then, the quantum scheme is $\epsilon_{\text{priv}}-$private with
\begin{equation}
	\epsilon_{\text{priv}}=\beta_\text{E}\,.
\end{equation}
\end{lemma}
The proof of lemma~\ref{privacylemma} follows straightforwardly from the proof of lemma 4 in Ref.~\cite{KLPGR22}.

We note that our implementation is perfectly protected against arbitrary multi-photon attacks because Alice does not report any losses to Bob~\cite{KLPGR22,BCDKPG21}. \dami{The problem of side-channel attacks is very general in quantum cryptography. However, effective counter-measures are available \dami{(e.g.~\cite{BCDKPG21})}, although we did not implement them in our experimental demonstration.}

\dami{In our implementation, the time synchronization between Alice’s and Bob’s electronic boards was accomplished through an optical fibre channel. Ideally, \adr{given Bob and Alice's mistrust}, each party should have an independent trustworthy method for synchronizing their agents' clocks. Potential time synchronization attacks would not affect unforgeability, but could compromise correctness and user privacy. If Bob cannot rely on his agents' clocks being synchronized then he cannot be sure Alice is presenting the token at a valid space-time region. If Alice cannot rely on her agents' clocks being synchronized then she cannot be sure she will present the token at a valid space-time region, and she might give away information (about her chosen token presentation region) sooner than she intended, or receive resources (in exchange for the token) in a different region than she intended. We emphasize that this issue arises in any quantum token scheme and is not specific to the S-tokens implemented here, and can be resolved by using independent reliable secure timing and frequency networks.
The problem of secure time synchronization is quite general in relativistic quantum cryptography. Although we did not solve this problem in our implementation, countermeasures are possible, for example synchronising the clocks in secure laboratories or
via a secure global position system.}

In our experiment, we obtained $\beta_\text{E}=10^{-5}$. Thus, our implementation is $10^{-5}-$private, given the assumptions of Lemma~\ref{privacylemma}. The privacy level of our scheme could be made arbitrarily good, since one can decrease $\beta_\text{E}$ exponentially in $r$, by computing $z$ as the sum modulo two of $r$ close to random bits~\cite{Pilinguplemma}.  This can be done at any time before the scheme.

\begin{lemma}
\label{robsutnesslemma}
If
\begin{equation}
	0<\gamma_\text{det} <P_{\text{det}}\,,
\end{equation}
then the quantum token scheme is $\epsilon_{\text{rob}}-$robust with
\begin{equation}
	\epsilon_{\text{rob}}=\left(\frac{P_{\text{det}}}{\gamma_\text{det}}\right)^{N\gamma_\text{det}}\left(\frac{1-P_{\text{det}}}{1-\gamma_\text{det}}\right)^{N(1-\gamma_\text{det})}\,.
\end{equation}
\end{lemma}

\dam{Our implementation is perfectly robust as Alice does not report losses to Bob, hence, Bob aborts with zero probability.}

\begin{lemma}
\label{correctnesslemma}
If
\begin{eqnarray}
	0&<&E<\gamma_\text{err}\,,\nonumber\\
	0&<&\nu_\text{cor}<\frac{P_{\text{det}}(1-2\beta_\text{PB})}{2}\,,\end{eqnarray}
\dam{for some $\nu_\text{cor}\in(0,1)$,} then the quantum scheme is $\epsilon_{\text{cor}}-$correct with
\begin{eqnarray}
\label{correctness}
	\epsilon_{\text{cor}}&=&\left(\frac{P_{\text{det}}(1-2\beta_\text{PB})}{2\nu_\text{cor}}\right)^{N\nu_\text{cor}}\left(\frac{2-P_{\text{det}}(1-2\beta_\text{PB})}{2-2\nu_\text{cor}}\right)^{N(1-\nu_\text{cor})}\nonumber\\
 &&\qquad+\left(\frac{E}{\gamma_\text{err}}\right)^{N\nu_\text{cor}\gamma_\text{err}}\left(\frac{1-E}{1-\gamma_\text{err}}\right)^{N\nu_\text{cor}(1-\gamma_\text{err})}.
\end{eqnarray}
\end{lemma}

\dami{
Our unforgeability proof uses the maximum confidence measurement of the following quantum state discrimination task.

\begin{definition}
\label{Pbound}
    Consider the following quantum state discrimination problem. For $k\in\Omega_\text{qub}$, we define $\rho_1^k=\rho_{00}^k$, $\rho_2^k=\rho_{01}^k$, $\rho_3^k=\rho_{10}^k$, $\rho_4^k=\rho_{11}^k$, $q_1^k=P^k_\text{PS}(0)P^k_\text{PB}(0)$, $q_2^k=P^k_\text{PS}(0)P^k_\text{PB}(1)$, $q_3^k=P^k_\text{PS}(1)P^k_\text{PB}(0)$, $q_4^k=P^k_\text{PS}(1)P^k_\text{PB}(1)$, and 
    \begin{equation}
    \label{BB1main}
r^k_i=\frac{q^k_i+q^k_{i+1}}{2} \, , \qquad\chi^k_i =\frac{q^k_i\rho^k_i+q^k_{i+1}\rho^k_{i+1}}{q^k_i+q^k_{i+1}} \, , \qquad \rho^k=\sum_{i=1}^4 r^k_i \chi^k_i\, ,
    \end{equation}
for all $i\in[4]$. Let $P_{\text{MC}}(\chi^k_j)$ be the maximum confidence measurement that the received state was $\chi^k_j$ when Alice \adr{ is distinguishing states from the 
ensemble $\{ \chi_j^k , r_i^k \}$ and her outcome is $j\in [4]$ \cite{CABGJ06}. 
This} maximum is taken over all positive operators $Q$ acting on a two dimensional Hilbert space. That is, we have
\begin{equation}
\label{BB2main}
 P_{\text{MC}}(\chi^k_j)=\max_{Q\geq 0}\frac{r^k_jTr[Q\chi^k_j]}{Tr[Q\rho^k]} \, .   
\end{equation}
\end{definition}

}

\begin{theorem}\label{unforgeabilitytheorem}
Suppose that the following constraints hold:
\adr{
\begin{eqnarray}\label{conditionsmain}
\max_{j\in[4],k\in \Omega_\text{qub}} \dami{2} P_{\text{MC}}(\chi^k_j) & < & 1 \, , \nonumber \\
 N\gamma_\text{det}&\leq &  n \leq N\, ,\nonumber\\
0&<&P_{\text{noqub},\theta}<\nu_\text{unf}<\gamma_\text{det}\Bigl(1-\frac{\gamma_\text{err}}{1-P_\text{bound}}\Bigr)\, ,
\end{eqnarray}
for predetermined $\gamma_\text{det}\in(0,1]$ and $\gamma_\text{err}\in[0,1)$ and for some $\nu_\text{unf}\in(0,1)$, where $n=\lvert \Lambda\rvert$, and where $P_\text{bound}$ satisfies
\begin{equation}
\label{BB3main}
\max_{j\in[4],k\in \Omega_\text{qub}} \dami{2} P_{\text{MC}}(\chi^k_j)\leq P_{\text{bound}} < 1\, .
\end{equation}
In the case} that losses are not reported we take $\gamma_\text{det}=1$ and $n=N$.
 Then the quantum token schemes $\mathcal{QT}_1$ and $\mathcal{QT}_2$ are $\epsilon_{\text{unf}}-$unforgeable with
\begin{eqnarray}
\label{unforgeabilitymain}
\epsilon_{\text{unf}} &=& \sum_{l=0}^{\lfloor N(1-\nu_\text{unf}) \rfloor}\left(\begin{smallmatrix}N\\ l\end{smallmatrix}\right)(1-P_{\text{noqub},\theta})^{l}(P_{\text{noqub},\theta})^{N-l}\, \nonumber\\
&&\quad + \sum_{l=0}^{\lfloor n\gamma_\text{err} \rfloor}\left(\begin{smallmatrix}n- \lfloor N\nu_\text{unf} \rfloor\\ l\end{smallmatrix}\right)(1-P_{\text{bound}})^{l}(P_{\text{bound}})^{n- \lfloor N\nu_\text{unf} \rfloor-l} \, .
\end{eqnarray}
\end{theorem}

\adr{Note that the conditions (\ref{conditionsmain}) \dami{and (\ref{BB3main}})} imply that the bound (\ref{unforgeabilitymain}) decreases exponentially with $N$.

\adr{Theorem \ref{unforgeabilitytheorem} is improved in two main ways from previous work \cite{KLPGR22} and earlier mistrustful cryptography security analyses: 1) it allows Bob's prepared states to deviate arbitrarily from the intended BB84 states up to an angle $\theta$ in the Bloch sphere without restricting the prepared states to form qubit orthonormal bases; and 2) it replaces $P_{\text{noqub}}$ by $P_{\text{noqub},\theta}$. That is, in the security analysis of Ref.~\cite{KLPGR22}, $\theta$ was considered an upper bound on the uncertainty angle in the Bloch sphere for state preparation. But, here we relax this assumption by allowing the uncertainty angle to be greater than $\theta$ with a probability $P_\theta$. The probability $P_{\text{noqub},\theta}$ considers this via equation (\ref{Pnoqubtheta}).
Lemmas \ref{robsutnesslemma} and \ref{correctnesslemma} also improve on the corresponding lemmas 2 and 3 and theorem 1 of Ref.~\cite{KLPGR22}, by using tighter Chernoff bounds.} 

In our implementation, Alice does not report any losses to Bob. Thus, our implementation is perfectly robust, as Bob aborts with zero probability, hence, we can ignore lemma \ref{robsutnesslemma} in the security analysis for our implementation. We can also set $P_\text{det}=\gamma_\text{det}=1$ in lemma \ref{correctnesslemma} and theorem \ref{unforgeabilitytheorem} for our experimental demonstration. The proofs for lemmas~\ref{robsutnesslemma} and~\ref{correctnesslemma} and for theorem~\ref{unforgeabilitytheorem} are given in the Supplementary Information.

We obtained $N=10,048$ in approximately five minutes. We obtained the following experimental parameters in our implementation with seven standard deviations: $E=0.062550$, $P_\text{noqub} = 
\dami{4.9\times 10^{-5}}$, $\beta_\text{PB}=
\dami{0.001360}$, $\beta_\text{PS}=
\dami{0.001120}$. Details for the estimation of these experimental parameters are given in the Supplementary Information.

We set $\gamma_\text{err}=0.094$ and $\nu_\text{cor}=0.457643134$. Using the previous experimental parameters, the required constraints in lemma~\ref{correctnesslemma} are satisfied and we obtain \dami{that the two terms in (\ref{correctness}) are respectively $\epsilon_\text{cor}^1=2.05304 \times 10^{-15}$ and $\epsilon_\text{cor}^2=1.89154 \times 10^{-15}$, giving $\epsilon_\text{cor}=\epsilon_\text{cor}^1+\epsilon_\text{cor}^2=3.94458 \times 10^{-15}$.} 

We measured $\theta=
\dami{5.115515}^\circ$ with $P_\theta=0.027$, guaranteed correct unless with an error probability smaller than $1.3\times 10^{-12}<2.6\times 10^{-12}$. This is consistent with our seven standard-deviation measurements of other experimental parameters, guaranteeing the accuracy of our measurements unless with error probabilities smaller than $2.6\times 10^{-12}$, as discussed below. Thus, from (\ref{Pnoqubtheta}), we obtained $P_{\text{noqub},\theta}=0.027047677$. With the obtained values of $\theta, \beta_\text{PB}$ and $\beta_\text{PS}$, we obtained a numerical bound $P_\text{bound}=\dami{0.884130}$ satisfying (\ref{BB3main}) using Mathematica software. Taking $\nu_\text{unf}=P_{\text{noqub},\theta}+\dami{0.0105=0.037547677}$, we obtained respective values for the two terms of \dami{$\epsilon_\text{unf}$ in (\ref{unforgeabilitymain}) of $\epsilon_\text{unf}^1=
3.72375\times 10^{-10}$ and $\epsilon_\text{unf}^2=5.11874\times 10^{-9}$, giving $\epsilon_\text{unf}=\epsilon_\text{unf}^1+\epsilon_\text{unf}^2=5.49112\times 10^{-9}$.  }

We note that $\epsilon_\text{cor}$ depends on several variables $V_1,V_2,\ldots,V_{K_\text{cor}}$, and $\epsilon_\text{unf}$ depends on variables $W_1,W_2,\ldots,W_{K_\text{unf}}$, whose values are estimated within a confidence interval of seven standard deviations. This means that the estimated value for each of these variables $V_j$ or $W_j$ could be wrong with a probability $P_\text{wrong}=2.6\times 10^{-12}$, which is the value corresponding for seven standard deviations. We assume that if any of the variables $V_1,V_2,\ldots,V_{K_\text{cor}}$ does not correspond to the estimated values then the token scheme is not correct, while if every variable $V_1,V_2,\ldots,V_{K_\text{cor}}$ corresponds to its estimated value then the token scheme is not correct with probability $\epsilon_\text{cor}$. Thus, assuming that these variables are independent, the probability that the token scheme is not correct is
\begin{equation}
\epsilon_\text{cor}'=1-(1-P_\text{wrong})^{K_\text{cor}}+\epsilon_\text{cor}(1-P_\text{wrong})^{K_\text{cor}}.
\end{equation}
The variables $V_j$ are $\beta_{PB}$, $E_{00}$, $E_{01}$, $E_{10}$, $E_{11}$, $T_\text{exp}$ and $f_\text{sys}$, where $f_\text{sys}$ is the frequency of quantum state generation by Bob and $T_\text{exp}$ is the time taken to generate the $N$ quantum states transmitted to Alice. Thus, we have $K_\text{cor}=7$. Therefore, with $P_\text{wrong}=2.6\times 10^{-12}$ and $\epsilon_\text{cor}=\dami{3.94458} \times 10^{-15}$, we obtain
\begin{equation}
    \epsilon_\text{cor}'=\dami{2.1}\times 10^{-11}.
\end{equation}
Thus, our implementation is proved \dami{$2.1\times 10^{-11}-$}correct.

Similarly, we assume that if any of the variables $W_1,\ldots,W_{K_\text{unf}}$ does not correspond to the estimated values then the token scheme is not unforgeable, while if every variable $W_1,\ldots,W_{K_\text{unf}}$ corresponds to its estimated value then the token scheme is not unforgeable with probability $\epsilon_\text{unf}$. Thus, assuming that these variables are independent, the probability that the token scheme is not unforgeable is
\begin{equation}
\epsilon_\text{unf}'=1-(1-P_\text{wrong})^{K_\text{unf}}+\epsilon_\text{unf}(1-P_\text{wrong})^{K_\text{unf}}.
\end{equation}
The variables $W_j$ are $\beta_{PB}, \beta_{PS}, P_\theta, P_\text{noqub}, T_\text{exp}$ and $ f_\text{sys}$. Thus, we have $K_\text{unf}=6$. Therefore, with $P_\text{wrong}=2.6\times 10^{-12}$ and \dami{
$\epsilon_\text{unf}=\dami{5.49112\times 10^{-9}}$, we obtain
\begin{equation}
\epsilon_\text{unf}'=\dami{5.52\times 10^{-9}} \, .
    \end{equation}
Thus, our implementation is proved \dami{$5.52 \times 10^{-9}-$}unforgeable.
}

\setcounter{equation}{0}
\setcounter{table}{0}
\setcounter{definition}{0}
\setcounter{figure}{0}
\section*{Supplementary information}

\section{Summary}

\label{sec1}
We give a brief summary of the content discussed in this supplementary information, emphasizing the most important points, and discussing how we believe they can be helpful beyond our implementation in the broader \adr{fields of experimental and theoretical mistrustful quantum cryptography}.

In section~\ref{sec:exp}, we provide a rigorous analysis to estimate various experimental parameters playing a role in our security analysis, for example, upper bounds on the biases in state preparation, basis choice, and selection of the encoding bit $z$, given by $\beta_\text{PS}$, $\beta_\text{PB}$ and $\beta_\text{E}$, respectively, as well as upper bounds on the error rate $E$ and the probability $P_\text{noqub}$ that Bob's heralding pulse has more than one photon.

In particular, it is crucial in mistrustful quantum cryptography implementations to guarantee that $P_\text{noqub}$ is suitably small in order for Bob to be sufficiently protected against photon number splitting attacks~\cite{HIGM95,BLMS00} by Alice (who receives the quantum states from Bob). Thus, we believe our analysis here can be helpful quite broadly in experimental mistrustful quantum cryptography. Our analysis is based on the assumption that the photon source has Poissonian statistics, which is well \adr{supported} in the literature~\cite{avenhaus2008photon,schneeloch2019introduction}.

In section~\ref{sec:theta}, we provide an experimental and theoretical analysis to derive an upper bound $\theta$ on the uncertainty angle on the Bloch sphere for Bob's state preparations, and an upper bound $P_\theta$ on the probability that the bound does not hold.  We believe this analysis can be useful quite broadly in quantum cryptography implementations, as we consider the imperfections of various photonic devices that are commonly used in quantum cryptography, like half wave plates, polarizing beam splitters and rotation mounts. 

It is crucial to characterize the values of $\theta$ and $P_\theta$ in the security analysis of quantum cryptographic protocols. In general, we expect that the security of realistic quantum cryptography protocols will decrease if the prepared states deviate from the intended states in an ideal protocol. Thus, a realistic security analysis must take into account such deviations, as characterized by the parameters $\theta$ and $P_\theta$ in our analysis, hence, must also estimate the value of these deviations using experimental data. 
\adr{
Our work here goes substantially beyond previous security analyses (e.g. \dami{\cite{TVUZ05,NFHM08,BBBGST11,LKBHTKGWZ13,LCCLWCLLSLZZCPZCP14,PJLCLTKD14,BOVZKD18,KLPGR22,SERGTBW23}}) in mistrustful quantum cryptography.
As far as we are aware, no previous security analysis has allowed for general deviations from BB84 states \dami{(or \adr{any alternative set of states stipulated by an ideal protocol})}, measured these deviations experimentally, and based security bounds on these empirical data.   Without these results, claimed security bounds are not reliable, and indeed experimental protocols may be completely insecure. }

\dami{In section~\ref{sec:times}, 
we discuss the time sequence of our implementation, and the time advantages achieved by our experiment.} 

In section~\ref{sec:related}, we discuss our quantum token scheme in the context of other related works. We also discuss ways in which our schemes can be straightforwardly extended.

In section~\ref{sec:security}, we proved the security of our quantum token scheme implementation. Our unforgeability proof holds even if Alice is required to report losses and applies for arbitrarily powerful dishonest Alice who may detect all quantum states received from Bob and choose to report an arbitrary subset of states as lost.

As discussed in section~\ref{sec:multiple}, our quantum token schemes and the security analysis extend straightforwardly to an arbitrary number of presentation spacetime regions. As we discuss, our experimental setup would guarantee a high degree of security in \adr{realistic} multi-node scenarios \adr{involving global or national networks}.

Our results in section~\ref{sec:security} are technical and apply more broadly to the area of mistrustful quantum cryptography. In many ideal quantum cryptographic protocols, including  relativistic quantum bit commitment \cite{LKBHTKGWZ13} and quantum money tokens \cite{KLPGR22}, Bob sends Alice random states from the BB84 \cite{BB84} or another given set. In practice, the states are prepared with misalignment, not uniformly distributed, are mixed, and include some multi-photon states.  To cheat, Alice must produce statistically plausible results for measurements in both BB84 bases, allowing for a given error level. We present a general security analysis based on maximum confidence quantum measurements \cite{CABGJ06} that strongly bounds Alice's probability of winning games of this type with arbitrary quantum strategies, and discuss applications to specific protocols.

Our main technical results are twofold. First, we consider a broad class of quantum tasks in which Alice receives quantum states from a given set in $N$ independent rounds and is required to obtain particular classical information about the prepared states for all rounds, with the possibility of failing in no more than $n$ rounds, for a given $0\leq n\leq N$. Effectively, Alice is playing a multi-round game which she wins if she succeeds in a sufficiently high proportion of the rounds.   

We show that if Alice's success probability  in the $k$th round is upper bounded by $P_\text{bound}^k$, conditioned on any quantum inputs $\rho_j$ and classical outputs $x_j$ for rounds $j\neq k$ and on any extra measurement outcome $o_{\text{extra}}$ obtained by Alice, for all $k\in[N]$, then Alice's success probability $P_\text{win}(n,N\vert o_\text{extra})$ in the task conditioned on the extra outcome $o_{\text{extra}}$ is upper bounded by the probability $P_{\text{bound}}^{\text{coins}}(n,N)$ of having no more than $n$ errors in $N$ independent coin tosses with success
probabilities $P_\text{bound}^1, P_\text{bound}^2, \ldots, P_\text{bound}^N$. Thus, we have
\begin{equation}
\label{B1}
   P_{\text{win}}(n,N\vert o_\text{extra})\leq P_{\text{bound}}^{\text{coins}}(n,N)  \leq \sum_{l=0}^{n}\left(\begin{smallmatrix}N\\ l\end{smallmatrix}\right)(1-P_{\text{bound}})^{l}(P_{\text{bound}})^{N-l},
\end{equation}
where $P_\text{bound}^k\leq P_{\text{bound}}<1$ for all $k\in[N]$. This further implies that we can upper bound the right hand side by a Chernoff bound decreasing exponentially with $N$ if $n<N(1-P_{\text{bound}})$.

This result is quite useful for a great variety of quantum cryptography protocols in which Alice's cheating probability reduces to wining the described task. In this case the security proof can be reduced to finding the upper bound $P_\text{bound}^k$ for the round $k$
 conditioned on any quantum inputs $\rho_j$ and classical outputs $x_j$ for rounds $j\neq k$ and on any extra measurement outcome $o_\text{extra}$ obtained by Alice, for all $k\in[N]$. Crucially, we note that the result applies to arbitrary quantum strategies by Alice, including arbitrary joint quantum measurements on the quantum states received in all $N$ rounds.
 
 Examples where this result is useful include relativistic quantum bit commitment protocols (e.g., \cite{LKBHTKGWZ13}), quantum money schemes (e.g., \cite{BOVZKD18}), quantum S-money token schemes \cite{KLPGR22}.  It can also be used for security proofs in other mistrustful quantum cryptography protocols, for example, quantum spacetime-constrained oblivious transfer protocols \cite{PGK18,PG19}. 

Second, we deduce the bound $ P_\text{bound}^k$ for an important and cryptographically relevant subset of the quantum tasks described above, in which Alice's task in each round can be shown to be equivalent to a quantum state discrimination task. In this case, we show that Alice's probability to win the task in round $k$, conditioned on any quantum input states $\rho_i$ and classical outputs $x_i$ for rounds $i\neq k$ and on any extra measurement outcomes $o_\text{extra}$, is upper bounded by her maximum confidence quantum measurement $\max_{j\in S_k} P_{\text{MC}}(\rho^k_j)$ \cite{{CABGJ06}}, where
\begin{equation}
\label{BB2}
 P_{\text{MC}}(\rho^k_j)=\max_{Q\geq 0}\frac{p^k_jTr[Q\rho^k_j]}{Tr[Q\rho^k]},  
\end{equation}
where in the relevant state discrimination task Alice receives the quantum state $\rho^k_j$ with probability $p^k_j$, for all $j\in S_k$, and where $\rho^k=\sum_{j\in S_k} p^k_j\rho^k_j$.

Because $P_{\text{MC}}$ can be shown to increase relatively little for small variations from the ideal protocol, this result allows us to derive significantly tighter and more general security bounds for S-money quantum tokens of Ref.~\cite{KLPGR22}, in which we allow the prepared states to deviate from the target BB84 state up to an angle $\theta$ on the Bloch sphere.  
\adr{Previous security analyses (e.g. \cite{KLPGR22}; see also \cite{LKBHTKGWZ13})} assumed that the four states belonged to two qubit orthonormal bases, which cannot be precisely guaranteed in a realistic experimental setup. 

We further refine the security analysis for the S-money quantum tokens of \cite{KLPGR22} by allowing a small probability $P_\theta$ that the qubit prepared states deviate from the intended BB84 states by an angle greater than $\theta$ in the Bloch sphere.  This allows security to be proven based on experimental data that sample the distribution of deviations from BB84 states.

More broadly, we believe our security analysis can be helpful to analyse the security of practical implementations of mistrustful quantum cryptography. Together with the analysis of multiphoton attacks in Ref.~\cite{BCDKPG21}, these results provide a more rigorous security analysis of implementations of mistrustful quantum cryptography with realistic experimental setups. This is crucial for 
developing the secure mistrustful quantum cryptographic applications envisaged for free space and fibre optic quantum networks and the eventual quantum internet \cite{K08,WER18}.

\section{Estimation of experimental imperfections}
\label{sec:exp}

In this section we discuss our experimental procedure to determine the reported values for the experimental imperfections given by $\beta_{\text{PS}},\beta_{\text{PB}},\beta_\text{E},E$ and $P_\text{noqub}$. The estimations of $\theta$ and $P_\theta$ are provided in section~\ref{sec:theta}. In this document, for a variable $y$ depending on variables $x_1,\ldots,x_n$, we estimate its standard deviation by
\begin{equation}
\label{Ap1}
    \sigma_y=\sqrt{\sum_{i=1}^n\Bigl(\frac{\partial y}{\partial x_j}  \sigma_{x_j} \Bigr)^2 }\,,
\end{equation}
where $\sigma_{x_j}$ is the standard deviation of $x_j$, for all $j\in\{1,2,\ldots,n\}$.

The frequency of our photon source was set at $f_{\text{sys}}=500$~kHz. We collected data for a time $T_{\text{exp}}=331,465$ s. Thus the total number of pulses was $N= T_{\text{exp}}f_{\text{sys}}=1.657325\times 10^{11}$. 
Our setup used a heralding single-photon source. Thus, only the photon pulses activating Bob's heralding detector are considered in our analysis below unless otherwise stated. We obtained $N_B=11,467,415$ pulses activating Bob's heralding detector. From these pulses, Alice's setup obtained $N_0 = 1,348,725 $ events with no detectors being activated, $N_1= 10,118,574$ events  activating only one detector, and  $N_2= 116$ events activating both of her detectors. The $N_0$ and $N_2$ pulses activating none or both detectors were assigned a random measurement outcome, in agreement with the scheme.









In the experimental setup, the biases $\beta_{\text{PB}}$ and $\beta_{\text{PS}}$ in selecting the preparation basis and the preparation state were determined by random numbers. The numbers of selected bases $u_i\in\{0,1\}$ and outcomes $t_i\in\{0,1\}$ were
\begin{eqnarray}
N(u_i=0)&=&5,737,415 \, , \qquad N(u_i=1)=5,730,000 \, ,\nonumber\\
N(t_i=0)&=&5,732,749 \, ,\qquad N(t_i=1)=5,734,666 \, .
\end{eqnarray}
Note that
\begin{equation}
    \sum_{u=0}^1 N(u_i=u) =\sum_{t=0}^1 N(t_i=t)=N_\text{B}=11,467,415\,.
\end{equation}
The estimated biases were
\begin{eqnarray}
\bar{\beta}_{\text{PB}} &=& \left|\frac{N(u_i=0)}{N_{\text{B}}}-0.5\right|=
\dami{0.000324}\, ,\nonumber\\
\bar{\beta}_{\text{PS}} &=& \left|\frac{N(t_i=0)}{N_{\text{B}}}-0.5\right|=
\dami{0.000084}\,  , 
\end{eqnarray}
\adr{to six decimal places.}

The standard deviations if the observed frequencies represent the probabilities are smaller than but very close to those for distributions in which $0$ and $1$ are equiprobable.
Conservatively, we use the latter, taking
\begin{equation}
    \sigma_{\bar{\beta}_{\text{PB}}}=\sigma_{\bar{\beta}_{\text{PS}}}=\frac{1}{\sqrt{N_\text{B}}}=
    \dami{0.000148}\, 
\end{equation}
\adr{to six decimal places.}

Thus, our upper bounds for the biases allowing for seven standard deviation fluctuations were
\begin{eqnarray}
    \beta_{\text{PB}} &=&\bar{\beta}_{\text{PB}}+7\sigma_{\bar{\beta}_{\text{PB}}}=    
    \dami{0.001360}\, ,\nonumber\\
\beta_{\text{PS}} &=&\bar{\beta}_{\text{PS}}+7\sigma_{\bar{\beta}_{\text{PS}}}=
\dami{0.001120}\, .
\end{eqnarray}
\adr{Here and below we generally give experimental data to six decimal places or significant figures to aid comparison between estimates, standard deviations, and upper bound estimates.} 

Alice chose the bit $z$, labelling the measurement that she applied to all received quantum states during the quantum token generation phase, using a quantum random number generator with a bias of $\beta_E=10^{-5}$.





The measured average error rates were 
\begin{equation}
    \bar{E}_{tu}=\frac{N_{tu}^{\text{error}}}{N_{tu}} \, , 
\end{equation}
where $N_{tu}$ is the total number of pulses prepared by Bob in the state labelled by $t$ in the basis labelled by $u$ that were measured by Alice in the basis labelled by $u$, and where $N_{tu}^{\text{error}}$ is the number of Alice's incorrect measurement outcomes $t\oplus 1$ obtained from the $N_{tu}$ pulses, for all $t,u\in\{0,1\}$. Given that the frequencies $\bar{E}_{tu}$ are obtained from binomial distributions, their standard deviations are 
\begin{equation}
    \sigma_{\bar{E}_{tu}}=\sqrt{\frac{\bar{E}_{tu}(1-\bar{E}_{tu})}{N_{tu}}} \, ,
\end{equation}
for all $t,u\in\{0,1\}$. We compute the upper bounds $E_{tu}$ on the error rates  by
\begin{equation}
    E_{tu}=\bar{E}_{tu}+7\sigma_{\bar{E}_{tu}} \, , 
\end{equation}
for all $t,u\in\{0,1\}$. This means that the hypothesis that $E_{tu}$ is a valid upper bound for the corresponding error probability is incorrect with a probability $\leq2.6\times 10^{-12}$, for all $t,u\in\{0,1\}$. The results are given in Table~S1.

\begin{table}
\centering
\renewcommand{\thetable}{S\arabic{table}}
\caption{\centering{Statistics for the average error rates.}}
\begin{tabular}{cccccc}
$(t,u)$ & $N^{\text{error}}_{tu}$ & $N_{tu}$ & $\dami{\bar{E}_{tu}}$ & $\dami{\sigma_{\bar{E}_{tu}}}$ & $\dami{E_{tu}}$\\
\hline
(0,0) & 89,317 & 1,508,557 & 5.92\dami{06911} \% &\dami{0.0192155 \%} &\dami{6.0551998 \%} \\
(0,1) & 92,020 & 1,507,895 & 6.10\dami{25469} \% &\dami{0.0194938 \%} &\dami{ 6.2390037\%}\\
(1,0) & 82,505 & 1,358,476 & 6.07\dami{33498} \% &\dami{0.0204919 \%} & \dami{ 6.2167933\%}\\
(1,1) & 82,923 & 1,356,953 & 6.11\dami{09707} \% &\dami{0.0205627 \%} & \dami{ 6.2549096\%}\\
\end{tabular}
\end{table}

The upper bound on the error probability used in our security analysis is given by
\begin{equation}
    E=\max_{t,u\in\{0,1\}}\{E_{tu}\}=
    \dami{ 6.2550}\, \% \, ,
\end{equation}
\dami{to six decimal places, rounding up.}



\subsection{Bounds on dark count probabilities}

\dam{The experimental setup guarantees that the number of generated pulses in the time interval $T_{\text{d}}$ is $N_{\text{d}}=T_{\text{d}} f_{\text{sys}}$, where the source emits pulses at the frequency $f_\text{sys}$. 
}

Let \dam{$d_{\text{B}}, d_{\text{A0}}, d_{\text{A1}}$} be the dark count probabilities of Bob's detector and Alice's two detectors, respectively. We define Alice's combined dark count probability to be
\begin{equation}
\label{Ap2}
d_{\text{A}} = 1 - (1 - d_{\text{A0}})(1 - d_{\text{A1}})\,.
\end{equation}
Let $N_{\text{d}_{\text{B}}}, N_{\text{d}_{\text{A0}}}$, and $N_{\text{d}_{\text{A1}}}$ be the number of dark counts in Bob's detector and Alice's detectors during the time interval $T_{\text{d}}$, and let $\sigma_{N_{\text{d}_{\text{B}}}}, \sigma_{N_{\text{d}_{\text{A0}}}}$, and $\sigma_{N_{\text{d}_{\text{A1}}}}$ be their standard deviations, respectively. In the limit $N_{\text{d}} \to \infty$, we have 
\begin{equation}
\label{Ap3}
d_{\text{A0}} = \frac{N_{\text{d}_{\text{A0}}}}{N_{\text{d}}} \, , \quad d_{\text{A1}} = \frac{N_{\text{d}_{\text{A1}}}}{N_{\text{d}}}\,,\quad d_{\text{B}} = \frac{N_{\text{d}_{\text{B}}}}{N_{\text{d}}} \, .
\end{equation}

In practice, there will be some uncertainty in these estimations. We obtain the standard deviations for $d_{\text{A}}$ and $d_{\text{B}}$ using the experimental data $T_{\text{d}}$, $f_{\text{sys}}$, $N_{\text{d}_\text{B}}$, $N_{\text{d}_\text{A0}}$, $N_{\text{d}_\text{A1}}$:
\begin{equation}
\sigma_{d_\text{B}} = \frac{\sigma_{N_{\text{d}_\text{B}}}}{\dam{N_\text{d}}} \, ,  \quad \sigma_{d_{\text{A0}}} = \frac{\sigma_{N_{\text{d}_{\text{A0}}}}}{\dam{N_\text{d}}} \, ,  \quad \sigma_{d_{\text{A1}}} = \frac{\sigma_{N_{\text{d}_{\text{A1}}}}}{\dam{N_\text{d}}} \, , 
\end{equation}
\begin{equation}
\sigma_{d_{\text{A}}} = \sqrt{((1 - d_{\text{A1}})\sigma_{d_{\text{A0}}})^2 + ((1 - d_{\text{A0}})\sigma_{d_{\text{A1}}})^2} \, , 
\end{equation}
where
\dam{
\begin{equation}
\sigma_{N_{\text{d}_{\text{B}}}} = \sqrt{N_{\text{d}_{\text{B}}}} \, , \quad \sigma_{N_{\text{d}_{\text{A0}}}} = \sqrt{N_{\text{d}_{\text{A0}}}} \, , \quad \sigma_{N_{\text{d}_{\text{A1}}}} = \sqrt{N_{\text{d}_{\text{A1}}}} \, .
    \label{Ap5}
\end{equation}
}

We have $T_{\text{d}} = 75906$ s, $f_{\text{sys}} = 500$ kHz, $N_{\text{d}_{\text{B}}} = 17111$, $N_{\text{d}_{\text{A0}}} = 12985$, and $N_{\text{d}_{\text{A1}}} = 13354$. From \dam{(\ref{Ap2})--(\ref{Ap5})}, we obtain:
\begin{eqnarray}
\label{Ap6}
d_{\text{A0}} &=& 3.42134 \times 10^{-7} \, ,  ~ d_{\text{A1}} = 3.51856 \times 10^{-7} \, ,  ~  d_{\text{A}} = 6.9399 \times 10^{-7} \, ,  ~  d_{\text{B}} = 4.50847 \times 10^{-7} \, , \nonumber\\
\sigma_{d_{\text{A0}}} &=& 3.00244 \times 10^{-9} \, ,  ~  \sigma_{d_{\text{A1}}} = 3.04481 \times 10^{-9} \, ,  ~  \sigma_{d_{\text{A}}} = 4.27615 \times 10^{-9} \, ,  ~  \sigma_{d_{\text{B}}} = 3.44661 \times 10^{-9} \, .\nonumber\\
\end{eqnarray}

\dam{

\subsection{Bounds on pulse detection probabilities}

The experimental setup guarantees that the number of generated pulses in the time interval $T_{\text{exp}}$ is $N=T_{\text{exp}} f_{\text{sys}}$, where the source emits pulses at the frequency $f_{\text{sys}}$.

Let $N_\text{A}$, $N_\text{B}$ and $N_\text{C}$ be the number of pulses activating Alice’s detectors, activating Bob’s detector, and creating a coincidence in Alice’s and Bob’s detectors, during the time interval $T_\text{exp}$, respectively. Let $P_\text{A}$, $P_\text{B}$ and $P_\text{C}$ be the corresponding pulse detection probabilities. In the limit $N\rightarrow \infty$, we have
\begin{equation}
\label{Ap7}
    P_\text{A}=\frac{N_\text{A}}{N} \, , \qquad		 P_\text{B}=\frac{N_\text{B}}{N} \, , \qquad		 P_\text{C}=\frac{N_\text{C}}{N} \, .
\end{equation}
In practice, there will be some uncertainty in these estimations. We obtain the standard deviations for $P_\text{A}$, $P_\text{B}$ and $P_\text{C}$ using the experimental data $N_\text{A}$, $N_\text{B}$, $N_\text{C}$:
\dam{
\begin{equation}
\label{Ap8}
    \sigma_{P_\text{A}}=\frac{\sigma_{N_\text{A}}}{N} \, , \qquad		 \sigma_{P_\text{B}}=\frac{\sigma_{N_\text{B}}}{N} \, , \qquad		 \sigma_{P_\text{C}}=\frac{\sigma_{N_\text{C}}}{N} \, , 
\end{equation}}
where
\begin{equation}
\label{Ap9}
    \sigma_{N_\text{A}}=\sqrt{N_\text{A}} \, , \qquad		\sigma_{N_\text{B}}=\sqrt{N_\text{B}} \, , \qquad		 \sigma_{N_\text{C}}=\sqrt{N_\text{C}} \, .
\end{equation}
We have $T_{\text{exp}}=331465$ s, $f_{\text{sys}}=500$ kHz, $N_\text{B}=11467415$, $N_\text{A}=12021392$ and $N_\text{C}=10118690$. From (\ref{Ap7})--(\ref{Ap9}) we obtain
\begin{eqnarray}
\label{Ap10}
P_\text{A}&=&7.25349\times10^{-5} \, , \quad P_\text{B}=6.91923\times10^{-5} \, , \quad P_\text{C}=6.10543\times10^{-5} \, , \nonumber\\
\sigma_{P_\text{A}}&=&2.09204\times10^{-8} \, , \quad \sigma_{P_\text{B}}=2.04327\times10^{-8} \, , \quad \sigma_{P_\text{C}}=1.91935\times10^{-8} \, . 
\end{eqnarray}
}

\subsection{Derivation of upper bound on $P_{\text{noqub}}$}
\label{Pnoqubapp}

\dami{In this subsection, we derive an upper bound on the probability $P_\text{noqub}$ that a heralded photon pulse that Alice sends Bob is multi-photon. \dami{We} assume that the photon number distribution of the photon pairs is Poissonian~\cite{avenhaus2008photon,schneeloch2019introduction}. That is, we assume the quantum density matrix for the photon pairs is given by
\begin{equation}
\rho = \sum_{j=0}^\infty \frac{\mu^je^{-\mu}}{j!} \ket{jj}\bra{jj} \, , 
\end{equation}
where $\ket{jj}\bra{jj}$ denotes the quantum state of $j$ pairs of photons.}

We assume Alice's detectors have a combined efficiency $\eta_{\text{A}} = q\eta_{\text{A0}} + (1 - q)\eta_{\text{A1}}$ and dark count probability $d_{\text{A}}$, where $q$ represents the probability of a photon going to the first detector (which we assume remains constant throughout the experiment), and where $d_{\text{A0}}$, $d_{\text{A1}}$ are the detectors' individual dark count probabilities with $1 - d_{\text{A}} = (1 - d_{\text{A0}})(1 - d_{\text{A1}})$. Bob's detector has efficiency $\eta_{\text{B}}$ and dark count probability $d_{\text{B}}$. Let $P_{\text{B}}$ be the probability that a pulse activates a detection in Bob's detector. Let $P_{\text{A}}$ be the probability that a pulse activates a detection in one of Alice's detectors. Let $P_{\text{C}}$ be the probability that a pulse activates a \dam{`coincidence'},  i.e., is detected by Bob's detector and one of Alice's detectors.

The form of $P_{\text{noqub}}$ when Bob's dark count probability \dam{$d_\text{B}$} is non-zero is given by:
\begin{equation}
\label{Ap11}
P_{\text{noqub}} = 1 - \frac{e^{-\mu}[d_\text{B}(1+\mu)+(1-d_\text{B})\mu\eta_\text{B}]} {d_\text{B} + (1-d_\text{B})(1-e^{-\mu\eta_\text{B}})} \, .
\end{equation}
We have derived the following equations, by assuming two bounding scenarios for calculating $P_\text{A}$ (and one for $P_\text{C}$ as we only need the upper bound). One scenario is where a multi-photon pulse is guaranteed to activate a detection at one of Alice's detectors, and another in which such a pulse never activates a detection at either of Alice's detectors:

\begin{equation}
\label{Ap12}
   d_\text{A}+(1-d_\text{A} ) e^{-\mu}\mu\eta_\text{A}\leq P_\text{A}\leq d_\text{A}+(1-d_\text{A} )\bigl[1-e^{-\mu}-\mu e^{-\mu} (1-\eta_\text{A})\bigr]\,,
\end{equation}
\begin{equation}
\label{Ap13}
   P_\text{B} = d_\text{B}+(1-d_\text{B} )\bigl[1-e^{-\mu\eta_\text{B}}\bigr]\,,
\end{equation}
\begin{eqnarray}
\label{Ap14}
P_\text{C}&\leq& d_\text{A}d_\text{B} + d_\text{A}(1-d_\text{B})\bigl[1-e^{-\mu\eta_\text{B}}\bigr]+d_\text{B}(1-d_\text{A} )\bigl[1-e^{-\mu}-\mu e^{-\mu}(1-\eta_\text{A})\bigr]\nonumber\\
&&\qquad+(1-d_\text{A})(1-d_\text{B})\bigl[e^{-\mu} \mu\eta_\text{A} \eta_\text{B}+1-e^{-\mu}-\mu e^{-\mu})\bigr]\nonumber\\
&\leq& d_\text{A}+d_\text{B}-d_\text{A}d_\text{B}+(1-d_\text{A})(1-d_\text{B})\bigl[e^{-\mu}\mu \eta_\text{A} \eta_\text{B} + 1-e^{-\mu}-\mu e^{-\mu})\bigr]\,.
\end{eqnarray}
Note that when calculating the upper bound on $P_\text{C}$, in the final term we took the bounding assumption that, when there are no dark counts, any multi-photon pulse will activate a detection at both Alice and Bob.

We can then rearrange (\ref{Ap12}) and (\ref{Ap14}) to get
\begin{eqnarray}
\label{Ap15}
e^{-\mu} \mu \eta_\text{A} &\leq& \frac{P_\text{A}-d_\text{A}}{1-d_\text{A}}<e^{-\mu}\mu \eta_\text{A}+1-e^{-\mu}-\mu e^{-\mu} \, , \nonumber\\
\frac{P_\text{C}-d_\text{A}-d_\text{B}+d_\text{A} d_\text{B}}{(1-d_\text{A} )(1-d_\text{B} )} &<&e^{-\mu} \mu \eta_\text{A} \eta_\text{B} + 1-e^{-\mu}-\mu e^{-\mu} \, .
\end{eqnarray}
Also, by a rearrangement of (\ref{Ap13}) we get
\begin{equation}
\label{Ap16}
\mu\eta_\text{B}=\ln{\Bigl(\frac{1-d_\text{B}}{1-P_\text{B}}\Bigr)}\, .
\end{equation}

To improve the clarity of the following calculations, we define new variables
\begin{equation}
\label{Ap17}
 x_\text{A}\equiv \frac{P_\text{A}-d_\text{A}}{1-d_\text{A}} \, , \quad  x_\text{B}\equiv {\rm ln}{\Bigl(\frac{1-d_\text{B}}{1-P_\text{B}}\Bigr)} \, , \quad
 x_\text{C}\equiv \frac{P_\text{C}-d_\text{A}-d_\text{B}+d_\text{A} d_\text{B}}{(1-d_\text{A} )(1-d_\text{B})} \, .
\end{equation}

We are able to use (\ref{Ap15})--(\ref{Ap17}) to derive the bound
\begin{equation}
\label{Ap18}
x_\text{C}<\frac{x_\text{A} x_\text{B}}{\mu}+1-e^{-\mu}-\mu e^{-\mu} \, . 
\end{equation}
We use the bound $1-\mu<e^{-\mu}$ for $\mu>0$  and an \adr{additional} weak assumption that $\mu<0.005$ \dam{(which is justified from the experimental data in section~\ref{upboundmu})} to \adr{obtain the more useful inequality}
\begin{equation}
\label{Ap19}
x_\text{C}<\frac{x_\text{A} x_\text{B}}{\mu}+0.005\mu \, . 
\end{equation}
We can then rearrange (\ref{Ap19}) to get
\begin{equation}
\label{Ap20}
\mu^2-200x_\text{C}\mu+200x_\text{A} x_\text{B}>0 \, .
\end{equation}
which leads to the restriction
\begin{equation}
\label{Ap21}
\mu<100x_\text{C}-\sqrt{10000x_\text{C}^2-200x_\text{A} x_\text{B}}=\mu^{\text{U}} \, , 
\end{equation}
if the observed $x_\text{A}$, $x_\text{B}$, $x_\text{C}$ satisfy $\mu<0.005\leq 100x_\text{C}+ \sqrt{10000x_\text{C}^2-200x_\text{A} x_\text{B}}$ (which the experimental observations do). We take $\mu^{\text{U}}$ as an upper bound for $\mu$. We can now obtain standard deviations of the quantities in (\ref{Ap17}) and (\ref{Ap21}):
\begin{eqnarray}
\label{Ap22}
\sigma_{x_\text{A}}&=&\sqrt{\Biggl(\frac{(1-P_\text{A})\sigma_{d_\text{A}}}{(1-d_\text{A})^2}\Biggr)^2+\Bigl(\frac{\sigma_{P_\text{A}}}{1-d_\text{A}}\Bigr)^2} \, , \nonumber\\
\sigma_{x_\text{B}}&=&\sqrt{\Bigl(\frac{\sigma_{d_\text{B}}}{1-d_\text{B}}\Bigr)^2+\Bigl(\frac{\sigma_{P_\text{B}}}{1-P_\text{B}}\Bigr)^2} \, , \nonumber\\
\sigma_{x_\text{C}}&=&\sqrt{\Biggl(\frac{\sigma_{P_\text{C}}}{(1-d_\text{A})(1-d_\text{B})}\Biggr)^2+\Biggl(\frac{(P_\text{C}-1)\sigma_{d_\text{A}}}{(1-d_\text{A})^2(1-d_\text{B})}\Biggr)^2+\Biggl(\frac{(P_\text{C}-1)\sigma_{d_\text{B}}}{(1-d_\text{A})(1-d_\text{B})^2}\Biggr)^2} \, , \nonumber\\
\sigma_{\mu^{\text{U}}}&=&\sqrt{\Biggl(\frac{\partial \mu^{\text{U}}}{\partial x_\text{A}} \sigma_{x_\text{A}} \Biggr)^2+\Biggl(\frac{\partial \mu^{\text{U}}}{\partial x_\text{B} } \sigma_{x_\text{B}} \Biggr)^2+\Biggl(\frac{\partial \mu^{\text{U}}}{\partial x_\text{C}} \sigma_{x_\text{C}} \Biggr)^2 } \, , \nonumber\\
\frac{\partial \mu^{\text{U}}}{\partial x_\text{A}}&=&100x_\text{B} \bigl(10000x_\text{C}^2-200x_\text{A} x_\text{B} \bigr)^{-\frac{1}{2}} \, , \nonumber\\
\frac{\partial \mu^{\text{U}}}{\partial x_\text{B}}&=&100x_\text{A} \bigl(10000x_\text{C}^2-200x_\text{A} x_\text{B} \bigr)^{-\frac{1}{2}} \, , \nonumber\\
\frac{\partial \mu^{\text{U}}}{\partial x_\text{C}}&=&100-10000x_\text{C} \bigl(10000x_\text{C}^2-200x_\text{A} x_\text{B} \bigr)^{-\frac{1}{2}} \, .
\end{eqnarray}
Then, by using (\ref{Ap11}) and (\ref{Ap22}), with use of $x_\text{B}=\mu\eta_\text{B}$, we can obtain a standard deviation of an upper bound of $P_{\text{noqub}}$, which we shall call $P_{\text{noqub}}^{\text{U}}$.

Next, we verify that the quantity $P_{\text{noqub}}$ increases with $\mu$, so that the upper bound $\mu^{\text{U}}$ can be used to calculate $P_{\text{noqub}}^{\text{U}}$. This is seen from the relation
\begin{equation}
\label{Ap23}
e^{-\mu^{\text{U}}}\bigl(1+\mu^{\text{U}}\bigr)<e^{-\mu} (1+\mu)\,,
\end{equation}
providing the upper bound
\begin{equation}
\label{Ap24}
P_{\text{noqub}}=1-\frac{e^{-\mu} \bigl[d_\text{B} (1+\mu)+(1-d_\text{B})x_\text{B}\bigr]}{d_\text{B}+(1-d_\text{B})\bigl(1-e^{-x_\text{B}}\bigr)}<1-\frac{e^{-\mu^{\text{U}}}\bigl[d_\text{B} \bigl(1+\mu^{\text{U}}\bigr)+(1-d_\text{B})x_\text{B}\bigr]}{d_\text{B}+(1-d_\text{B})\bigl(1-e^{-x_\text{B}}\bigr)}=P_{\text{noqub}}^{\text{U}} \, , 
\end{equation}
and the associated standard deviation
\begin{eqnarray}
\label{Ap25}
\sigma_{P_{\text{noqub}}^{\text{U}}} &=&\sqrt{\Biggl(\frac{\partial P_\text{{noqub}}^{\text{U}}}{\partial \mu^{\text{U}}} \sigma_{\mu^{\text{U}} } \Biggr)^2+\Biggl(\frac{\partial P_{\text{noqub}}^{\text{U}}}{\partial x_\text{B}} \sigma_{x_\text{B}} \Biggr)^2+\Biggl(\frac{\partial P_{\text{noqub}}^{\text{U}}}{\partial d_\text{B}} \sigma_{d_\text{B}}\Biggr)^2 } \, , \nonumber\\
\frac{\partial P_{\text{noqub}}^{\text{U}}}{\partial \mu^{\text{U}} }&=&\frac{e^{-\mu^{\text{U}} } \bigl(d_\text{B} \mu^{\text{U}}+(1-d_\text{B} ) x_\text{B} \bigr)}{d_\text{B}+(1-d_\text{B} )\bigl(1-e^{-x_\text{B} } \bigr) } \, , \nonumber\\
\frac{\partial P_{\text{noqub}}^{\text{U}}}{\partial x_\text{B} }&=&-\frac{e^{-\mu^{\text{U}} } (1-d_\text{B})\bigl(1-e^{-x_\text{B} } (1+d_\text{B} \mu^{\text{U}}+(1-d_\text{B} ) x_\text{B})\bigr)}{\bigl(d_\text{B}+(1-d_\text{B} )(1-e^{-x_\text{B}})\bigr)^2} \,,\nonumber\\
\frac{\partial P_{\text{noqub}}}{\partial d_\text{B} }&=&-\frac{e^{-\mu^{\text{U}}} \bigl(1+\mu^{\text{U}}-x_\text{B}-e^{-x_\text{B}} (1+\mu^{\text{U}})\bigr)}{\bigl(d_\text{B}+(1-d_\text{B} )(1-e^{-x_\text{B}})\bigr)^2} \, .
\end{eqnarray}
We then bound $P_{\text{noqub}}$ above by the upper bound plus 7 corresponding standard deviations to achieve our final result
\begin{equation}
    \label{Ap26}
P_{\text{noqub}}<1-\frac{e^{-\mu^{\text{U}}} \bigl[d_\text{B} (1+\mu^{\text{U}} )+(1-d_\text{B} ) x_\text{B}\bigr]}{d_\text{B}+(1-d_\text{B} )(1-e^{-x_\text{B}})}+7\sigma_{P_{\text{noqub}}^{\text{U}}}=P_{\text{noqub}}^{\text{U}}+7\sigma_{P_{\text{noqub}}^{\text{U}}} \, . 
\end{equation}
We use 7 standard deviations so that the probability of exceeding the bound is small enough to satisfy our security criteria.

This upper bound evaluates numerically as
\begin{equation}
    \label{Ap27}
\dami{P_\text{noqub}\leq P_{\text{noqub}}^{\text{max}}=} P_{\text{noqub}}^{\text{U}}+7\sigma_{P_{\text{noqub}}^{\text{U}}}=
\dam{4.9\times 10^{-5}}\, ,
\end{equation}
\dam{to six \adr{decimal places}.}

\subsection{Lower bounds on detection probabilities}
\label{lowerboundsapp}
We also have from (\ref{Ap15})--(\ref{Ap17}) that when $0<\mu<1$ and using $1- \mu<e^{-\mu}$
\begin{equation}
\label{Ap28}
    \eta_{\text{A}}^{\text{L}}\equiv \frac{x_{\text{A}}}{\mu^{\text{U}}}-\mu^{\text{U}}<\frac{x_{\text{A}}-\mu^2}{\mu}<\eta_{\text{A}} \, ,  \qquad \eta_{\text{B}}^{\text{L}}\equiv \frac{x_{\text{B}}}{\mu^{\text{U}}}<\frac{x_{\text{B}}}{\mu}=\eta_{\text{B}} \, .
\end{equation}
By using the upper bound for $\mu$ in (\ref{Ap21}), we can compute 
\begin{equation}
\label{Ap29}
\sigma_{\eta_{\text{A}}^{\text{L}}}=\sqrt{\biggl(\frac{\sigma_{x_{\text{A}}}}{\mu^{\text{U}}}\biggr)^2+\biggl(\biggl(\frac{x_{\text{A}}}{(\mu^{\text{U}})^2}-1\biggr)\sigma_{\mu^{\text{U}}}\biggr)^2} \, ,  \qquad \sigma_{\eta_{\text{B}}^{\text{L}}}=\sqrt{\biggl(\frac{\sigma_{x_{\text{B}}}}{\mu^{\text{U}}}\biggr)^2+\biggl(\frac{x_{\text{B}}\sigma_{\mu^{\text{U}}}}{(\mu^{\text{U}})^2}\biggr)^2} \, .
\end{equation}
We evaluate these numerically to obtain
\begin{equation}
\label{Ap30}
\mu^{\text{U}}=8.30097\times 10^{-5} \, ,  \quad  \eta_{\text{A}}^{\text{L}}=0.865369\,, \quad  \eta_{\text{B}}^{\text{L}}=0.828142\,,
\end{equation}
and
\begin{equation}
\label{Ap31}
\sigma_{\mu^{\text{U}}}=4.51565\times 10^{-8} \, ,  \quad  \sigma_{\eta_{\text{A}}^{\text{L}}}=5.36449 \times 10^{-4} \, ,  \quad
\sigma_{\eta_{\text{B}}^{\text{L}}}=5.15047 \times 10^{-4} \, .
\end{equation}

\subsection{Derivation of an upper bound on $\mu$ from a lower bound on $\eta_{\text{B}}$}
\label{upboundmu}
Start by assuming that
\begin{equation}
\label{Ap32}
\eta_{\text{B}}>0.02\,.
\end{equation}
This is consistent with the experimental value of the lower bound $\eta_{\text{B}}^{\text{L}}$ of $\eta_{\text{B}}$ computed in section~\ref{lowerboundsapp}.

We show below that it follows from (\ref{Ap32}) and the experimental data that
\begin{equation}
\label{Ap33}
\mu<0.005\,,
\end{equation}
as assumed in section \ref{Pnoqubapp}.

Given the calculated quantity 
\begin{equation}
    \label{Ap34}
    x_{\text{B}}=\mu\eta_{\text{B}}=\ln \biggl( \frac{1-d_\text{B}}{1-P_\text{B}}\biggr)=6.87\times 10 ^{-5} \, , 
\end{equation}
we get that
\begin{equation}
    \label{Ap35}
\mu<50(x_\text{B}+7\sigma_{x_\text{B}}) \, ,     
\end{equation}
when considering an additional bound of seven standard deviations on $x_\text{B}$. Using the derived quantities for $x_\text{B}$ in section \ref{Pnoqubapp}, we get
\begin{equation}
\label{Ap36}
\mu<0.0035<0.005\,,
\end{equation}
as claimed.

\blk
\section{Upper bound on the uncertainty angle in the preparation state}
\label{sec:theta}

In this section we describe the experimental procedure used to determine an upper bound $\theta$ for the uncertainty angle on the Bloch sphere for Bob's prepared states, and an upper bound $P_\theta$ on the probability that this bound is not satisfied.

As shown in Fig.~1 of the main text, Bob prepared the quantum states using two Pockels cells to modulate the bases and encoded bits. The Pockels cells were driven by quantum random number generators (QRNGs). We label the four target states by `0', `1', `+' and `-'. Alice used a half-wave plate (HWP) and a polarizing beam splitter (PBS) to measure the quantum states.  We used Alice's setup to measure the quantum states in one of two bases by setting the HWP at one of two possible angles using a rotation mount. We experimentally estimated $\theta$ and $P_\theta$ using this joint setup, with some variations discussed below. 
Given the experimental setup, we refer to these below as measurements by Alice.   Note however that in a real world implementation, these estimations would be performed by Bob using his own independent measurement setup. 

To estimate $\theta$, Alice's two single photon detectors were replaced by two power meters; and the intensity of the incoming light pulse was set to the higher value of approximately 18 mW. For each of the four states prepared by Bob, we measured the intensity of light measured by each of the two power meters. We repeated this 1000 times for each of the four states prepared by Bob. If Bob prepared the target states perfectly and the experimental setup were ideal then only one of the two power meters would measure a non-zero value. However, due to imperfections of the preparation procedure and the experimental setup, both power meters measure non-zero values, although one is much smaller than the other.

We first assume that the optical devices involved all work ideally and that the nonzero value for the smaller measured intensity arises only due to an uncertainty angle $\alpha$ in the Bloch sphere for Bob's prepared states. 

We estimate $\alpha$ as follows. Bob prepared $K>>1$ photons in \dami{a} qubit state $\vert \psi\rangle$, aiming to prepare \dami{a} qubit state $\lvert \phi_0\rangle\dami{\in\{\ket{0},\ket{1},\ket{+},\ket{-}\}}$. He sent the pulses through the HWP and the PBS, which were set aiming to apply a quantum measurement in the orthonormal qubit basis $\{\lvert \phi_0\rangle,\lvert \phi_1\rangle\}$. In this case, the probability that a photon goes to the power meter corresponding to the state $\lvert \phi_0\rangle$ which measures the maximum intensity is given by
\begin{eqnarray}
P_\text{max}&=&\cos^2\Bigl(\frac{\alpha}{2}\Bigr)\nonumber\\
&\approx&\frac{I_\text{max}}{I_\text{min}+I_\text{max}} \, , 
    \end{eqnarray}
where $I_\text{min}$ and $I_\text{max}$ are the smaller and bigger intensities measured by the respective power meters, corresponding to measuring the states $\lvert \phi_1\rangle$ and $\lvert \phi_0\rangle$, respectively; and where the $\approx$ symbol arises due to the approximation of probabilities by the observed experimental frequencies. Since each pulse was $18~$mW,  the number of photons $K$ in each pulse satisfied $K>>1$, and we can take the second line as an equality to very good approximation. Therefore, we obtain
\begin{equation}
\label{idealalpha}
    \alpha = 2\arccos{\sqrt{1-\frac{1}{1+C}}} \, , 
\end{equation}
where
\begin{equation}
    C=\frac{I_\text{max}}{I_\text{min}}
\end{equation}
is the contrast of intensities measured by the power meters.

The procedure above gives us a value $\alpha_j$ for the $j$th pulse. We repeat this procedure for $n=1000$ pulses and obtain $\alpha=\max\{\alpha_j\}_{j\in[n]}$. With this definition, the bound $\alpha_j \leq \alpha$ is satisfied for all $n$ measurements. 

Now suppose that the probability that $\alpha_j > \alpha$ for a general pulse $j$ is $P_{\alpha}$, independently for each pulse.   The probability of finding $\alpha_j < \alpha$ for all pulses in our data is 
\begin{equation}
\label{AAA1}
P_\text{error}=(1-P_\alpha)^n \, . 
\end{equation}
If we take 
\begin{equation}
\label{AAA2}
    P_\alpha=0.027 \, , 
\end{equation}
we have 
\begin{equation}
\label{AAA3}
    P_\text{error}=1.2967\times 10^{-12}<2.09769\times 10^{-12} \, . 
\end{equation}
We thus infer that $P_\alpha \leq 0.027$ almost certainly, with the probability of 
the contrary being of the order of (\ref{AAA3}).   
This is consistent with the value of $2.09769\times 10^{-12}$ corresponding to the seven standard deviation confidence used for all other parameters measured in our experiment that are relevant for the security analysis.

This procedure is repeated for each of the four states targeted by Bob. Thus, we obtain four values for the angle $\alpha$: 
\begin{eqnarray}
\label{idealalphas}
    \alpha_0&=& 
    \dam{2.231222^\circ}\, ,\nonumber\\
    \alpha_1&=& 
    \dam{3.429185^\circ}\, ,\nonumber\\
    \alpha_+&=&
    \dam{2.769766 ^\circ}\, ,\nonumber\\
    \alpha_-&=& 
    \dam{2.088437^\circ}\,  ,
\end{eqnarray}
\dam{to six decimal places.} The uncertainty angle for Bob's state preparation would be given by $\theta=\max_{i\in\{0,1,+,-\}}\{\alpha_{i}\}$, giving $\theta=
\dam{3.429185^\circ}$. However, we need to consider the imperfections of the PBS, HWP and rotation mount to derive a more accurate value of $\theta$, as discussed below.

\subsection{Considering the experimental imperfections of the PBS, HWP and rotation mount}

We now provide a derivation of $\theta$ and $P_\theta$ considering the experimental imperfections of the PBS, HWP and rotation mount using the experimental setup depicted in Fig.~S1.


\begin{figure}
\label{figS2}
\renewcommand{\thefigure}{S\arabic{figure}}
	\centering
	\includegraphics[width=0.55\textwidth]{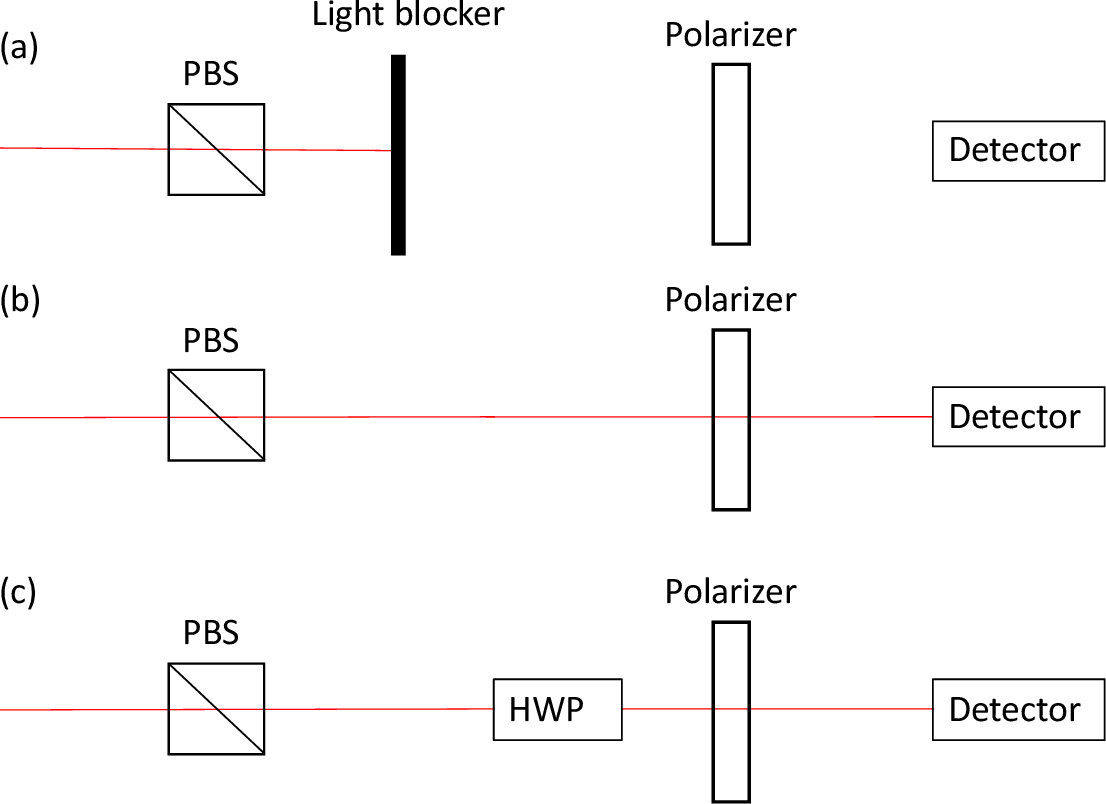}
	\caption{Experimental setup to measure the imperfections of the PBS, HWP and rotation mount.}
\end{figure}

We first measured the noise. We performed ten measurements of the intensity after blocking the laser, as shown in Fig.~S1(a). 
We obtained an average power of $I_\text{ave}= 5.74$~nW with a standard deviation of $0.91068$~nW. In subsequent measurements of the light intensity, we subtracted $I_\text{ave}$ from the measured power.

Then we measured the level of imperfection of the PBS. In practice, if a perfectly horizontally (vertically) polarized photon enters the PBS, it exits via the vertically (horizontally) polarized channel with a small \dami{nonzero} probability.

\dami{We model the actions of an imperfect PBS by a unitary operation $U_\text{PBS}$ as follows
\begin{equation}
    U_\text{PBS}\ket{i}=C_{i0}\ket{0}+C_{i1}\ket{1} \, ,     
\end{equation}
where $C_{il}\in\mathbb{C}$, $\sum_{j=0}^1\lvert C_{ij}\rvert ^2=1$, $C_{10}^*C_{00}+C_{11}^*C_{01}=0$, for all $i,l\in\{0,1\}$. We assume that $U_\text{PBS}$ is close to the identity operator, hence,  $\lvert C_{ii}\rvert ^2\approx 1$ and $\lvert C_{i\bar{i}}\rvert ^2<<1$, for all $i\in\{0,1\}$. Thus, $U_\text{PBS}$ acts as a rotation in the Bloch sphere by a small angle. We estimated the angle $\delta_\text{PBS}$ that $U_\text{PBS}$ rotates the state $\lvert 0\rangle$, corresponding to a horizontally polarized photon, in the Bloch sphere. We assume that this is a typical value for the rotation angle in the Bloch sphere implemented by $U_\text{PBS}$ on an arbitrary qubit state $\ket{\psi}$ input by Bob.

}


We placed the PBS before the polarizer, as shown in Fig. S1(b). The light after the PBS is supposed to be horizontally(vertically)-polarized. We set the polarizer to angles of 0 and 90 degrees to test the actual polarization. We measured the residual vertical component when the light is supposed to be horizontally-polarized, and vice versa. 
\dami{We repeated this procedure ten times}. We obtained a contrast of intensities given by
\begin{equation}
    \bar{C}_\text{PBS}= \frac{I^{\text{PBS}}_\text{max}}{I^{\text{PBS}}_\text{min}}= 161448\,,
\end{equation}
\dami{where $I^{\text{PBS}}_\text{min}$ and $I^{\text{PBS}}_\text{max}$ are the lower and 
higher intensities measured by the respective power meters, corresponding to measuring the states $\lvert \phi_1\rangle$ and $\lvert \phi_0\rangle$, respectively. We obtained} a standard deviation of
\begin{equation}
    \sigma_{\bar{C}_\text{PBS}}=1700 \, . 
\end{equation}
\dami{We consider a range of values for the intensity contrast including seven standard deviations, as follows:
\begin{equation}
C^\text{min}_\text{PBS} \leq    C_\text{PBS} \leq C^\text{max}_\text{PBS} \, ,
    \end{equation}
where
\begin{equation}
\label{PBScontrastvalue}
    C^\text{min}_\text{PBS}=    \bar{C}_\text{PBS}-7\sigma_{\bar{C}_\text{PBS}} = 149548\, ,\qquad C^\text{max}_\text{PBS}=    \bar{C}_\text{PBS}+7\sigma_{\bar{C}_\text{PBS}} =173348\, .
\end{equation}

This allows us to estimate $\delta_{\text{PBS}}$ as follows. The probability that a horizontally polarized photon (i.e., having quantum state $\lvert 0\rangle$) is detected in the vertical polarization channel (i.e., corresponding to the quantum state $\lvert 1\rangle$), is given by
\begin{equation}
1-\cos^2{\Bigl(\frac{\delta_{\text{PBS}}}{2}\Bigr)}\approx \frac{I^{\text{PBS}}_\text{min}}{I^{\text{PBS}}_\text{min}+I^{\text{PBS}}_\text{max}}=\frac{1}{1+C_{\text{PBS}}} \, , 
\end{equation}
where the $\approx$ symbol arises due to the approximation of probabilities by the observed experimental frequencies. The intensity of the input pulse was set at approximately $6~$mW, hence,  the number of photons $K$ in each pulse satisfied $K>>1$, and we can take equality to very good approximation. Considering the seven standard deviation measurement of the contrast $C_{\text{PBS}}$, we obtain the upper bound
\begin{equation}
\label{idealalpha}
    \delta_{\text{PBS}} \leq  2\arccos{\sqrt{1-\frac{1}{1+C^\text{min}_{\text{PBS}}}}}= 
    \dam{0.296321^\circ}\, ,
\end{equation}
\dam{to six decimal places,} where $C^\text{min}_{\text{PBS}}$ is given by (\ref{PBScontrastvalue}).

}

We then measured the imperfection of the HWP. We recall that the HWP was set at one of two possible angles by Alice during the quantum token generation in order to measure in one of the two bases. These angles were targeted at $0^\circ$ and $22.5^\circ$. If the HWP worked perfectly and these angles were precisely obtained then the HWP would map horizontal (vertical) polarization to horizontal (vertical) and to diagonal at $45^\circ$ (antidiagonal, i.e., at $135^\circ$) respectively. That is, the states $\ket{0}, \ket{1}$ would be mapped to the states $\ket{0}, \ket{1}$ when the HWP is set at $0^\circ$ or to $\ket{+}, \ket{-}$ when the HWP is set at $22.5^\circ$. 
However, imperfections of the HWP, PBS and the rotation mount imply that these mappings take place with uncertainty angles in the Bloch sphere, as we deduce below.



\dami{We model the actions of an imperfect HWP by a unitary operation $U_\text{HWP}$.} 
We \dami{estimate} an upper bound on the rotation error angle in the Bloch sphere introduced by the HWP, with the setup illustrated in Fig.~S1(c). To do this, let us assume for now that all source of error comes from the HWP and the PBS, neglecting the errors due to the rotation mount. Thus, we assume the polarizers are set exactly at $90^\circ$ between each other and are perfectly aligned with the axes of the PBS. Hence we estimate a maximum rotation error angle due to the HWP of
\begin{equation}
\label{beta}
\beta\leq \dami{\delta_{\text{PBS}}+2 {\rm arccos} \sqrt{1-\frac{1}{1+C^{\text{min}}_{\text{HWP}}}} \, ,} 
\end{equation}
where 
\dami{$C^\text{min}_\text{HWP}$ is a lower bound on} the contrast of intensities measured by the power meters.  
\dami{Note that we include the term $\delta_\text{PBS}$ to consider the imperfections of the PBS, as modelled above.}

We took ten measurements and obtained a contrast of intensities given by
\begin{equation}
    \bar{C}_{\text{HWP},01} =145551\,,
\end{equation}
with a standard deviation of 
\begin{equation}
   \sigma_{ \bar{C}_{\text{HWP},01}} = 1700\,,
\end{equation}
when we targeted the HWP at $0^\circ$; and a contrast of intensities given by
\begin{equation}
    \bar{C}_{\text{HWP},+-} = 9973\,,
\end{equation}
with a standard deviation of 
\begin{equation}
   \sigma_{ \bar{C}_{\text{HWP},+-}} = 14\,,
\end{equation}
when we targeted the HWP at $22.5^\circ$.

We consider a range of values for the intensity contrast including seven standard deviations, as follows:
\begin{eqnarray}
C^\text{min}_{\text{HWP},01} &\leq&    C_{\text{HWP},01} \leq C^\text{max}_{\text{HWP},01} \, ,\nonumber\\
C^\text{min}_{\text{HWP},+-} &\leq&    C_{\text{HWP},+-} \leq C^\text{max}_{\text{HWP},+-} \, ,
    \end{eqnarray}
where
\begin{eqnarray}
\label{HWPcontrastvalue}
    C^\text{min}_{\text{HWP},01}&=&    \bar{C}_{\text{HWP},01}-7\sigma_{\bar{C}_{\text{HWP},01}} = 133651\, ,\qquad C^\text{max}_{\text{HWP},01}=    \bar{C}_{\text{HWP},01}+7\sigma_{\bar{C}_{\text{HWP},01}} =157451\, ,\nonumber\\
    C^\text{min}_{\text{HWP},+-}&=&    \bar{C}_{\text{HWP},+-}-7\sigma_{\bar{C}_{\text{HWP},+-}} = 9875\, ,\qquad C^\text{max}_{\text{HWP},+-}=    \bar{C}_{\text{HWP},+-}+7\sigma_{\bar{C}_{\text{HWP},+-}} = 10071\, .\nonumber\\
    \end{eqnarray}

In order to obtain the maximum upper bounds for $\beta$ in (\ref{beta}), we take the \dami{lower bounds $C^\text{min}_{\text{HWP},01}$ and $C^\text{min}_{\text{HWP},+-}$ given by (\ref{HWPcontrastvalue})}. Thus, from \dami{(\ref{idealalpha}) and (\ref{beta})--(\ref{HWPcontrastvalue})}, we obtained the following upper bounds on the rotation error angle in the Bloch sphere introduced by the HWP
\begin{eqnarray}
\label{betabounds}
    \beta_{01}&\leq&
    \dam{0.609769^\circ}\, ,\nonumber\\
    \beta_{\pm}&\leq &
    \dam{ 1.449428^\circ}\, ,
\end{eqnarray}
\adr{to six decimal places,} when we targeted the HWP at $0^\circ$ (corresponding to measuring in the $\{\ket{0},\ket{1}\}$ basis, approximately) and at $22.5^\circ$ (corresponding to measuring in the $\{\ket{+},\ket{-}\}$ basis, approximately), respectively. \dami{We assume these provide valid upper bounds for arbitrary quantum states $\ket{\psi}$ input by Bob.}

Finally, we consider the imperfection of the rotation mount, which gives an uncertainty angle of $\frac{\delta_{\text{RM}}}{2}=0.05^\circ$, corresponding to $\delta_{\text{RM}}=0.1^\circ$ in the Bloch sphere. We now deduce the final uncertainty angle $\theta$ in the Bloch sphere taking into account $\alpha,\beta\dami{,\delta_\text{PBS}}$ and $\delta_{\text{RM}}$.

We first consider that the horizontal or vertical polarization states (i.e., $\ket{0}$ or $\ket{1}$) are prepared and so the HWP is aimed at $0^\circ$. To deduce an upper bound on $\theta$ we assume the worst case scenario in which the Bloch vectors of the prepared state and of the final state after passing through the HWP are aligned to the real plane defined by the axes of the PBS, as in Fig.~S2(a). We assume $\beta<\theta$.

\begin{figure}
\renewcommand{\thefigure}{S\arabic{figure}}
	\centering
	\includegraphics[width=0.55\textwidth]{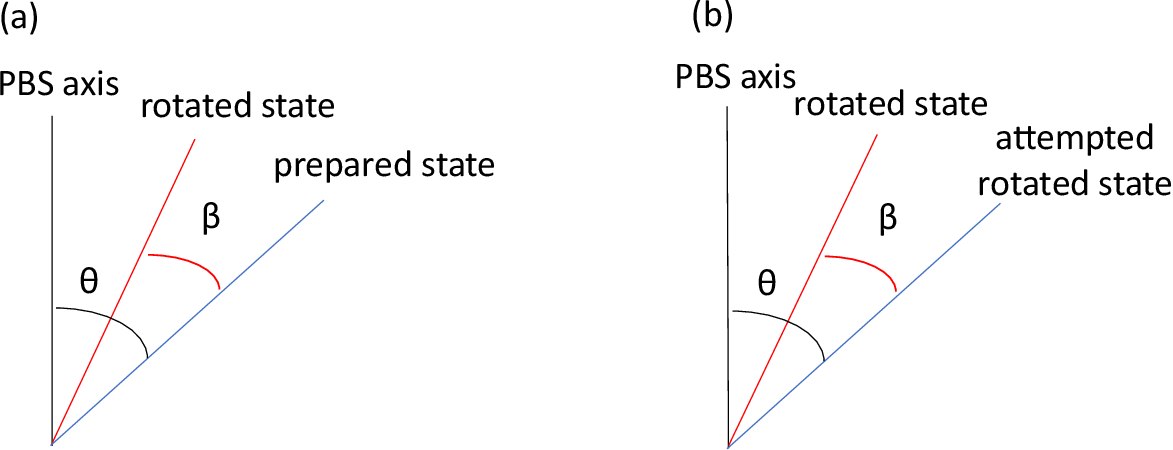}
	\caption{Rotation angle errors due to the HWP imperfections.}
	\label{fig:betatheta}
\end{figure}

We first assume that the rotation mount is perfect. Thus, we have
\begin{equation}
\label{X1}
    \theta-\beta \leq \alpha\dami{+\delta_\text{PBS}} \, , 
\end{equation}
hence,
\begin{equation}
\label{AAA4}
    \theta\leq \alpha + \beta\dami{+\delta_\text{PBS}} \, .
\end{equation}
\dami{Note that we consider the PBS imperfections by adding the angle $\delta_{\text{PBS}}$.}

Now let us consider that the states $\ket{+}$ or $\ket{-}$ are prepared. Thus, the HWP is aimed at $22.5^\circ$. In this case, if the HWP were perfect and perfectly aligned, the Bloch vector of the imperfectly prepared state would be rotated to the ‘attempted rotated state’ shown in Fig.~S2(b). Due to the imperfections and misalignment of the HWP, the state is in fact rotated to the ‘rotated state’ illustrated in Fig.~S2(b). Fig.~S2(b) illustrates the worst case scenario that allows us to derive an upper bound on $\theta$. As in Fig.~S2(a), we assume the two Bloch vectors lie on the real plane defined by the PBS axes. Thus, as above, (\ref{X1}) and (\ref{AAA4}) hold.

If $\beta \geq \theta$ then (\ref{AAA4}) holds trivially. Thus, (\ref{AAA4}) holds in general without needing to assume $\beta<\theta$.

Due to the imperfection of the rotation mount, we need to consider the uncertainty angle  $\delta_{\text{RM}}$ in the Bloch sphere, contributing up to twice this value due to misalignment of the HWP, and contributing on this value due to misalignment of the PBS, giving a total uncertainty angle on the Bloch sphere of $3\delta_{\text{RM}}$.  \adr{Conservatively, we} double this uncertainty to allow for the possibility that these uncertainties take maximum and opposite values when Alice and Bob perform the quantum token generation and when Bob estimates $\theta$ as discussed in this section. Thus, we obtain
\begin{eqnarray}
\label{thetas}
    \theta_i &\leq& \alpha_i + \beta_{01}+\dami{\delta_{\text{PBS}} +} 6\delta_{\text{RM}} \, , \nonumber\\
     \theta_l &\leq& \alpha_l + \beta_{\pm}+\dami{\delta_{\text{PBS}} +} 6\delta_{\text{RM}} \, , 
\end{eqnarray}
for $i=0,1$ and $l=`+',`-'$. Thus, from \dami{(\ref{idealalphas}), (\ref{idealalpha}),} (\ref{betabounds}) and (\ref{thetas}), we obtain
\begin{eqnarray}
\label{thetanumbers}
    \theta_0 &\leq &
    \dam{3.737312^\circ}\,  , \nonumber\\
     \theta_1 &\leq &
     \dam{4.935275^\circ}\,  , \nonumber\\
      \theta_+ &\leq &
\dam{5.115515^\circ}\,  , \nonumber\\
     \theta_- &\leq &  
     \dam{4.434186^\circ}\,  , \nonumber\\
\end{eqnarray}
\adr{to six decimal places.}
Finally, taking $\theta=\max_{i\in\{0,1,+,-\}}\{\theta_i\}$, and abusing notation by setting equality instead of inequality, we obtain our final upper bound
\begin{equation}
    \theta= 
    \dam{5.115515^\circ}\,  ,
\end{equation}
\adr{to six decimal places.} We also obtained 
\begin{equation}
    P_\theta= 0.027 \, , 
\end{equation}
by taking $P_\theta=P_\alpha$. As discussed above, \adr{if the uncertainty angle in Bob's state preparations is greater than $\theta$ with a probability greater than $P_\theta$, and the 
distributions are independent, the uncertainty bounds satisfied by our data would be obtained with probability $\leq 1.296 \times 10^{-12} < 2.6 \times 10^{-12}$. }

\adr{As noted above, this analysis assumes that Bob's PBS and HWP behave
with suitably small deviations from ideal specifications. 
We have estimated these from empirical data, obtaining significantly better estimates
than the manufacturer's stated error tolerances.   
We note that our analysis would imply a very high degree of unforgeability even if the 
deviations were significantly larger and $\theta$ were significantly higher.  
In practical application, if Bob has any reason
to suspect his devices might deviate substantially from the ideal, he could  
carry out full device tomographic tests.  }

\section{\adr{Time sequence and transaction times}}
\label{sec:times}

\dami{The chronological sequence of our quantum token implementation is shown in Fig.~S3. See main text for a complete description of the scheme.}

\begin{figure}
\renewcommand{\thefigure}{S\arabic{figure}}
	\centering
	\includegraphics[width=0.55\textwidth]{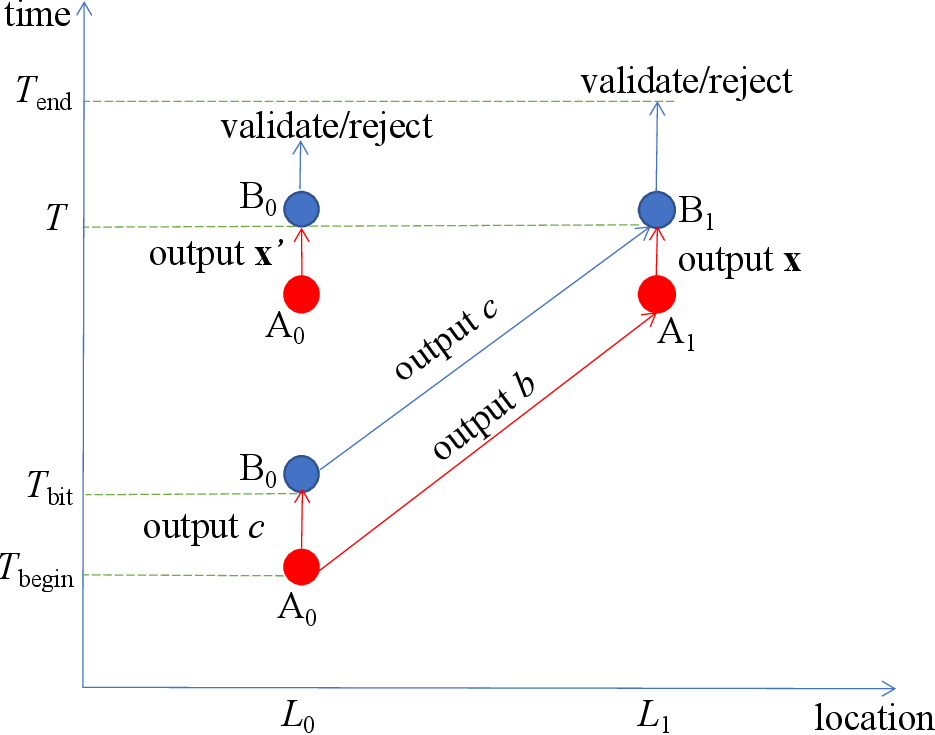}
	\caption{The chronological sequence of the token  transaction. 
 \dami{A$_0$ sends the bit $b$ to A$_1$ at the time $T_{\text{begin}}$ and the bit $c=z\oplus b$ to B$_0$ as soon as possible after, at the time $T_\text{bit}$. B$_0$ sends $c$ to B$_1$, which receives it by the time $T=T_\text{bit}+\Delta T_\text{comm}$, where $\Delta T_\text{comm}$ is the time that it takes a bit to be communicated from A$_0$ to A$_1$ and from B$_0$ to B$_1$. A$_b$ and A$_{b\oplus 1}$ present the token $\textbf{x}$ and the dummy token $\textbf{x}'$ to B$_b$ and B$_{b\oplus 1}$ within the time interval $[T,T+\Delta T]$, respectively. The case $b=1$ is illustrated. B$_0$ and B$_1$ validate or reject the token by the time $T_\text{end}\geq T+\Delta T$, using field programmable gate arrays (FPGAs). It is worth noting that pipeline processing is employed in the FPGAs to accelerate the verification speed. The transaction time is defined by $\Delta T_\text{tran}=T_\text{\text{end}}-T_{\text{begin}}$. The communications are output by electronic boards at the corresponding times as illustrated. We measured the difference between the times $T_\text{\text{end}}$ and $T_{\text{begin}}$ using an oscilloscope to obtain the transaction time $\Delta T_\text{tran}$.  }}
	\label{fig:sequence}
\end{figure}

\dami{In Fig.~3 of the main text, we presented the results for the quantum and comparative advantages (QA and CA). We implemented the quantum token scheme 20 times, randomly choosing the bit $b$ that denoted the presentation location $L_b$.

Table S2 gives details of the obtained time measurements relevant for computing QA in our intracity implementation in Jinan, with laboratories communicated by $L_\text{fibre}=2,766$~m of optical fibre and physically separated by 425 m. The average value for the transaction time $\Delta T_\text{tran}$ was $15.34~\mu s$. The transaction time of the classical cross-checking scheme discussed in the Methods section is given by $\Delta T_\text{tran,C}= 2 \Delta T_\text{comm}$ (Eq.~(8) in Methods), where $\Delta T_\text{comm}=\frac{L_\text{fibre}}{c_\text{fibre}}$ is the time that it takes to communicate a bit between $L_0$ and $L_1$ over the optical fibre channel, which took the value $\Delta T_\text{tran,C}=27.66~\mu s$ in our implementation in Jinan, and where $c_\text{fibre}=2\times 10^{8}$~ms$^{-1}$ is the speed of light through the optical fibre. Thus, we obtained a quantum advantage of $QA= \Delta T_\text{tran,C} - \Delta T_\text{tran}= 27.66~\mu \text{s} - 15.34~\mu \text{s} = 12.32~\mu \text{s}$ in our intracity implementation in Jinan.

Table S3 gives details of the obtained time measurements relevant for computing CA in our intercity implemetation between Yiyuan and Mazhan,  with laboratories communicated by 60.54 km field-deployable optical fibre and physically separated by $D=51.60$~km. The average value for the transaction time $\Delta T_\text{tran}$ was $304.20~\mu$s. The transaction time of the classical cross-checking scheme discussed in the Methods section is given by $\Delta T_\text{tran,CF}= \frac{2D}{c}=344~\mu$s (Eq.~(9) in Methods), where $c=3\times 10^{8}~$ms$^{-1}$ is the speed of light through a vacuum. Thus, we obtained a comparative advantage of $CA= \Delta T_\text{tran,CF} - \Delta T_\text{tran}= 344~\mu \text{s} - 304.20~\mu \text{s} = ~39.80~\mu \text{s}$ in our intercity implementation.


The measured values for the time intervals $\Delta T_\text{proc}$ for the FPGA communication and processing times, excluding the communication times between the distant locations $L_0$ and $L_1$, gave an average value of $1.507~\mu s$. The time intervals $\Delta T_\text{proc}$ are already included in the transaction times $\Delta T_\text{tran}$ reported in tables S2 and S3.

}





\begin{table}
\renewcommand{\thetable}{S\arabic{table}}
\centering
\caption{
\dami{Transaction times in the Jinan intracity implementation. The error rates correspond to the token presentation and validation stage at location $L_b$, which are below the predetermined threshold $\gamma_\text{err}=9.4\%$. The bits $b$ and $z$ denote the location presentation $L_b$ and the measurement basis $\mathcal{D}_z$ chosen by Alice.}}
\begin{tabular}{c|cccc}

\dami{Trial} & $b$ 
& $z$ & $\Delta T_{\text{tran}}$ ($\mu\text{s}$) 
& Error rate (\%)\\
\hline
1 & 1 & 1 & 15.309 & 5.347 \\
2 & 1 & 0 & 15.345 & 6.129 \\
3 & 0 & 0 & 15.348 & 6.073 \\
4 & 0 & 1 & 15.326 & 6.100 \\
5 & 1 & 0 & 15.334 & 6.126 \\
6 & 1 & 0 & 15.336 & 6.716 \\
7 & 0 & 1 & 15.347 & 5.573 \\
8 & 1 & 1 & 15.311 & 6.174 \\
9 & 1 & 0 & 15.335 & 6.100 \\
10 & 0 & 0 & 15.330 & 6.396 \\
11 & 1 & 1 & 15.370 & 5.907 \\
12 & 0 & 1 & 15.323 & 6.082 \\
13 & 0 & 0 & 15.333 & 5.541 \\
14 & 0 & 0 & 15.342 & 6.280 \\
15 & 1 & 1 & 15.319 & 6.538 \\
16 & 0 & 0 & 15.333 & 5.969 \\
17 & 0 & 1 & 15.351 & 5.725 \\
18 & 1 & 0 & 15.340 & 5.898 \\
19 & 1 & 1 & 15.343 & 5.904 \\
20 & 0 & 1 & 15.338 & 5.879 \\

\end{tabular}
\end{table}

\begin{table}
\centering
\renewcommand{\thetable}{S\arabic{table}}
\caption{\dami{Transaction times in the intercity implementation between Yiyuan and Mazhan. The error rates correspond to the token presentation and validation stage at location $L_b$, which are below the predetermined threshold $\gamma_\text{err}=9.4\%$. The bits $b$ and $z$ denote the location presentation $L_b$ and the measurement basis $\mathcal{D}_z$ chosen by Alice.}}
\begin{tabular}{c|cccc}

Trial & $b$ & $z$ & $\Delta T_{\text{tran}}$ ($\mu\text{s}$) & Error rate (\%)\\
\hline
1 & 0 & 1 & 304.198 & 6.432 \\
2 & 0 & 1 & 304.204 & 6.532 \\
3 & 1 & 0 & 304.194 & 5.204 \\
4 & 1 & 1 & 304.186 & 5.981 \\
5 & 0 & 1 & 304.220 & 5.917 \\
6 & 1 & 1 & 304.196 & 5.799 \\
7 & 1 & 1 & 304.220 & 5.461 \\
8 & 1 & 0 & 304.204 & 6.457 \\
9 & 0 & 0 & 304.210 & 5.526 \\
10 & 1 & 1 & 304.219 & 5.796 \\
11 & 0 & 1 & 304.212 & 5.895 \\
12 & 1 & 1 & 304.198 & 6.129 \\
13 & 0 & 0 & 304.215 & 5.643 \\
14 & 0 & 0 & 304.213 & 6.344 \\
15 & 0 & 0 & 304.185 & 5.876 \\
16 & 0 & 0 & 304.204 & 6.489 \\
17 & 1 & 0 & 304.211 & 5.986 \\
18 & 1 & 0 & 304.208 & 5.916 \\
19 & 0 & 1 & 304.204 & 5.900 \\
20 & 1 & 0 & 304.182 & 6.532 \\

\end{tabular}
\end{table}

\adr{\section{Related work \dami{and extensions of our schemes}}
\label{sec:related}

Various proposals for schemes for quantum money tokens and related concepts have
been considered and partially implemented.   Most such schemes require
quantum states to be propagated over a network and maintained with high fidelity.
Given current technology, implementations of these schemes
generally have thus been very short-range 
and short-lived, meaning that they are not practically applicable in their present
form.   However, they serve as valuable benchmarks of current technology.  
It will be important to continue careful comparisons between the functionality, resources required and technological feasibility 
of all proposals related to quantum money, including the S-money tokens discussed in the present work, as technology develops.   

Bartkiewicz et al.~\cite{BCGLMN17} describe an experimental implementation of 
partial cloning attacks on photon states representing components of quantum
money tokens.   The tokens were very short-lived, as they did not implement quantum
memory.  Bozzio et al.~\cite{BOVZKD18} demonstrated an on-the-fly version of quantum money tokens
using weak coherent states of light.   
Guan et al.~\cite{GAAZLYWZP18}  also implemented short-lived quantum money with 
light pulses, using high-dimensional time-bin qudits.    
Again, neither of these implementations were integrated with quantum memory, 
and so the tokens were very short-lived. 
Jirakova et al.~\cite{JBCL19} presented proof-of-concept attacks on the 
implementation reported in Ref.~\cite{BOVZKD18}. 

Behera et al.~\cite{BBP17} implemented a version of the so-called quantum cheques proposed by Moulick and Panigrahi~\cite{moulick2016quantum} on a 5-qubit IBM quantum computer.   
This test-of-concept demonstration was localized to this device, and relatively insecure because of the small number of qubits.  

Schiansky et al.~\cite{SERGTBW23} recently reported an experimental demonstration of quantum digital payments using an
interesting 
scheme that uses some features of quantum S-money.  Ref.~\cite{SERGTBW23} also offers comments on and comparisons with our scheme and a previous partial implementation~\cite{KLPGR22}.  It is important to note that the proposed use case 
of the scheme of Ref.~\cite{SERGTBW23} is very different from that on which we focus here, and its suggested
advantages are intended to apply in that use case.  
We comment briefly here only on some key points relevant to our present discussion, leaving a more complete comparison and analysis of the two schemes in various use cases for future work.

The scheme of Ref.~\cite{SERGTBW23} is not described in a relativistic context and not 
designed for the scenarios discussed in the present work.   
It involves three types of party, banks (trusted token providers or TTPs in the terminology of 
 Ref.~\cite{SERGTBW23}), clients and merchants.
To compare directly to the scenarios and use cases discussed in the present work, it is helpful to simplify
by eliminating the merchants (or alternatively considering them as untrusted classical channels).    
In that case, one or more TTPs would supply streams of quantum states over unauthenticated channels
that are accessible to clients, who may carry out appropriate measurements to generate classical tokens. 

If there are two or more spatially separated TTPs this allows a single client (with multiple agents) to generate independent classical tokens from each stream and present each to the relevant TTP, who will 
validate them without cross-checking.
The client can thus fraudulently multiply spend a single credit.
Hence the scheme does not protect against multiple spending in the scenarios 
considered in the present work.  (Its aim is to give a different type of protection, 
namely to prevent multiple transactions with the same identifier/label, without the need to cross-check.)   Protection could be obtained by cross checking, but in that case there is no time advantage compared to a purely classical cross checking scheme.
Alternatively, unforgeability can be guaranteed if there is a single TTP, as in the experimental demonstration of Ref.~\cite{SERGTBW23}. But in this case, two-way communication is required between the client and the TTP. Again this gives no time advantage over classical cross-checking schemes.

The scheme of Ref.~\cite{SERGTBW23} requires clients to first decide where to spend their credit
and then generate classical token data by carrying out appropriate measurements on a string of quantum
states.   
In contrast, our recent~\cite{KLPGR22} and current implementation allow the user 
flexibility to make this choice shortly before presenting the token for verification.
The time between choice and presentation is limited in principle only by causal
communication constraints, and in practice we achieved the sub-millisecond transaction times reported above. 
Our implementation allows the quantum communication between bank and user to take place
arbitrarily long (even years) \dami{before} the user chooses when and where to spend the token.
In this respect it replicates the real-world functionality of money and credit cards, 
which are typically obtained from banks without any need to commit to using them with a given merchant or at a given time and place.

The implementation of Ref.~\cite{SERGTBW23} required ``a few tens of minutes''~\cite{SERGTBW23} \dami{from Alice’s (the client's) choice of presentation point to token validation}.    This should be compared with the 
\adr{$\approx 1.5~\mu{\rm s}$ required to validate a presented token in our experiments,
the total transaction times of $\approx  15~\mu{\rm s}$ and $\approx 30\dami{4}~\mu{\rm s} $  for our intra-city and inter-city experiments} 
and also, importantly, with the \dami{times required for classical cross-checking schemes to complete a token transaction, which would be $\approx 28~\mu{\rm s}$ in our intracity setup using our 2,766 m long optical fibre link and $\approx 344~\mu{\rm s}$ in our intercity scenario using ideal light speed communication through free space with nodes separated by $51.6$ km, and which would be $\approx 7~\mu{\rm s}$ for the 641 m long optical fibre link  used in the experiment of Ref.~\cite{SERGTBW23} -- or even shorter if the classical cross-checking scheme used ideal light speed communication through free space.}
Ref. ~\cite{SERGTBW23} suggests that with sufficiently good technology their token generation time could be reduced to the order of one second.  Our transaction times are still far shorter than this target, which is also considerably longer than the classical cross-checking time required even for a global network. For comparison, great circle free space communication between antipodal Earth points would take about 67 ms.   

It may be challenging to extend the scheme of Ref.~\cite{SERGTBW23} to networks with multiple transaction nodes and longer distances between nodes. One major challenge is overcoming quantum channel losses. The experiment of Ref.~\cite{SERGTBW23} implemented quantum communication through a 641 metres optical fibre with a total loss of $\approx 22\%$, of which $\approx 7\%$ corresponded to the fibre. With their reported error rate of $\approx 3.3\%$, their results are just within the secure region for the simplest case of two merchants, as illustrated in Fig. 4 of Ref.~\cite{SERGTBW23}. For realistic use cases with multiple merchants and large scale networks, higher losses would make it harder to obtain security. 
In contrast, quantum channel losses do not scale with network size in our implemented protocol, as all quantum communications may take place long before token use, over a single short dedicated channel with a known low loss rate. The communications between distant nodes in our protocol are purely classical. We analyze and discuss in section 7 below how our implementation can be extrapolated to multi-node global networks and still achieve good security.

If implemented in a scenario in which transaction times and positions are critical, users/clients of Ref.~\cite{SERGTBW23}} would need to rely on clock synchronization and position authentication, using networks which are potentially insecure.   As noted above, these issues potentially affect the correctness of any quantum token scheme, including the present implementation of ours and that of Ref.~\cite{SERGTBW23}, when position and time are relevant.  If the client’s coordinates are incorrect, she may present a token in an incorrect region, i.e. one where no merchant or bank agent is present, or in a valid but unintended region.  This means she may spend her token to acquire resources in the wrong space-time region, which may be to her
disadvantage (for example, if she has a time- and location-dependent trading strategy).   It also affects the privacy of any quantum token scheme when position and time are relevant.   The client may present quantum money tokens in a region in the causal past of her intended region, which would allow the spoofer (e.g. the bank, or a third party) to learn her intended presentation region.   Alternatively, she may present tokens in a region space-like separated from her intended region, which allows the spoofer exploitable information about her trading strategy in regions where she requires such information to be unavailable.   Again, this may be disadvantageous when her trading strategy and those of adversaries (such as the spoofer) are time- and location-dependent and utilize all locally available information.   We emphasize again that secure position authentication and time synchronization are solvable problems – assuming that some data, some locations and a reference clock can be trusted – although they are not experimentally addressed in our present implementation,  nor in Ref.~\cite{SERGTBW23}, nor indeed, as far as we are aware, in any implementation to date in the field of relativistic quantum cryptography (see e.g.~\cite{LKBHTKGWZ13,LCCLWCLLSLZZCPZCP14,LKBHTWZ15,VMHBBZ16,ABCDHSZ21}). 
Of course, in scenarios in which transaction times and positions are not important, users of our scheme can label possible recipients of token presentations in ways that are not related to any purported space time coordinates.
In this case they need not rely on secure clock synchronisation or position authentication.

Neither the present scheme nor that of Ref.~\cite{SERGTBW23} require space-like separated presentation regions or other space-time constraints to ensure unforgeability. The reason we impose these constraints in the present implementation is to guarantee that the implementation demonstrates an advantage over classical alternatives. 

Our scheme assumes secure and authenticated classical
channels among agents of the same party.  These can be straightforwardly and efficiently implemented with 
pre-distributed secret keys. Our scheme can also be straightforwardly adapted to prevent any third
party from impersonating the user or the bank in any of the communications between the bank and the
user (using previously distributed secret keys, for instance). It can also be straightforwardly adapted
to prevent any third party from impersonating the bank in the quantum communication, which occurs
only once and is from the bank to the user. For example, after transmitting the quantum states to
the user in the quantum token generation stage, the bank can indicate the states prepared, the bases
used for preparation, and the labels for a random subset, using a secure and authenticated classical channel (implemented with previously distributed secret keys, for instance), which effectively allows
the user to authenticate that the quantum communication was from the bank. The quantum channel
does not need to be assumed secure, because the security properties of user privacy and unforgeability hold straightforwardly as claimed even if some third party is interfering in the quantum channel.
Since our scheme satisfies user privacy against a dishonest bank, it also trivially satisfies user privacy
against a third party trying to impersonate the bank. Since our scheme satisfies unforgeability against
a dishonest user, it also trivially satisfies unforgeability against a third party trying to impersonate
the user.

\section{Security proof}
\label{sec:security}

In Ref.~\cite{KLPGR22}, two quantum token schemes were presented, called $\mathcal{QT}_1$ and $\mathcal{QT}_2$. Here we provide security analyses and proofs for these schemes that improve the results of Ref. ~\cite{KLPGR22} in the following ways.

First, the unforgeability proof of Ref.~\cite{KLPGR22} used the parameter $\theta$, which was defined as an upper bound on the uncertainty angle on the Bloch sphere for each of the two qubit bases prepared by Bob with respect to the target computational and Hadamard bases. This definition of $\theta$ was used in the proof of theorem 1 of Ref.~\cite{KLPGR22} to guarantee that the angle in the Bloch sphere between the two bases prepared by Bob was within the range $\bigl[\frac{\pi}{2}-2\theta,\frac{\pi}{2}+2\theta\bigr]$.

Crucially, our unforgeability proof here improves the one of Ref.~\cite{KLPGR22} by allowing Bob's prepared qubit states to deviate independently from the BB84 states up to a small uncertainty angle $\theta$ on the Bloch sphere, instead of making the assumption that they belong to two qubit orthonormal bases. Furthermore, our proof here also allows for a small probability $P_\theta>0$ that the prepared qubit states deviate from a BB84 state by more than an angle $\theta$ on the Bloch sphere.


With these changes, we define
\begin{equation}\label{noqubtheta}
P_{\text{noqub},\theta}=1-(1-P_\text{noqub})(1-P_\theta) \, ,
\end{equation}
which is an upper bound on the probability that a quantum state that Bob sends Alice has dimension greater than two (corresponding to more than one qubit) or that it deviates from the intended BB84 state by an angle greater than $\theta$ on the Bloch sphere. In our updated unforgeability proof we assume, very conservatively, that an arbitrarily powerful dishonest Alice can perfectly read the bits that Bob encodes in quantum states with dimension greater than two (in multiple qubit states) and in qubit states that deviate from BB84 states by an angle greater than $\theta$.   \adr{In practice the information she can extract from these inrequent ``rogue" pulses will generally be far more limited.}   


Second, the robustness, correctness and unforgeability guarantees proved here in lemmas 2, 3 and theorem 1 use tighter Chernoff bounds than those used in lemmas 2, 3 and theorem 1 of Ref.~\cite{KLPGR22}. These Chernoff bounds are stated and proved below.

\adr{Although the experimental implementation reported in this paper is of scheme $\mathcal{QT}_2$ (with minor variations that do not affect the security analysis)
and not $\mathcal{QT}_1$, the unforgeability proof given here applies to both these quantum token schemes, originally defined in Ref.~\cite{KLPGR22}.  
We believe this broad scope proof will be helpful for future analyses and implementations of both schemes. }

\setcounter{theorem}{0}

Our unforgeability proof uses the maximum confidence measurement of the following quantum state discrimination task.

\begin{definition}
\label{Pbound}
    Consider the following quantum state discrimination problem. For $k\in\Omega_\text{qub}$, we define $\rho_1^k=\rho_{00}^k$, $\rho_2^k=\rho_{01}^k$, $\rho_3^k=\rho_{10}^k$, $\rho_4^k=\rho_{11}^k$, $q_1^k=P^k_\text{PS}(0)P^k_\text{PB}(0)$, $q_2^k=P^k_\text{PS}(0)P^k_\text{PB}(1)$, $q_3^k=P^k_\text{PS}(1)P^k_\text{PB}(0)$, $q_4^k=P^k_\text{PS}(1)P^k_\text{PB}(1)$, and 
    \begin{equation}
    \label{BB1}
r^k_i=\frac{q^k_i+q^k_{i+1}}{2} \, , \qquad\chi^k_i =\frac{q^k_i\rho^k_i+q^k_{i+1}\rho^k_{i+1}}{q^k_i+q^k_{i+1}} \, , \qquad \rho^k=\sum_{i=1}^4 r^k_i \chi^k_i\, ,
    \end{equation}
for all $i\in[4]$. Let $P_{\text{MC}}(\chi^k_j)$ be the maximum confidence measurement that the received state was $\chi^k_j$ when Alice \adr{ is distinguishing states from the 
ensemble $\{ \chi_j^k , r_i^k \}$ and her outcome is $j\in [4]$ \cite{CABGJ06}. 
This} maximum is taken over all positive operators $Q$ acting on a two dimensional Hilbert space. That is, we have
\begin{equation}
\label{BB2}
 P_{\text{MC}}(\chi^k_j)=\max_{Q\geq 0}\frac{r^k_jTr[Q\chi^k_j]}{Tr[Q\rho^k]} \, .   
\end{equation}
\end{definition}

\begin{theorem}\label{theorem1rep}
Suppose that the following constraints hold:
\adr{
\begin{eqnarray}\label{conditions}
\max_{j\in[4],k\in \Omega_\text{qub}} \dami{2} P_{\text{MC}}(\chi^k_j) & < & 1 \, , \nonumber \\
 N\gamma_\text{det}&\leq &  n \leq N\, ,\nonumber\\
0&<&P_{\text{noqub},\theta}<\nu_\text{unf}<\gamma_\text{det}\Bigl(1-\frac{\gamma_\text{err}}{1-P_\text{bound}}\Bigr)\, ,
\end{eqnarray}
for predetermined $\gamma_\text{det}\in(0,1]$ and $\gamma_\text{err}\in[0,1)$ and for some $\nu_\text{unf}\in(0,1)$, where $n=\lvert \Lambda\rvert$, and where $P_\text{bound}$ satisfies
\begin{equation}
\label{BB3}
\max_{j\in[4],k\in \Omega_\text{qub}} \dami{2} P_{\text{MC}}(\chi^k_j)\leq P_{\text{bound}} < 1\, .
\end{equation}
In the case} that losses are not reported we take $\gamma_\text{det}=1$ and $n=N$.
 Then the quantum token schemes $\mathcal{QT}_1$ and $\mathcal{QT}_2$ are $\epsilon_{\text{unf}}-$unforgeable with
\begin{eqnarray}
\label{unforgeability}
\epsilon_{\text{unf}} &=& \sum_{l=0}^{\lfloor N(1-\nu_\text{unf}) \rfloor}\left(\begin{smallmatrix}N\\ l\end{smallmatrix}\right)(1-P_{\text{noqub},\theta})^{l}(P_{\text{noqub},\theta})^{N-l}\, \nonumber\\
&&\quad + \sum_{l=0}^{\lfloor n\gamma_\text{err} \rfloor}\left(\begin{smallmatrix}n- \lfloor N\nu_\text{unf} \rfloor\\ l\end{smallmatrix}\right)(1-P_{\text{bound}})^{l}(P_{\text{bound}})^{n- \lfloor N\nu_\text{unf} \rfloor-l} \, .
\end{eqnarray}
\end{theorem}

\adr{Note that the conditions (\ref{conditions}) \dami{and (\ref{BB3}}) imply that the bound (\ref{unforgeability}) decreases exponentially with $N$.}

\subsection{Mathematical preliminaries}

For any finite dimensional Hilbert space $\mathcal{H}$, we define the sets of linear operators, positive semi-definite operators and quantum density matrices acting on $\mathcal{H}$, respectively by $\mathcal{L}(\mathcal{H})$, $\mathcal{P}(\mathcal{H})$ and $\mathcal{D}(\mathcal{H})$.

\setcounter{lemma}{3}
\begin{lemma}
\label{lemma1}
Let $\mathcal{H}_A$ and $\mathcal{H}_B$ be arbitrary finite dimensional Hilbert spaces. Let $S, X, Y$ be arbitrary finite non-empty sets.
Let $\{q_i \}_{i\in S}$ and $\{p_x \}_{x\in X}$ be probability distributions.  Suppose that Alice receives the quantum state $\rho_i\otimes \sigma_x$ with probability $q_ip_x$, where $\rho_i\in \mathcal{D}(\mathcal{H}_A)$  and $\sigma_x\in \mathcal{D}(\mathcal{H}_B)$, for all $i\in S$ and $x\in X$. Suppose that Alice plays a quantum state discrimination task obtaining a measurement outcome $j\in S$, guessing that she received the state $\rho_j\in \mathcal{D}(\mathcal{H}_A)$, and that she also obtains a measurement outcome $y \in Y$. The probability $P_{xy}$ that Alice discriminates the quantum state successfully from the ensemble $\{q_i,\rho_i\}_{i\in S}$, conditioned on the extra state being $\sigma_x$ and her extra outcome being $y$ satisfies
\begin{equation}
\label{C1}
   P_{xy}\leq \max_{j\in S}P_{\text{MC}}(\rho_j)\, ,
\end{equation}
for all $x\in X$ and all $y\in Y$,
where
\begin{equation}
 P_{\text{MC}}(\rho_j)=\max_{Q\in \mathcal{P}(\mathcal{H}_A)}\frac{q_jTr[Q\rho_j]}{Tr[Q\rho]}   
\end{equation}
is the maximum confidence measurement that the received state was $\rho_j$ when Alice's outcome is $j\in S$ \cite{CABGJ06},
and where
\begin{equation}
\rho=\sum_{i\in S}q_i \rho_i \, .   
\end{equation}
\end{lemma}
\begin{proof}
Consider an arbitrary strategy by Alice. Alice introduces an ancilla of arbitrary finite dimensional Hilbert space $\mathcal{H}_C$ in a pure state $\lvert \xi\rangle$ and applies a projective measurement $\{\pi_{j,y}\}_{j\in S,y\in Y}$ on the whole Hilbert space $\mathcal{H}=\mathcal{H}_A\otimes \mathcal{H}_B\otimes \mathcal{H}_C$. We define
\begin{equation}
    \label{D2}
    \phi_x=\sigma_x \otimes \lvert \xi\rangle\langle\xi \rvert \, , 
\end{equation}
   and
   \begin{equation}
    Q_{i,y}^x= Tr_{BC} [\pi_{i,y}(\mathds{1}_A \otimes \phi_x)] \, , 
\end{equation}
where $Tr_{BC}$ denotes tracing over the Hilbert space $\mathcal{H}_B\otimes \mathcal{H}_C$. We note that  $Q_{i,y}^x\in \mathcal{P}(\mathcal{H}_A)$ for all $i\in S$, $x\in X$ and $y\in Y$. 
We have that
\begin{equation}
     Tr[\pi_{i,y}(\omega\otimes \phi_x)]=Tr[Q_{i,y}^x \omega] \, , 
\end{equation}
for any $\omega\in \mathcal{L}(\mathcal{H}_A)$.

Let $P(\rho_i, \sigma_x, g_j, e_y)$ denote the probability distribution for the state received by Alice being $\rho_i\otimes \sigma_x$ and her measurement outcomes being $j$ and $y$, for all $i,j\in S$, $x\in X$ and $y\in Y$. We note that $P(\rho_j\vert \sigma_x)=q_j$ since $P(\rho_j, \sigma_x)=q_jp_x$ for all $j\in S$ and $x\in X$. Using Bayes' theorem, the probability $P_{xy}$ that Alice discriminates the quantum state successfully from the ensemble $\{q_i,\rho_i\}_{i\in S}$, conditioned on the extra state being $\sigma_x$ and her extra outcome being $y$, satisfies
\begin{eqnarray}
    \label{D1}
    P_{xy}&=&\sum_{i\in S}P(\rho_i, g_i\vert e_y,\sigma_x)\nonumber\\
    &=&\sum_{i\in S}\frac{P(\rho_i\vert g_i,e_y, \sigma_x)P(g_i,e_y\vert \sigma_x)}{P(e_y\vert \sigma_x)}\nonumber\\
    &\leq&\max_{j\in S}\{P(\rho_j\vert g_j,e_y, \sigma_x)\}\sum_{i\in S}\frac{P(g_i,e_y\vert \sigma_x)}{P(e_y\vert \sigma_x)}\nonumber\\
    &=&\max_{j\in S}P(\rho_j\vert g_j,e_y, \sigma_x)\nonumber\\
 &=&\max_{j\in S}\frac{P(g_j,e_y\vert \rho_j, \sigma_x)P(\rho_j\vert \sigma_x)}{P(g_j,e_y\vert \sigma_x)}\nonumber\\
        &=&\max_{j\in S}\frac{q_jTr[Q_{jy}^x\rho_j]}{Tr[Q_{jy}^x\rho]}\nonumber\\
        &\leq&\max_{j\in S, Q\in \mathcal{P}(\mathcal{H}_A)}\frac{q_jTr[Q\rho_j]}{Tr[Q\rho]}\nonumber\\
        &=&\max_{j\in S}P_{\text{MC}}(\rho_j)\, .
\end{eqnarray}
\end{proof}

\begin{definition}
    \label{definition1}
    We define the following task. Alice receives a finite dimensional quantum state $\rho_{\textbf{i}}=\otimes _{k=1}^N \rho_{i_k}^k$ encoding an input $\textbf{i}=(i_1,\ldots,i_N)$ with probability $P_\text{in}(\textbf{i})=\prod_{k=1}^N P_\text{in}^k(i_k)$, where $i_k\in I_k$ for some finite non-empty set $I_k$, and where $\{P_\text{in}^k(i_k)\}_{i_k\in I_k}$ is a probability distribution, for all $k\in[N]$, and for some $N\geq 1$. Alice obtains an output $\textbf{o}=(o_1,\ldots,o_N)$, where $o_k\in O_k$ for some finite non-empty set $O_k$, for all $k\in[N]$. Alice wins the $k$th round with probability given by a function $F_\text{win}^k(i_k,o_k)\in [0,1]$, for $i_k\in I_k$ and $o_k\in O_k$, for all $k\in[N]$. Alice wins the task if she loses (i.e., does not win) in no more than $n$ rounds, for some $n\in\{0,1,2,\ldots,N\}$.
\end{definition}

\begin{observation}
     \label{observation1}
To have a nontrivial task, we must have $F_\text{win}^k(i_k,o_k)>0$ for at least one $(i_k,o_k)\in I_k\times O_k$ and for at least one $k\in[N]$, that is, Alice must be able to win the task with non-zero probability. Additionally, we note that Alice trivially wins the task with unit probability if $n=N$, as in this case she is allowed to lose in all rounds; we allow this trivial case in the definition above as this simplifies the statements for the results below and their proofs.
\end{observation}
  
\begin{observation}
        \label{observation2}
The task in which Alice must win a standard quantum state discrimination game in each round $k$ and make no more than $n$ errors out of the $N$ rounds corresponds to the case $O_k=I_k$ and $F_\text{win}^k(i_k,o_k)=\delta_{i_k,o_k}$, for all $k\in[N]$. In the unforgeability proof below for our quantum token schemes we consider a variation of this task in which $O_k=I_k=\{1,2,3,4\}$ and $F_\text{win}^k(i_k,o_k)=1$ if $o_k=\{i_k,i_{k}-1\}$ (where we use the notation $i\pm 4 = i$, for all $i=1,2,3,4$) and $F_\text{win}^k(i_k,o_k)=0$ otherwise, for all $k\in[N]$. We allow the task to be more general in definition~\ref{definition1}, as  lemma~\ref{lemma2} has broader applications in quantum cryptography and quantum information theory. 
\end{observation}

\begin{lemma}
\label{lemma2}
Consider the task of definition~\ref{definition1}. Suppose that Alice's winning probability $P^k_{\text{win}}\bigl(\tilde{\textbf{i}}_k,\tilde{\textbf{o}}_k,o_\text{extra}\bigr)$ in the round $k$ conditioned on an input $\tilde{\textbf{i}}_k=(i_1,\ldots,i_{k-1},i_{k+1},i_{k+2},\ldots,i_N)$ and an output $\tilde{\textbf{o}}_k=(o_1,\ldots,o_{k-1},o_{k+1},o_{k+2},\ldots,o_N)$ \adr{(i.e., conditioned on inputs and outputs for all rounds apart from $k$)}, and on an extra output $o_\text{extra}$, satisfies
\begin{equation}
    \label{B0}P^k_{\text{win}}\bigl(\tilde{\textbf{i}}_k,\tilde{\textbf{o}}_k,o_\text{extra}\bigr)\leq P_\text{bound}^k\, ,
\end{equation}
for all $i_j\in I_j$, $o_j\in O_j$, $j\in[N]\setminus\{k\}$ and $k\in[N]$. Then, Alice's probability $P_{\text{win}}(n,N\vert o_\text{extra})$ to win the task conditioned on any extra output $o_\text{extra}$ satisfies
\begin{equation}
\label{B1}
   P_{\text{win}}(n,N\vert o_\text{extra}) \leq P_{\text{bound}}^{\text{coins}}(n,N)\leq \sum_{l=0}^{n}\left(\begin{smallmatrix}N\\ l\end{smallmatrix}\right)(1-P_{\text{bound}})^{l}(P_{\text{bound}})^{N-l} \, , 
\end{equation}
where $P_{\text{bound}}^{\text{coins}}(n,N)$ is the probability of having no more than $n$ errors in $N$ independent coin tosses with success probabilities $P_\text{bound}^1,\ldots, P_\text{bound}^N$, 
and where 
\begin{equation}
\label{B2}
P_{\text{bound}}^k\leq P_{\text{bound}} \, , 
\end{equation}
for all $k\in [N]$, $n\in\{0,1,\ldots,N\}$ and $N\geq 1$.
\end{lemma}


\begin{proof}
The second inequality in (\ref{B1}) follows straightforwardly by noting that $P_{\text{bound}}^{\text{coins}}(n,N)$ is maximized by maximizing each of the success probabilities $P_\text{bound}^1,\ldots, P_\text{bound}^N$ in the independent coin tosses, and by using (\ref{B2}). 

Below we show the first inequality in (\ref{B1}). Let $P(r_1,r_2,\ldots,r_N\vert o_\text{extra})$ denote Alice's probability distribution of her results, conditioned on her extra output $o_\text{extra}$, where $r_k\in\{\text{w}_k,\text{l}_k\}$ with $o_k=\text{w}_k$ denoting that Alice wins (succeeds) in the $k$th round and with $o_k=\text{l}_k$ denoting that Alice loses (makes an error) in the $k$th round, for all $k\in[N]$. 

We show the case $n=0$ and $N\geq 1$. We have
   \begin{equation}
   \label{B3}
     P_{\text{win}}(0,N\vert o_\text{extra})=P(\text{w}_1\vert o_\text{extra})P(\text{w}_2\vert\text{w}_1, o_\text{extra})\cdots P(\text{w}_N\vert\text{w}_1,\ldots,\text{w}_{N-1},o_\text{extra}).
\end{equation}
From (\ref{B0}), Alice's probability to win the $k$th round conditioned on any inputs and outputs for other rounds and on any extra output $o_\text{extra}$ is upper bounded by $P_\text{bound}^k$, for all $k\in [N]$. Thus, we have
\begin{equation}
   \label{B4}
     P_{\text{win}}(0,N\vert o_\text{extra})\leq \prod_{k=1}^N P_\text{bound}^k=P_{\text{bound}}^{\text{coins}}(0,N)\, ,
\end{equation}
as claimed.

As mentioned in observation~\ref{observation1}, the case $n=N$ for $N\geq 1$ is trivial, as in this case Alice is allowed any number of errors, achieving $P_{\text{win}}(N,N\vert o_\text{extra})=1$, which is consistent with (\ref{B1}), as $P_{\text{bound}}^{\text{coins}}(N,N)=1$.


We show the case $n=1$ and $N\geq 2$. We have
\begin{eqnarray}
    \label{B7}
    P_{\text{win}}(1,N\vert o_\text{extra})&=&\text{Pr}[\text{1~error~in~first~$N-1$~rounds}\vert o_\text{extra}]\times\nonumber\\
    &&\quad  P(\text{w}_{N}\vert \text{1~error~in~first~$N-1$~rounds},o_\text{extra})\nonumber\\
    &&\qquad+P_{\text{win}}(0,N-1\vert o_\text{extra})\nonumber\\
    &\leq& P_{\text{win}}(0,N-1\vert o_\text{extra})\times\nonumber\\
    &&\quad [1-P(\text{w}_{N}\vert \text{1~error~in~first~$N-1$~rounds},o_\text{extra})]\nonumber\\
    &&\qquad+ P(\text{w}_{N}\vert \text{1~error~in~first~$N-1$~rounds},o_\text{extra})\, , 
\end{eqnarray}
where in the second line we used
\begin{equation}
\label{B7.1}
    \text{Pr}[\text{1~error~in~first~$N-1$~rounds}\vert o_\text{extra}]+P_{\text{win}}(0,N-1\vert o_\text{extra})\leq 1\, .
\end{equation}
Since
\begin{equation}
\label{B7.2}
    1-P(\text{w}_{N}\vert \text{1~error~in~first~$N-1$~rounds},o_\text{extra})\geq 0\, ,
\end{equation}
we see from the second line of (\ref{B7}), that we maximize $P_{\text{win}}(1,N\vert o_\text{extra})$ by maximizing $P_{\text{win}}(0,N-1\vert o_\text{extra})$. From (\ref{B4}), this is upper bounded by the probability of having no errors in $N-1$ independent coin tosses with success probabilities $P_\text{bound}^1,\ldots, P_\text{bound}^{N-1}$. Furthermore, from~(\ref{B0}), we have 
\begin{equation}
\label{B7.3}
P(\text{w}_{N}\vert \text{1~error~in~first~$N-1$~rounds},o_\text{extra})\leq P_\text{bound}^N\, .
\end{equation}
Thus, we see from the first line of (\ref{B7}) and from (\ref{B7.3}) that
\begin{eqnarray}
    \label{B8}
    P_{\text{win}}(1,N\vert o_\text{extra})&\leq& P_{\text{bound}}^{\text{coins}}(0,N-1)\nonumber\\
    &&\quad+P_{\text{bound}}^{\text{coins}}[\text{1~error~in~first~$N-1$~rounds}]P_\text{bound}^N\nonumber\\
    &=&P_{\text{bound}}^{\text{coins}}(1,N)\, ,
\end{eqnarray}
as claimed.

To complete the proof we show the case $n=m\geq 2$ and $N\geq 3$ by induction. We suppose it holds for $n=m-1$ and show that it holds for $n=m$, for any $N\geq 3$ and any $2\leq m\leq N$. We have
\begin{eqnarray}
    \label{B9}
    P_{\text{win}}(m,N\vert o_\text{extra})&=&\text{Pr}[\text{$m$~errors~in~first~$N-1$~rounds}\vert o_\text{extra}]\times\nonumber\\
    &&\quad \times P(\text{w}_{N}\vert \text{m~errors~in~first~$N-1$~rounds},o_\text{extra})\nonumber\\
    &&\qquad+ P_{\text{win}}(m-1,N-1\vert o_\text{extra})\nonumber\\
    &\leq& P_{\text{win}}(m-1,N-1\vert o_\text{extra})\times\nonumber\\
    &&\quad \!\!\!\times [1-P(\text{w}_{N}\vert \text{$m$~errors~in~first~$N-1$~rounds},o_\text{extra})]\nonumber\\
    &&\quad\!+ P(\text{w}_{N}\vert \text{$m$~errors~in~first~$N-1$~rounds},o_\text{extra})\, ,
\end{eqnarray}
where in the second line we used
\begin{equation}
\label{B9.1}
P_{\text{win}}(m-1,N-1\vert o_\text{extra})+\text{Pr}[\text{$m$~errors~in~first~$N-1$~rounds}\vert o_\text{extra}]\leq 1\, .
    \end{equation}
Since
\begin{equation}
\label{B9.2}
    [1-P(\text{w}_{N}\vert \text{$m$~errors~in~first~$N-1$~rounds},o_\text{extra})]\geq 0\, ,
\end{equation}
we see from the second line of (\ref{B9}), that we maximize $P_{\text{win}}(m,N\vert o_\text{extra})$ by
maximizing $P_{\text{win}}(m-1,N-1\vert o_\text{extra})$. By assumption, this is upper bounded by the probability of having no more than $m-1$ errors in $N-1$ independent coin tosses with success probabilities $P_\text{bound}^1,\ldots, P_\text{bound}^{N-1}$. Furthermore, from~(\ref{B0}), we have 
\begin{equation}
\label{B9.3}
    P(\text{w}_{N}\vert \text{$m$~errors~in~first~$N-1$~rounds},o_\text{extra})\leq P_\text{bound}^N\, .
\end{equation}
Thus, we see from the first line of (\ref{B9}) and from (\ref{B9.3}) that
\begin{eqnarray}
    \label{B10}
    P_{\text{win}}(m,N\vert o_\text{extra})&=&P_{\text{bound}}^{\text{coins}}[\text{$m$~errors~in~first~$N-1$~rounds}]P_\text{bound}^{N}\nonumber\\
    &&\quad+ P_{\text{bound}}^{\text{coins}}(m-1,N-1)\nonumber\\
    &=&P_{\text{bound}}^{\text{coins}}(m,N)\, ,
\end{eqnarray}
as claimed.
\end{proof}

\begin{proposition}\label{proposition3}
Let $X_1,X_2,\ldots,X_N$ be identical independent random variables taking values $X_k\in \{0,1\}$, for all $k\in[N]$. Let $X=\sum_{k=1}^N {X}_k$, and let the expectation value of $X_k$ be $E(X_k)=p$ for all $k\in[N]$. 
Then we have the Chernoff bound 
\begin{equation}
\label{firstChernoffbound}
	\text{Pr}[X\leq (1-\epsilon)Np]\leq \left(\frac{1}{1-\epsilon}\right)^{Np(1-\epsilon)}\left(\frac{1-p}{1-p(1-\epsilon)}\right)^{N\bigl(1-p(1-\epsilon)\bigr)} \, , 
\end{equation}
for $0< \epsilon<1$ and $0<p\leq1$.  
We also have the Chernoff bound 
\begin{equation}
\label{secondChernoffbound}
	\text{Pr}[X\geq (1+\epsilon)Np]\leq \left(\frac{1}{1+\epsilon}\right)^{Np(1+\epsilon)}\left(\frac{1-p}{1-p(1+\epsilon)}\right)^{N\bigl(1-p(1+\epsilon)\bigr)} \, , 
\end{equation}
for $\epsilon>0$ and $0<p(1+\epsilon)<1$. 
\end{proposition}
\begin{proof}
We denote the expectation value of a random variable $x$ by $E(x)$. We begin by proving 
(\ref{secondChernoffbound}). Let $\epsilon>0$ and $0<p(1+\epsilon)<1$, which implies $p>0$. For $X_k\in\{0,1\}$ and $E(X_k)=p$, we get that $\text{Pr}[X_k=1]=p$. By Markov’s inequality, as $p, \epsilon > 0$, and for $t > 0$, we have
\begin{eqnarray}
	\text{Pr}[X\geq (1+\epsilon)Np]&=&\text{Pr}\left[e^{tX}\geq e^{t(1+\epsilon)Np}\right]\nonumber\\
 &\leq&\frac{E[e^{tX}]}{e^{t(1+\epsilon)Np}}\nonumber\\
 &=&\frac{ \sum_{l=0}^N \left(\begin{smallmatrix}l\\N\end{smallmatrix}\right)p^l(1-p)^{N-l}e^{tl}}{e^{t(1+\epsilon)Np}}\nonumber\\
  &=&\frac{ \sum_{l=0}^N \left(\begin{smallmatrix}l\\N\end{smallmatrix}\right)(pe^t)^l(1-p)^{N-l}}{e^{t(1+\epsilon)Np}}\nonumber\\
 &=&\frac{(1-p+pe^t)^N}{e^{t(1+\epsilon)Np}} \,, 
\end{eqnarray}
 where in the last line we used the binomial theorem. By taking the infimum over $t>0$, we have
\begin{equation}\label{bound1}
	\text{Pr}[X\geq (1+\epsilon)Np]\leq\inf_{t>0}\frac{(1-p+pe^t)^N}{e^{t(1+\epsilon)Np}} \, .
\end{equation}
Next, we define
\begin{equation}
	f(t) = \frac{1-p+pe^t}{e^{t(1+\epsilon)p}} \, .
\end{equation}
Then, we have
\begin{equation}
	\frac{df}{dt} = pe^{-t(1+\epsilon)p}\bigl[e^t\bigl(1-p(1+\epsilon)\bigr)-(1+\epsilon)(1-p)\bigr]
 \end{equation}
 and
\begin{equation}
	\frac{d^2f}{dt^2} = pe^t [1-p(1+\epsilon)]^2e^{-t(1+\epsilon)p}>0\,,
\end{equation}
as $p(1+\epsilon)<1$ and $p>0$. Thus, $f$ is strictly convex and attains its infimum when \dami{$\frac{df}{dt}=0$}. This occurs when
\begin{equation}
	t=\ln\left(\frac{(1-p)(1+\epsilon)}{1-p(1+\epsilon)}\right)\,,
\end{equation}
which satisfies $t>0$, for $\epsilon>0$ and $0<p(1+\epsilon)<1$. Substituting this back into (\ref{bound1}) gives
\begin{equation}
	\text{Pr}[X\geq (1+\epsilon)Np]\leq\left(\frac{1}{1+\epsilon}\right)^{Np(1+\epsilon)}\left(\frac{1-p}{1-p(1+\epsilon)}\right)^{N\bigl(1-p(1+\epsilon)\bigr)} \,.
\end{equation}

Now we show (\ref{firstChernoffbound}). Let $0< \epsilon<1$. We first note that the case $p=1$ is trivially satisfied, as both sides of (\ref{firstChernoffbound}) are zero in this case; $\text{Pr}[X\leq (1-\epsilon)Np]= \text{Pr}[X\leq (1-\epsilon)N]=0$ because $X=N$ and $(1-\epsilon)N<N$; the right hand side of (\ref{firstChernoffbound}) is also zero because $1-p=0$. Now we consider the case $0<p<1$. We define $Y_k=1-X_k$ and $Y=\sum_{k=1}^N {Y}_k=N-X$. We note that $Y_k\in\{0,1\}$, $Y_k$ are identical independent random variables, and $E(Y_k)=1-p$. By Markov's inequality, as $0<\epsilon<1$ and $0<p<1$, for $t>0$, we have
\begin{eqnarray}
	\text{Pr}[Y\geq (1-p(1-\epsilon))N]&=&\text{Pr}\left[e^{tY}\geq e^{t(1-p(1-\epsilon))N}\right]\nonumber\\
  &\leq&\frac{E[e^{tY}]}{e^{t(1-p(1-\epsilon))N}}\nonumber\\
  &=&\frac{\sum_{l=0}^N \left(\begin{smallmatrix}l\\N\end{smallmatrix}\right) p^l(1-p)^{N-l}e^{t(N-l)}}{e^{t(1-p(1-\epsilon))N}}\nonumber\\
 &=& \frac{\sum_{l=0}^N \left(\begin{smallmatrix}l\\N\end{smallmatrix}\right) p^l[(1-p)e^t]^{N-l}}{e^{t(1-p(1-\epsilon))N}}\nonumber\\
  &=&\frac{(p+(1-p)e^t)^N}{e^{t(1-p(1-\epsilon))N}} \, , 
\end{eqnarray}
where in the last line we used the binomial theorem. By taking the infimum over $t>0$, we obtain
\begin{equation}\label{bound2}
	\text{Pr}[Y\geq (1-p(1-\epsilon))N]\leq\inf_{t>0}\frac{(p+(1-p)e^t)^N}{e^{t(1-p(1-\epsilon))N}} \, .
\end{equation}
Next, we define
\begin{equation}
	g(t) = \frac{p+(1-p)e^t}{e^{t(1-p(1-\epsilon))}} \, .
\end{equation}
Then, we have
\begin{equation}
    \frac{dg}{dt}=pe^{-t(1-p(1-\epsilon))}\bigl[e^t(1-p)(1-\epsilon)-\bigl(1-p(1-\epsilon)\bigr)\bigr]
\end{equation}
and
\begin{equation}
	\frac{d^2g}{dt^2} = pe^{-t(1-p(1-\epsilon))}\Bigl[e^tp(1-p)(1-\epsilon)^2 +\bigl(1-p(1-\epsilon)\bigr)^2 \Bigr]>0,
\end{equation}
since $0<p<1$, $0<\epsilon<1$ and $t>0$. Thus, $g$ is strictly convex and attains its infimum when \dami{$\frac{dg}{dt}=0$}. This occurs when
\begin{equation}
	t=\ln\left(\frac{1-p(1-\epsilon)}{(1-p)(1-\epsilon)}\right)\,,
\end{equation}
which satisfies $t>0$, for $0<\epsilon<1$ and $0<p<1$. Substituting this back into (\ref{bound2}), and using $0<\epsilon<1$ and $0<p<1$, gives
\begin{equation}
	\text{Pr}[X\leq (1-\epsilon)Np]=\text{Pr}[Y\geq (1-p(1-\epsilon))N]\leq\left(\frac{1}{1-\epsilon}\right)^{Np(1-\epsilon)}\left(\frac{1-p}{1-p(1-\epsilon)}\right)^{N(1-p(1-\epsilon))} \, .
\end{equation}

\end{proof}

\begin{proposition}
\label{newproposition1}
Consider random variables $X_k\in\{0,1\}$ with $0<\text{Pr}[X_k=1]$, for all $k\in\{1,2,\ldots, N\}\equiv[N]$. We define the random variable $X=\sum_{k=1}^NX_k$, which satisfies $X\in[0,N]$. For any $t\in[N]$, consider random variables $Y_{t,k}\in\{0,1\}$ satisfying $0<\text{Pr}[Y_{t,t}=1]\leq \text{Pr}[X_t=1]$ and $\text{Pr}[Y_{t,k}=1]=\text{Pr}[X_k=1]$ for all $k\neq t$. We define the random variable $Y_t=\sum_{k=1}^NY_{t,k}$, which satisfies $Y_t\in[0,N]$. Then, for all $a\in[0,N]$, we have $\text{Pr}[X<a]\leq \text{Pr}[Y_t<a]$.
\end{proposition}
\begin{proof}

The case $a=0$ is trivial, as $\text{Pr}[X<0]=\text{Pr}[Y_t<0]=0$ since $X,Y_t\in[0,N]$ by definition. Thus, below we assume $0<a\leq N$. We have,
    \begin{equation}
\begin{split}
	\text{Pr}[X<a]=&\text{Pr}\bigg[X<a\bigg|a-1\leq\sum_{k\neq t}X_k<a\bigg]\text{Pr}\bigg[a-1\leq\sum_{k\neq t}X_k <a\bigg]\\
	&+\text{Pr}\bigg[X<a\bigg|\sum_{k\neq t} X_k<a-1\bigg]\text{Pr}\bigg[\sum_{k\neq t} X_k <a-1\bigg]\\
	=&\text{Pr}[X_t=0]\text{Pr}\bigg[a-1\leq\sum_{k\neq t}X_k<a\bigg]+\text{Pr}\bigg[\sum_{k\neq t} X_k<a-1\bigg]\\
	=&\text{Pr}[X_t=0]\text{Pr}\bigg[a-1\leq\sum_{k\neq t}Y_{t,k}<a\bigg]+\text{Pr}\bigg[\sum_{k\neq t}Y_{t,k}<a-1\bigg]\\
	\leq&\text{Pr}[Y_{t,t}=0]\text{Pr}\bigg[a-1\leq\sum_{k\neq t} Y_{t,k}<a\bigg]+\text{Pr}\bigg[\sum_{k\neq t} Y_{t,k}<a-1\bigg]\\
	=&\text{Pr}\bigg[Y<a\bigg|a-1\leq\sum_{k\neq t} Y_{t,k}<a\bigg]\text{Pr}\bigg[a-1\leq\sum_{k\neq t}Y_{t,k}<a\bigg]\\
	&+\text{Pr}\bigg[Y<a\bigg|\sum_{k\neq t}Y_k<a-1\bigg]\text{Pr}\bigg[\sum_{k\neq t}Y_k<a-1\bigg]\\
	= &\text{Pr}[Y<a]\,.
\end{split}
\end{equation}
\end{proof}

\begin{proposition}
\label{propositionprobcol}
Consider random variables $W_k,Z_k\in\{0,1\}$ with $0<\text{Pr}[Z_k=1]\leq\text{Pr}[W_k=1]$, for all $k\in\{1,2,\ldots, N\}\equiv[N]$. We define the random variables $W=\sum_{k=1}^NW_k$ and $Z=\sum_{k=1}^NZ_k$, which satisfy $W,Z\in[0,N]$. Then, for all $a\in[0,N]$, we have $\text{Pr}[W<a]\leq \text{Pr}[Z<a]$.
\end{proposition}
\begin{proof}
	Define $W^{(0)}=W$ and $W_k^{(0)}=W_k$. Starting from $t=1$:
\begin{enumerate}
	\item Set $W_k^{(t)}=W_k^{(t-1)}$ for $k\neq t$ and $W_t^{(t)}=Z_t$. Set $W^{(t)}=\sum_{k=1}^N W_k^{(t)}$.
	\item By proposition~\ref{newproposition1} with $X=W^{(t-1)}$ and $Y_t=W^{(t)}$, $\forall a\in [0,N]$, $\text{Pr}[W^{(t-1)}<a]\leq \text{Pr}[W^{(t)}<a]$.
	\item If $t<N$, add 1 to $t$ and return to step 1. If $t=N$, then $W^{(N)}=Z$, concluding the process.
\end{enumerate}
This gives the result
\begin{equation}
	\text{Pr}[W<a]=\text{Pr}[W^{(0)}<a]\leq \text{Pr}[W^{(1)}<a]\leq ... \leq \text{Pr}[W^{(N)}<a] = \text{Pr}[Z<a]\,.
\end{equation}
\end{proof}

\subsection{Proof of robustness and correctness}

\setcounter{lemma}{1}
\begin{lemma}\label{lemma2rep}
If
\begin{equation}\label{rob1repeated}
	0<\gamma_\text{det} <P_{\text{det}} \, , 
\end{equation}
then the quantum token schemes $\mathcal{QT}_1$ and $\mathcal{QT}_2$ of Ref.~\cite{KLPGR22} are $\epsilon_{\text{rob}}-$robust with
\begin{equation}\label{rob2repeated}
	\epsilon_{\text{rob}}=\left(\frac{P_{\text{det}}}{\gamma_\text{det}}\right)^{N\gamma_\text{det}}\left(\frac{1-P_{\text{det}}}{1-\gamma_\text{det}}\right)^{N(1-\gamma_\text{det})} \, .
\end{equation}
\end{lemma}

\begin{proof}
The proof is a straightforward adaptation of the proof of lemma 2 of Ref.~\cite{KLPGR22} using the tighter Chernoff bound (\ref{firstChernoffbound}) of proposition \ref{proposition3}. Let $P_{\text{abort}}$ be the probability that Bob aborts the token scheme if Alice and Bob follow the token scheme honestly. 
By definition of the quantum token schemes $\mathcal{QT}_1$ and $\mathcal{QT}_2$ of Ref.~\cite{KLPGR22},
we have
\begin{equation}
\label{aaax1}
P_{\text{abort}}=\text{Pr}[n<\gamma_\text{det} N]\,.
\end{equation}
We note that the expectation value of $n$ is $E(n)=NP_\text{det}$. From (\ref{rob1repeated}), we have that $0<1-\frac{\gamma_\text{det}}{P_\text{det}}<1$. Thus, we obtain from the Chernoff bound (\ref{firstChernoffbound}) of proposition \ref{proposition3} with $p=P_{\text{det}}$ and $\epsilon=1-\frac{\gamma_\text{det}}{P_\text{det}}$ that
 \begin{equation}
 \label{aaax3}
 \text{Pr}[n<\gamma_\text{det} N]\leq\left(\frac{P_{\text{det}}}{\gamma_\text{det}}\right)^{N\gamma_\text{det}}\left(\frac{1-P_{\text{det}}}{1-\gamma_\text{det}}\right)^{N(1-\gamma_\text{det})} \, .
 \end{equation}
It follows from (\ref{aaax1}) and (\ref{aaax3}) that
 \begin{equation}
 \label{aaax4}
 P_{\text{abort}}\leq\epsilon_{\text{rob}} \, , 
 \end{equation}
 with $\epsilon_{\text{rob}}$ given by (\ref{rob2repeated}). 
It follows from (\ref{aaax4}) that the token schemes $\mathcal{QT}_1$ and $\mathcal{QT}_2$ are  $\epsilon_{\text{rob}}$-robust with $\epsilon_{\text{rob}}$ given by (\ref{rob2repeated}).
\end{proof}

\begin{lemma}\label{lemma3rep}
If
\begin{eqnarray}\label{cor1repeated}
	0&<&E<\gamma_\text{err} \, , \nonumber\\
	0&<&\nu_\text{cor}<\frac{P_{\text{det}}(1-2\beta_\text{PB})}{2} \, , \end{eqnarray}
then the quantum token schemes $\mathcal{QT}_1$ and $\mathcal{QT}_2$ of Ref.~\cite{KLPGR22} are $\epsilon_{\text{cor}}-$correct with
\begin{eqnarray}\label{cor2repeated}
	\epsilon_{\text{cor}}&=&\left(\frac{P_{\text{det}}(1-2\beta_\text{PB})}{2\nu_\text{cor}}\right)^{N\nu_\text{cor}}\left(\frac{2-P_{\text{det}}(1-2\beta_\text{PB})}{2-2\nu_\text{cor}}\right)^{N(1-\nu_\text{cor})}\nonumber\\
 &&\qquad +\left(\frac{E}{\gamma_\text{err}}\right)^{N\nu_\text{cor}\gamma_\text{err}}\left(\frac{1-E}{1-\gamma_\text{err}}\right)^{N\nu_\text{cor}(1-\gamma_\text{err})} \, .
\end{eqnarray}
\end{lemma}

\begin{proof}~
The proof is a straightforward adaptation of the proof of lemma 3 of Ref.~\cite{KLPGR22} using the tighter Chernoff bound (\ref{firstChernoffbound}) of proposition \ref{proposition3}. Let $P_{\text{error}}$ be the probability that Bob does not accept Alice's token as valid if Alice and Bob follow the token scheme honestly. 
By definition of the quantum token schemes $\mathcal{QT}_1$ and $\mathcal{QT}_2$ of Ref.~\cite{KLPGR22}, we have
\begin{equation}
\label{aaax1.2}
P_{\text{error}}=\sum_{\lvert \Delta_b\rvert=0}^N P_{\text{error}}(\lvert \Delta_b\rvert)\text{Pr}(\lvert \Delta_b\rvert)\,,
\end{equation}
where
\begin{equation}
\label{aaax1.1}
P_{\text{error}}(\lvert \Delta_b\rvert)=\text{Pr}\Bigl[n_{\text{errors}}>\lvert \Delta_b\rvert\gamma_\text{err}\big\vert \lvert \Delta_b\rvert\Bigr]\,,
\end{equation}
and where $n_{\text{errors}}$ is the number of bit errors in the substring $\mathbf{x}_b$ of the the token $\mathbf{x}$ that Alice presents to Bob at $Q_b$ (the spacetime region $R_b$), compared to the bits of the substring $\mathbf{r}_b$ of $\mathbf{r}$ encoded by Bob.
From 
(\ref{aaax1.2}),
we have
\begin{eqnarray}
\label{aaay3.1}
P_{\text{error}}&=&\sum_{\lvert \Delta_b\rvert<\nu_\text{cor} N}P_{\text{error}}(\lvert \Delta_b\rvert)\text{Pr}(\lvert \Delta_b\rvert)\nonumber\\
&&\quad+ \sum_{\lvert \Delta_b\rvert\geq \nu_\text{cor} N}P_{\text{error}}(\lvert \Delta_b\rvert)\text{Pr}(\lvert \Delta_b\rvert)\nonumber\\
&\leq&\text{Pr}\big[ \lvert \Delta_b\rvert < \nu_\text{cor} N\bigr]\nonumber\\
&&\quad+ \sum_{\lvert \Delta_b\rvert\geq \nu_\text{cor} N}P_{\text{error}}(\lvert \Delta_b\rvert)\text{Pr}(\lvert \Delta_b\rvert)\,.
\end{eqnarray}
We show below that
\begin{equation}
\label{aaay3.2x}
P_{\text{error}}(\lvert \Delta_b\rvert )\leq\left(\frac{E}{\gamma_\text{err}}\right)^{|\Delta_b|\gamma_\text{err}} \left(\frac{1-E}{1-\gamma_\text{err}}\right)^{|\Delta_b|(1-\gamma_\text{err})} \, , 
\end{equation}
and that
\begin{equation}
\label{aaay3.2}
\text{Pr}[\lvert \Delta_b\rvert\!<\!\nu_\text{cor} N]\leq \left(\frac{P_{\text{det}}(1-2\beta_\text{PB})}{2\nu_\text{cor}}\right)^{N\nu_\text{cor}}\left(\frac{2-P_{\text{det}}(1-2\beta_\text{PB})}{2-2\nu_\text{cor}}\right)^{N(1-\nu_\text{cor})} \, .
\end{equation}
From (\ref{aaay3.1}) -- (\ref{aaay3.2}), and showing that the function
\begin{equation}
\label{function}
 F(m)= \left(\frac{E}{\gamma_\text{err}}\right)^{m\gamma_\text{err}} \left(\frac{1-E}{1-\gamma_\text{err}}\right)^{m(1-\gamma_\text{err})}
\end{equation}
decreases with increasing $m=\lvert \Delta_b\rvert$ in the range $[\nu_{\text{cor}}N,N]$, we obtain
\begin{equation}
	\label{neweq2}
	P_\text{error}\leq\epsilon_\text{cor} \, , 
\end{equation}
with $\epsilon_\text{cor}$ given by (\ref{cor2repeated}). Thus, the token schemes $\mathcal{QT}_1$ and $\mathcal{QT}_2$ are $\epsilon_\text{cor}-$correct with $\epsilon_\text{cor}$ given by (\ref{cor2repeated}), as claimed.

We show that $F(m)$ given by (\ref{function}) decreases with increasing $m$ in the range $[\nu_{\text{cor}}N,N]$. We note from (\ref{cor1repeated}) that $0<\nu_{\text{cor}}<1$. We also note that $F(m)=q^m$, where
\begin{equation}
    q=\left(\frac{E}{\gamma_\text{err}}\right)^{\gamma_\text{err}} \left(\frac{1-E}{1-\gamma_\text{err}}\right)^{(1-\gamma_\text{err})} \, .
\end{equation}
Thus, the claim follows by showing that 
\begin{equation}
  \label{qineq}
  0<q<1\,.  
  \end{equation}
We show this as follows. First, we note that
\begin{equation}
    q=\frac{G(E)}{G(\gamma_\text{err})} \, , 
\end{equation}
where
\begin{equation}
       \label{FunctionG}
    G(x)=x^{\gamma_\text{err}}(1-x)^{1-\gamma_\text{err}}
\end{equation}
for $x\in(0,1)$. Thus, (\ref{qineq}) follows if 
\begin{equation}
\label{FunctionG2}
0<G(E)<G(\gamma_\text{err})\,.
\end{equation}
From (\ref{cor1repeated}) and from $\gamma_\text{err}<1$, by definition, we have $0<E<\gamma_\text{err}<1$. Thus, the first inequality in (\ref{FunctionG2}) follows straightforwardly, whereas the second inequality follows if $G(x)$ increases by increasing $x$ in the range $x\in(0,\gamma_\text{err})$. This holds because
\begin{equation}
    \frac{dG}{dx}=\bigl(\gamma_\text{err}(1-x)-(1-\gamma_\text{err})x\bigr)x^{\gamma_\text{err}-1}(1-x)^{-\gamma_\text{err}}>0
\end{equation}
as
\begin{equation}
    \gamma_\text{err}(1-x)>\gamma_\text{err}(1-\gamma_\text{err})>(1-\gamma_\text{err})x\,,
\end{equation}
since $0<x<\gamma_\text{err}<1$.

We show (\ref{aaay3.2x}). We recall that $E = \max_{t,u}\{E_{tu}\}$, where $E_{tu}$ is Alice’s
error rate when Bob prepares states $\lvert \psi_{tu}^k\rangle$ and Alice measures
in the basis of preparation by Bob, for $t, u \in \{0, 1\}$. Let us assume for now that $E_{tu}=E$ for $t,u\in\{0,1\}$. Given $\lvert \Delta_b\rvert$, we note that the expectation value of $n_\text{error}$ equals $E\lvert \Delta_b\rvert$. We have 
$\frac{\gamma_{\text{err}}}{E}-1>0$ and $0<\gamma_\text{err}<1$ from (\ref{cor1repeated}) and by definition of $\gamma_\text{err}$. Thus, from Chernoff bound (\ref{secondChernoffbound}) of Proposition \ref{proposition3} with $p=E$ and $\epsilon=\frac{\gamma_\text{err}}{E}-1$, we have
\begin{equation}
\label{aaay2}
\text{Pr}\bigl[n_{\text{errors}}>\lvert \Delta_b\rvert \gamma_\text{err}\big\vert \lvert \Delta_b\rvert \bigr]\leq\left(\frac{E}{\gamma_\text{err}}\right)^{|\Delta_b|\gamma_\text{err}} \left(\frac{1-E}{1-\gamma_\text{err}}\right)^{|\Delta_b|(1-\gamma_\text{err})} \, .
\end{equation}

The function $f(E)=E^{\left|\Delta_b\right|\gamma_{err}}\left(1-E\right)^{\left|\Delta_b\right|\left(1-\gamma_{err}\right)}$ is increasing with increasing $E$, because from (\ref{cor1repeated}) we have that $f'(E)=\left|\Delta_b\right|\left[\gamma_{err}-E\right]E^{\left|\Delta_b\right|\gamma_{err}-1}\left(1-E\right)^{\left|\Delta_b\right|\left(1-\gamma_{err}\right)-1}>0$. 
Let $E_{\text{max}}\geq E$. Thus, from 
(\ref{aaay2}), we have
\begin{equation}
\label{aaay2.2}
\text{Pr}\bigl[n_{\text{errors}}\!>\!\lvert \Delta_b\rvert \gamma_\text{err}\big\vert \lvert \Delta_b\rvert \bigr]\!\leq\left(\frac{E_\text{max}}{\gamma_\text{err}}\right)^{|\Delta_b|\gamma_\text{err}} \left(\frac{1-E_\text{max}}{1-\gamma_\text{err}}\right)^{|\Delta_b|(1-\gamma_\text{err})} \, .
\end{equation}
It follows from (\ref{aaax1.1}) and (\ref{aaay2.2}) that 
\begin{equation}
\label{aaay3.2xx}
P_{\text{error}}(\lvert \Delta_b\rvert )\leq \left(\frac{E_\text{max}}{\gamma_\text{err}}\right)^{|\Delta_b|\gamma_\text{err}} \left(\frac{1-E_\text{max}}{1-\gamma_\text{err}}\right)^{|\Delta_b|(1-\gamma_\text{err})} \, .
\end{equation}
Since in general we have $E_{tu}\leq E$, for $t,u\in\{0,1\}$, we can replace $E_\text{max}$ by $E$ in (\ref{aaay3.2xx}) and obtain (\ref{aaay3.2x}).

We show (\ref{aaay3.2}). Since for the quantum state $\lvert \psi_k\rangle$, with $g(k)=j$, for $k\in\Lambda $ and $j\in[n]$, $\mathcal{B}$ encodes the bit $t_k=r_j$ in the basis labelled by $u_k=s_j$, with $u_k$ chosen with probability $P^{k}_{\text{PB}}(u_k)$ satisfying $\frac{1}{2}-\beta_\text{PB}\leq P^{k}_{\text{PB}}(u_k)\leq \frac{1}{2}+\beta_\text{PB}$ for $t_k,u_k\in\{0,1\}$, the probability $\text{Pr}[y_j=s_j]$ satisfies
\begin{equation}
\label{neweq3}
\text{Pr}[y_j=s_j]\geq \frac{1}{2}-\beta_\text{PB} \, .
\end{equation}
This is easily seen as follows. By the definition of $\Delta_b$ given in the token schemes $\mathcal{QT}_1$ and $\mathcal{QT}_2$, we see that $\lvert\Delta_b\rvert$ corresponds to the number of labels $k\in\Lambda$ satisfying $g(k)=j\in[n]$ for which it holds that $y_j=s_j$, where we recall $y_j$ and $s_j$ are the bits labelling the qubit measurement basis by Alice and the preparation basis by Bob, respectively. Thus, $\text{Pr}[y_j=s_j]=\sum_{a=0}^1 \text{Pr}[s_j=a]\text{Pr}[y_j=a] \geq \frac{1}{2}-\beta_\text{PB}$, as claimed. 

Next, we define a new variable $X=\sum_{k=1}^{N}X_k$, where $X_k\in\{0,1\}$ and $\Pr{\left[X_k=1\right]}=P_{det}\left(\frac{1}{2}-\beta_{PB}\right)$. By Proposition \ref{propositionprobcol}, it is seen that
\begin{equation}
	\Pr{\left[\left|\Delta_b\right|<\upsilon_{cor}N\right]}\le\Pr{\left[X<\upsilon_{cor}N\right]} \, , 
\end{equation}
as $\left|\Delta_b\right|$ can be considered as a sum of random variables $Y_i\in\{0,1\}$, with
\begin{equation}
\Pr{\left[Y_i=1\right]}=P_{\text{det}}\Pr{\left[y_j=s_j\right]\geq P_{\text{det}}\left(\frac{1}{2}-\beta_{PB}\right)}>0\,,
\end{equation}
satisfying the conditions of proposition \ref{propositionprobcol}.

We define
\begin{equation}
\label{aaay3.5}
 \epsilon=1-\frac{2\nu_\text{cor}}{P_{\text{det}}(1-2\beta_\text{PB})} \, .
\end{equation}
From the condition (\ref{cor1repeated}), we have $0<\epsilon<1$. Thus, from Chernoff bound (\ref{firstChernoffbound}) of Proposition \ref{proposition3} with $p=P_{\text{det}}\left(\frac{1}{2}-\beta_{PB}\right)$ and $\epsilon$ as in (\ref{aaay3.5}), we have
\begin{equation}
\label{aaay3.3}
\text{Pr}[\lvert \Delta_b\rvert<\nu_\text{cor} N]\leq\Pr{\left[X<\upsilon_{cor}N\right]}\leq\left(\frac{P_{det}\left(\frac{1}{2}-\beta_{PB}\right)\ \ }{\nu_{cor}}\right)^{N\nu_{cor}}\left(\frac{1-P_{det}\left(\frac{1}{2}-\beta_{PB}\right)\ \ }{1-\nu_{cor}}\right)^{N\left(1-\nu_{cor}\right)} \, , 
\end{equation}
which gives (\ref{aaay3.2}) as claimed. This completes the proof.

 \end{proof}

\stepcounter{lemma}
\stepcounter{lemma}

\subsection{Proof of unforgeability (proof of theorem 1)}

We first show unforgeability for the practical quantum token scheme implemented in our experiment in which Alice is not allowed to report losses. Then, we extend the proof straightforwardly to the general case in which losses can be reported and variations of the scheme in which Alice measures each received qubits in a random basis, as denoted by $\mathcal{QT}_1$ in Ref.~\cite{KLPGR22} (instead of measuring all qubits in the same, random, basis as done in our implementation and in the scheme $\mathcal{QT}_2$ of Ref.~\cite{KLPGR22}).

Recall that Bob sends Alice the quantum state $\rho_{\textbf{t}\textbf{u}}=\otimes_{k=1}^N\rho^k_{t_ku_k}$ with probability $P_{\text{PS}}(\textbf{t})P_{\text{PB}}(\textbf{u})$, where $P_{\text{PS}}(\textbf{t})=\prod_{k=1}^N P_{\text{PS}}^k(t_k)$ and $P_{\text{PB}}(\textbf{u})=\prod_{k=1}^N P_{\text{PB}}^k(u_k)$ for all $k\in[N]$, and where $\textbf{t}=(t_1,\ldots,t_N)$ denotes the encoded bits and $\textbf{u}=(u_1,\ldots,u_N)$ denotes the encoded bases.


\subsubsection{No losses are reported: all received quantum states are qubits}

\dami{We first consider the case in which Alice \adr{does not} report any losses to Bob.}

We assume that all quantum states received by Alice are qubits. Consider an arbitrary cheating strategy by Alice. Alice introduces an ancilla or arbitrary finite Hilbert space dimension in a pure state $\lvert \xi\rangle$ and applies a projective measurement on the whole Hilbert space. She obtains a measurement outcome $(c,v,\textbf{x}^0,\textbf{x}^1)$, where $c$ is a bit that Alice's agent A$_0$ must give Bob's agent B$_0$, $\textbf{x}^i$ is an $N$-bit strings that she presents to Bob at the spacetime region $R_i$, for $i=0,1$, and $v$ is any extra variable obtained by Alice, which can have an empty value. Bob accepts Alice's presented string $\textbf{x}^i$ as a valid token in $R_i$ if it holds that $d(\textbf{x}_i^i,\textbf{t}_i)\leq N_i\gamma_\text{err}$, where $\textbf{a}_i$ is the restriction of the $N$-bit string $\textbf{a}$ to the bits with labels in the set $\Delta _i=\{k\in[N]\vert u_k=c\oplus i\}$ and $N_i=\lvert \Delta_i\rvert $, for $i=0,1$, and where $d(\textbf{r},\textbf{s})$ is the Hamming distance between bit strings $\textbf{r}$ and $\textbf{s}$.

Alice's success probability in her forging strategy, conditioned on the values $\textbf{t},\textbf{u},c,v$, is 
\begin{eqnarray}
\label{p1}
P^{\textbf{t}\textbf{u}cv}_{\text{forge}}&=&\text{Pr}[d(\textbf{x}_0^0,\textbf{t}_0)\leq \gamma_\text{err}N_0, d(\textbf{x}_1^1,\textbf{t}_1)\leq \gamma_\text{err}N_1\vert \textbf{t}\textbf{u}cv]\nonumber\\
&\leq&\text{Pr}[d(\textbf{x}_0^0,\textbf{t}_0)+d(\textbf{x}_1^1,\textbf{t}_1)\leq \gamma_\text{err}N\vert \textbf{t}\textbf{u}cv]\,,
    \end{eqnarray}
where we have used $N=N_0+N_1$. For given values of $\textbf{u}$ and $c$, we define the $N-$bit strings $\textbf{y}^{c\textbf{u}}_0=(\textbf{x}_0^0,\textbf{x}_1^1)$ and $\textbf{y}^{c\textbf{u}}_1=(\textbf{x}_0^1,\textbf{x}_1^0)$. Thus, we have
\begin{equation}
    \label{p2}
    d(\textbf{x}_0^0,\textbf{t}_0)+d(\textbf{x}_1^1,\textbf{t}_1)=d(\textbf{y}^{c\textbf{u}}_0,\textbf{t})\,.
\end{equation}
From, (\ref{p1}) and (\ref{p2}), we have
\begin{equation}
\label{p3}
P^{\textbf{t}\textbf{u}cv}_{\text{forge}}\leq\text{Pr}[d(\textbf{y}^{c\textbf{u}}_0,\textbf{t})\leq \gamma_\text{err}N\vert \textbf{t}\textbf{u}cv]\,.
    \end{equation}
We define the string $\textbf{x}=\bigl((x_1^0,x_1^1),(x_2^0,x_2^1)\ldots,(x_N^0,x_N^1)\bigr)$ of $N$ two-bit values corresponding to the bits that Alice gives Bob at the presentation regions. We note that for each $k\in[N]$, the bit $x_k^0$ is presented at $R_0$ and the bit $x_k^1$ is presented at $R_1$. Thus, the condition $d(\textbf{y}^{c\textbf{u}}_0,\textbf{t})\leq  \gamma_\text{err}N$ in (\ref{p3}) is achieved if Alice makes no more than $\gamma_\text{err}N$ errors in producing two-bit number guesses $x_k=(x_k^0,x_k^1)$ satisfying $x_k^{c\oplus u_k}=t_k$ for $k\in[N]$.

This can be understood equivalently as follows. For the $k$th received state $\rho_{t_ku_k}^k$, Alice produces two guesses for the state, one given by the bit $x_k^c$ corresponding to guessing that the state is $\rho_{x_k^c0}^k$ and one given by the bit $x_k^{c\oplus 1}$ corresponding to guessing that the state is $\rho_{x_k^{c\oplus 1}1}^k$. She succeeds for the $k$th state if one of her two guesses includes the received state $\rho_{t_ku_k}^k$. In other words, for any $c\in\{0,1\}$, Alice's two-bit value $x_k$ corresponds to a guess $g^k\in\{1,2,3,4\}$ of one of the four sets $S^k_1=\{\rho_{00}^k,\rho_{01}^k\}$, $S^k_2=\{\rho_{10}^k,\rho_{01}^k\}$, $S_3^k=\{\rho_{10}^k,\rho_{11}^k\}$ and $S_4^k=\{\rho_{00}^k,\rho_{11}^k\}$, and she succeeds for the $k$th state if the received state $\rho_{t_ku_k}^k$ belongs to Alice's chosen set $S^k_{g^k}$. The upper bound in (\ref{p3}) corresponds to Alice's probability 
\begin{equation}
    \label{p4}
    P_\text{guess}(N\gamma_\text{err},N\vert \textbf{t}\textbf{u}cv)= \text{Pr}[d(\textbf{y}^{c\textbf{u}}_0,\textbf{t})\leq N \gamma_\text{err} \vert \textbf{t}\textbf{u}cv]
\end{equation}
of making no more than $N\gamma_\text{err}$ errors from the $N$ received states in the task just described.

We now connect this task with a standard quantum state discrimination task. Consider round $k\in[N]$. If Alice receives the state $\rho^k_{t_ku_k}$ from Bob, she wins in that round if her guessed set $S_{g_k}$ contains the state $\rho^k_{t_ku_k}$. 
We define $\rho_1^k=\rho_{00}^k$, $\rho_2^k=\rho_{01}^k$, $\rho_3^k=\rho_{10}^k$, $\rho_4^k=\rho_{11}^k$. We define $q_1^k=P^k_\text{PS}(0)P^k_\text{PB}(0)$, $q_2^k=P^k_\text{PS}(0)P^k_\text{PB}(1)$, $q_3^k=P^k_\text{PS}(1)P^k_\text{PB}(0)$, $q_4^k=P^k_\text{PS}(1)P^k_\text{PB}(1)$. With this notation, we have $S^k_1=\{\rho_1^k,\rho_2^k\}$, $S^k_2=\{\rho_2^k,\rho_3^k\}$, $S_3^k=\{\rho_3^k,\rho_4^k\}$ and $S_4^k=\{\rho_4^k,\rho_1^k\}$. Alice receives the state $\rho_{i^k}^k$ with probability $q_{i^k}^k$, and obtains a guess $g^k\in[4]$, for $i^k\in[4]$. Alice succeeds in the round $k$ if $g^k=i^k$ or if $g^k=i^k-1$, as both sets $S_{i^k}^k$ and $S_{i^k-1}^k$ contain the state $\rho_{i^k}^k$ received from Bob. We are using the notation $i\pm4=i$ for all $i$.


Consider now Alice playing this task as part of a collective strategy in which she applies a projective measurement $\{\pi_{g^ky}\}$ on a bigger Hilbert space of arbitrary finite dimension with possible extra outcomes $y$ (including guesses on states for other rounds) and including the states for rounds $j\neq k$  with any extra ancillas, denoted globally by a state $\phi_z$. 

The probability that Alice succeeds in the task, 
conditioned on the extra state being $\phi_z$ and the extra outcome being $y$, is
\begin{eqnarray}
\label{E1}
    P^{k,\text{win}}_{yz}&=&\sum_{i=1}^4 q_{i}^k \frac{Tr[(\pi_{i,y}+\pi_{i-1,y})(\rho_{i}^k\otimes \phi_z)]}{\sum_{j=1}^4Tr[\pi_{j,y} (\rho^k\otimes \phi_z)]}\nonumber\\
    &=&\sum_{i=1}^4 (q^k_{i}+q^k_{i+1})Tr  \Biggl[\frac{\pi_{i,y}\bigl((q^k_{i}\rho^k_{i}+q^k_{i+1}\rho^k_{i+1})\bigr)\otimes \phi_z}{(q^k_{i}+q^k_{i+1})Tr[\sum_{j=1}^4\pi_{j,y} (\rho^k\otimes \phi_z)]}\Biggr]\nonumber\\
    &=&2\sum_{i=1}^4 r^k_{i} Tr  \Biggl[\frac{\pi_{i,y}\bigl(\chi^k_{i}\otimes \phi_z\bigr)}{Tr[\sum_{j=1}^4\pi_{j,y} (\rho^k\otimes \phi_z)]}\Biggr]\, ,\nonumber\\
\end{eqnarray}
where $\rho^k=\sum_{i=1}^4 q^k_{i}\rho^k_{i}$, $r^k_{i}=\frac{q^k_{i}+q^k_{i+1}}{2}$, and $\chi^k_{i} =\frac{q^k_{i}\rho^k_{i}+q^k_{i+1}\rho^k_{i+1}}{q^k_{i}+q^k_{i+1}}$, for all $i=1,2,3,4$. We note that $\{r^k_{i}\}_{i=1}^4$ is a probability distribution and that $\chi^k_{i}$ is a density matrix, for all $i=1,2,3,4$. We also note that $\sum_{i=1}^4 r^k_{i}\chi^k_{i}=\rho^k$. Thus, we see that $ P^{k,\text{win}}_{yz}$ is twice the probability $P^{k,\text{discr}}_{yz}$ to succeed in a quantum state discrimination task given by the ensemble $\{r^k_{i},\chi^k_{i}\}_{i=1}^4$ conditioned on on the extra state being $\phi_z$ and the extra outcome being $y$:
\begin{equation}
    \label{E2}
    P^{k,\text{win}}_{yx}=2P^{k,\text{discr}}_{yx} \, .
\end{equation}

From (\ref{E2}) and from lemma~\ref{lemma1}, we have
\begin{equation}
    \label{E3}
    P^{k,\text{win}}_{yx}\leq 2\max_{j\in [4]}P_{\text{MC}}(\chi^k_j)\,,
\end{equation}
for all $x\in X$ and all $y\in Y$,
where
\begin{equation}
\label{E4}
 P_{\text{MC}}(\chi^k_j)=\max_{Q\geq 0}\frac{r^k_jTr[Q\chi^k_j]}{Tr[Q\rho^k]}   
\end{equation}
is the maximum confidence measurement that the received state was $\chi^k_j$ when Alice's outcome is $j\in [4]$ \cite{CABGJ06},
where the maximum is taken over all positive operators $Q$ acting on a two dimensional Hilbert space, where
\begin{equation}
\rho^k=\sum_{i=1}^4 r^k_i \chi^k_i\,,
\end{equation}
    and where
    \begin{equation}
    \label{E5}
r^k_i=\frac{q^k_i+q^k_{i+1}}{2},\qquad\chi^k_i =\frac{q^k_i\rho^k_i+q^k_{i+1}\rho^k_{i+1}}{q^k_i+q^k_{i+1}} \, , 
            \end{equation}
for all $i\in[4]$.

Thus, from (\ref{p4}) and (\ref{E3}), and from lemma~\ref{lemma2}, we have
\begin{equation}
    \label{E6}
    P_\text{guess}(N\gamma_\text{err},N\vert \textbf{t}\textbf{u}cv)\leq P_{\text{bound}}^{\text{coins}}(\lfloor N\gamma_\text{err} \rfloor,N)\,,
\end{equation}
where $P_{\text{bound}}^{\text{coins}}(\lfloor N\gamma_\text{err} \rfloor,N)$ is the probability of having no more than $\lfloor N\gamma_\text{err} \rfloor$ errors in $N$ independent coin tosses with success probabilities $P_\text{bound}^1,\ldots, P_\text{bound}^N$, 
where
\begin{equation}
\label{E7}
P_{\text{bound}}^k=2\max_{j\in [4]}P_{\text{MC}}(\chi_j^k)\,,
\end{equation}
for all $k\in [N]$.

Thus, from (\ref{p3}), (\ref{p4}), (\ref{E6}), (\ref{E7}) and lemma~\ref{lemma2}, we obtain that Alice's success probability $P^{cv}_{\text{forge}}$ in an arbitrary forging strategy for the quantum token scheme, conditioned on her bit outcome $c$ and on any extra variable $v$ that she obtains in her collective measurement, satisfies
\begin{eqnarray}
\label{E8}
P^{cv}_{\text{forge}}&=&\sum_{\textbf{t}}\sum_{\textbf{u}}P_{\text{PS}}(\textbf{t})P_{\text{PB}}(\textbf{u})P^{\textbf{t}\textbf{u}cv}_{\text{forge}}\nonumber\\
&\leq&\sum_{\textbf{t}}\sum_{\textbf{u}}P_{\text{PS}}(\textbf{t})P_{\text{PB}}(\textbf{u})\text{Pr}[d(\textbf{y}^{c\textbf{u}}_0,\textbf{t})\leq \gamma_\text{err}N\vert \textbf{t}\textbf{u}cv]\nonumber\\
&=&P_\text{guess}(N\gamma_\text{err},N\vert cv)\nonumber\\
&\leq& P_{\text{bound}}^{\text{coins}}(\lfloor N\gamma_\text{err} \rfloor,N)\nonumber\\
&\leq& \sum_{l=0}^{\lfloor N\gamma_\text{err} \rfloor}\left(\begin{smallmatrix}N\\ l\end{smallmatrix}\right)(1-P_{\text{bound}})^{l}(P_{\text{bound}})^{N-l} \, , 
    \end{eqnarray}
where
\begin{equation}
\label{E9}
P_{\text{bound}}\geq 2\max_{j\in [4]}P_{\text{MC}}(\chi_j^k)\, .
\end{equation}
for all $k\in [N]$.

\subsubsection{No losses are reported:  a fraction of quantum states are multiple qubits}

We now consider that there is a small probability $P_{\text{noqub}}>0$ that the quantum state $\rho^k= \rho_{t_ku_k}^k$ has Hilbert space dimension bigger than two, for all $k\in[N]$. Suppose that there are $m$ labels $k\in\Omega_{\text{noqub}}$ for which $\rho^k$ has Hilbert space dimension bigger than two. We allow Alice to be arbitrarily powerful only limited by the laws of quantum physics. Thus, we suppose that Alice knows the set $\Omega_{\text{noqub}}$ and that she can learn Bob's encoded bits $t_k$ in these quantum states.

Since Alice is allowed to make up to $N\gamma_\text{err}$ errors, she succeeds with unit probability in her cheating strategy if $m\geq N-N\gamma_\text{err}$. However, as discussed below, the probability that $m\geq N(1-\gamma_\text{err})$ is negligible if $P_\text{noqub}$ is small enough and if we choose $\gamma_\text{err}$ sufficiently small.

Thus, below we assume that
\begin{equation}
    \label{E9.1}
    m < N(1-\gamma_\text{err})\,.
\end{equation}
In this case, Alice's strategy reduces to play the original task on the $N-m$ qubit states, i.e., on $\rho^k$ with labels $k\in [N]\setminus \Omega_\text{noqub}$. Since lemma~\ref{lemma2} allows us to condition the probability to win the task on any extra quantum states and on any extra outcomes, the result in (\ref{E8}) can be straightforwardly extended for this case, as we can condition on any quantum states with labels $k\in \Omega_{\text{noqub}}$ and on any outputs of Alice for these labels. Thus, we can simply replace $N$ by $N-m$ and keep the number of allowed errors $\lfloor N\gamma_\text{err} \rfloor$ fixed in (\ref{E8}) to obtain that Alice's success probability $P^{mcv}_{\text{forge}}$ in an arbitrary forging strategy for the quantum token scheme, conditioned on the value $m=\lvert \Omega_\text{noqub}\rvert$, on her bit outcome $c$, on the set $\Omega_\text{noqub}$ and on any extra variable $v$ that she obtains in her collective measurement satisfies
\begin{equation}
\label{E10}
P^{mcv}_{\text{forge}}\leq  \sum_{l=0}^{\lfloor N\gamma_\text{err} \rfloor}\left(\begin{smallmatrix}N-m\\ l\end{smallmatrix}\right)(1-P_{\text{bound}})^{l}(P_{\text{bound}})^{N-m-l} \, .
    \end{equation}
for all $m,N\in\mathbb{N}$ and $\gamma_\text{err}\in[0,1)$ satisfying (\ref{E9.1}).

Let $\nu_\text{unf}$ be a constant satisfying
\begin{equation}
\label{E10.1}
0<P_\text{noqub}<\nu_\text{unf}<1-\frac{\gamma_\text{err}}{1-P_\text{bound}}\,.
\end{equation}
We have
\begin{equation}
\label{E10.2}
P^{cv}_{\text{forge}}= \sum_{m < N\nu_\text{unf}}\text{Pr}[\lvert \Omega_\text{noqub}\rvert=m] P^{mcv}_{\text{forge}}+ \sum_{m \geq N\nu_\text{unf}} \text{Pr}[\lvert \Omega_\text{noqub}\rvert=m] P^{mcv}_{\text{forge}} \, . 
    \end{equation}

We upper bound the second term of (\ref{E10.2}). We have
\begin{eqnarray}
\label{E10.3}
\sum_{m \geq N\nu_\text{unf}} \text{Pr}[\lvert \Omega_\text{noqub}\rvert=m] P^{mcv}_{\text{forge}}&\leq &\sum_{m \geq N\nu_\text{unf}} \text{Pr}[\lvert \Omega_\text{noqub}\rvert=m]\nonumber\\
&=&  \sum_{l=0}^{\lfloor N(1-\nu_\text{unf}) \rfloor}\left(\begin{smallmatrix}N\\ l\end{smallmatrix}\right)(1-P_{\text{noqub}})^{l}(P_{\text{noqub}})^{N-l} \, .\nonumber\\
\end{eqnarray}

We upper bound the first term of (\ref{E10.2}). Let $m<N\nu_\text{unf}$. We have
\begin{eqnarray}
\label{E12}
P^{mcv}_{\text{forge}}&\leq&  \sum_{l=0}^{\lfloor N\gamma_\text{err} \rfloor}\left(\begin{smallmatrix}N-m\\ l\end{smallmatrix}\right)(1-P_{\text{bound}})^{l}(P_{\text{bound}})^{N-m-l}\nonumber\\
&=&  P_\text{errors}(\lfloor N\gamma_\text{err} \rfloor,N-m,P_\text{bound})\,,
  \end{eqnarray}
where in the second line we used $P_\text{errors}(e,T,p)$ to denote the probability to make no more than $e$ errors in $T$ independent coin tosses with success probability $p$.

We can easily see that 
\begin{equation}
    \label{E12.1}
P_\text{errors}(e,T',p)\leq P_\text{errors}(e,T,p)\,,   
\end{equation}
for $T'\geq T$. It suffices to show this for $T'=T+1$ and arbitrary $e,T,p$. We have
\begin{equation}
    \label{E12.2}
P_\text{errors}(e,T+1,p)=P_\text{errors}(e,T,p)p+P_\text{errors}(e-1,T,p)(1-p)\nonumber\\
\leq P_\text{errors}(e,T,p),
    \end{equation}
where we used $P_\text{errors}(e-1,T,p)\leq P_\text{errors}(e,T,p)$.

Since $m<N\nu_\text{unf}$, from (\ref{E12}) and (\ref{E12.1}), we have
\begin{eqnarray}
   \label{E12.3}
P^{mcv}_{\text{forge}} &\leq& P_\text{errors}(\lfloor N\gamma_\text{err} \rfloor,N-m)\nonumber\\
&\leq& P_\text{errors}(\lfloor N\gamma_\text{err} \rfloor,\lceil N(1-\nu_\text{unf}) \rceil)\,.\nonumber\\
&=&  \sum_{l=0}^{\lfloor N\gamma_\text{err} \rfloor}\left(\begin{smallmatrix}\lceil N(1-\nu_\text{unf}) \rceil\\ l\end{smallmatrix}\right)(1-P_{\text{bound}})^{l}(P_{\text{bound}})^{\lceil N(1-\nu_\text{unf}) \rceil-l} \, .\nonumber\\
    \end{eqnarray}

Thus, we have
\begin{eqnarray}
    \label{F0}
    &&\sum_{m < N\nu_\text{unf}}\text{Pr}[\lvert \Omega_\text{noqub}\rvert=m] P^{mcv}_{\text{forge}}\nonumber\\
    &&\qquad\leq \sum_{l=0}^{\lfloor N\gamma_\text{err} \rfloor}\left(\begin{smallmatrix}\lceil N(1-\nu_\text{unf}) \rceil\\ l\end{smallmatrix}\right)(1-P_{\text{bound}})^{l}(P_{\text{bound}})^{\lceil N(1-\nu_\text{unf}) \rceil-l} \, .
\end{eqnarray}
Therefore, from (\ref{E10.2}), (\ref{E10.3}) and (\ref{E12.3}), we have
\begin{eqnarray}
    \label{F0.1}
    P^{cv}_{\text{forge}}&\leq&\sum_{l=0}^{\lfloor N(1-\nu_\text{unf}) \rfloor}\left(\begin{smallmatrix}N\\ l\end{smallmatrix}\right)(1-P_{\text{noqub}})^{l}(P_{\text{noqub}})^{N-l}\nonumber\\
    &&\quad\!\!\!+\sum_{l=0}^{\lfloor N\gamma_\text{err} \rfloor}\left(\begin{smallmatrix}\lceil N(1-\nu_\text{unf}) \rceil\\ l\end{smallmatrix}\right)(1-P_{\text{bound}})^{l}(P_{\text{bound}})^{\lceil N(1-\nu_\text{unf}) \rceil-l} \, .
\end{eqnarray}
Since the bound holds, for all $c=0,1$ and for any values of the extra variable $v$, we have
\begin{eqnarray}
    \label{F0.2}
    P_{\text{forge}}&\leq&\sum_{l=0}^{\lfloor N(1-\nu_\text{unf}) \rfloor}\left(\begin{smallmatrix}N\\ l\end{smallmatrix}\right)(1-P_{\text{noqub}})^{l}(P_{\text{noqub}})^{N-l}\nonumber\\
    &&\quad\!\!\!+\sum_{l=0}^{\lfloor N\gamma_\text{err} \rfloor}\left(\begin{smallmatrix}\lceil N(1-\nu_\text{unf}) \rceil\\ l\end{smallmatrix}\right)(1-P_{\text{bound}})^{l}(P_{\text{bound}})^{\lceil N(1-\nu_\text{unf}) \rceil-l} \, , 
\end{eqnarray}
where $\nu_\text{unf}$ is a constant satisfying (\ref{E10.1}).

\subsubsection{No losses are reported: allowing $\theta_k>\theta$ with a non-zero probability $P_\theta$}

Now suppose that there is a 
non-zero probability $P_\theta$ that $\theta_k>\theta$ for all $k\in[N]$. We define the set of labels $\Omega_{\text{noqub},\theta}=\{k\in[N]\vert \theta_k>\theta \text{~or~} k\in\Omega\}$. Let $P_{\text{noqub},\theta}$ be the probability that $k\in\Omega_{\text{noqub},\theta}$ for all $k\in[N]$. We have
\begin{equation}
    \label{F1}
    P_{\text{noqub},\theta}=1-(1-P_{\text{noqub}})(1-P_\theta)\,.
\end{equation}
Similarly to the security analysis above, we assume that Alice can learn the encoded bits perfectly for labels $k\in\Omega_{\text{noqub},\theta}$. Thus, the analysis proceeds straightforwardly as above by replacing $\Omega_{\text{noqub}}$ with $\Omega_{\text{noqub},\theta}$ and $P_{\text{noqub}}$ with $P_{\text{noqub},\theta}$. Thus, we have
\begin{eqnarray}
    \label{F1.1}
    P_{\text{forge}}&\leq&\sum_{l=0}^{\lfloor N(1-\nu_\text{unf}) \rfloor}\left(\begin{smallmatrix}N\\ l\end{smallmatrix}\right)(1-P_{\text{noqub},\theta})^{l}(P_{\text{noqub},\theta})^{N-l}\nonumber\\
    &&\quad\!\!\!+\sum_{l=0}^{\lfloor N\gamma_\text{err} \rfloor}\left(\begin{smallmatrix}\lceil N(1-\nu_\text{unf}) \rceil\\ l\end{smallmatrix}\right)(1-P_{\text{bound}})^{l}(P_{\text{bound}})^{\lceil N(1-\nu_\text{unf}) \rceil-l} \, , 
\end{eqnarray}
where $\nu_\text{unf}$ is a constant satisfying
\begin{equation}
\label{F1.2}
0<P_{\text{noqub},\theta}<\nu_\text{unf}<1-\frac{\gamma_\text{err}}{1-P_\text{bound}}\,.
\end{equation}

\subsubsection{\dami{The general case: losses} are reported}

We now \adr{extend our results to implementations in which (unlike the one we report in this work)} Alice is allowed to report a small fraction of lost quantum states. In order for Bob not to abort, we assume that Alice reports a set of pulse labels $\Lambda\subseteq [N]$ to Bob, satisfying
\begin{equation}
    \label{Z0}
    n=\lvert \Lambda\rvert \geq N\gamma_\text{det} \, , 
\end{equation}
for a predetermined $\gamma_\text{det}\in(0,1]$. Thus, Alice and Bob play the original task restricted to quantum state $\rho^k$ with labels $k\in\Lambda$, where the allowed error rate for valid token presentation remains fixed to $\gamma_\text{err}$. That is, the maximum allowed number of errors is $\lfloor n\gamma_\text{err}\rfloor$ in this case. We note that the case $\gamma_\text{det}=1$ comprises the case of no losses discussed above.

As mentioned above, we allow Alice to be arbitrarily powerful only limited by the laws of quantum physics. Thus, we assume that Alice receives all quantum states. Reporting a fraction of losses gives Alice an advantage because she can discard quantum states that give her the smallest probabilities to guess their encoded bits. Thus, we assume that Alice reports as lost the maximum possible number of single-qubit states and reports as received the maximum possible number of multi-qubit states. In other words, Alice chooses a set $\Lambda$ with the greatest possible overlap with $\Omega_\text{noqub}$.

Furthermore, we assume that Alice can learn Bob's encoded bits $t_k$ perfectly for $k\in\Omega_{\text{noqub},\theta}$, that is, if the $k$th pulse has more than one photon or if the deviation of its quantum state from the intended BB8a state is given by an $\theta_k>\theta$ on the Bloch sphere. Thus, in our security analysis, we assume that Alice chooses a set $\Lambda$ with the greatest possible overlap with $\Omega_{\text{noqub},\theta}$.

Since lemma~\ref{lemma2} allows us to condition the probability to win the considered task on any extra quantum states and on any extra outcomes, the result in (\ref{E10}) can be straightforwardly extended for this case, as we can condition on any quantum states with labels $k\in [N]\setminus \Lambda$ that were reported as lost by Alice and on any extra measurement outcomes including the set $\Lambda$ chosen by Alice. Thus, similarly to (\ref{E10.2}), Alice's cheating probability, conditioned on her bit value $c$ and on any other variables $v$, satisfies
\begin{equation}
\label{Z1}
P^{cv}_{\text{forge}}= \sum_{m < N\nu_\text{unf}}\text{Pr}[\lvert \Omega_{\text{noqub},\theta}\rvert=m] P^{mcv}_{\text{forge}}+ \sum_{m \geq N\nu_\text{unf}} \text{Pr}[\lvert \Omega_{\text{noqub},\theta}\rvert=m] P^{mcv}_{\text{forge}} \, ,  
    \end{equation}
for any constant $\nu_\text{unf}\in(0,1)$.  We choose $\nu_\text{unf}$ to satisfy
\begin{equation}
    \label{Z2}
    0<P_{\text{noqub},\theta}<\nu_\text{unf}<\gamma_\text{det}\Bigl(1-\frac{\gamma_\text{err}}{1-P_\text{bound}}\Bigr)\,,
\end{equation}
which allows us to show below that both terms in (\ref{Z1}) decrease exponentially with $N$. We include the set $\Lambda$ of reported pulses chosen by Alice in the variables $v$, satisfying $n=\lvert \Lambda\rvert\geq N\gamma_\text{det}$ for Bob not to abort. As discussed above, we assume that Alice can perfectly guess Bob's encoded bits $t_k$ for $k\in \Omega_{\text{noqub},\theta}$. For this reason, we have replaced $\Omega_{\text{noqub}}$ by $\Omega_{\text{noqub},\theta}$ in (\ref{E10.2}) to obtain (\ref{Z1}).

We bound the second term in (\ref{Z1}). We have
\begin{eqnarray}
\label{Z3}
     \sum_{m \geq N\nu_\text{unf}} \text{Pr}[\lvert \Omega_{\text{noqub},\theta}\rvert=m] P^{mcv}_{\text{forge}} \!\!\!&\leq&\!\!\!\!\! \sum_{m \geq N\nu_\text{unf}} \text{Pr}[\lvert \Omega_{\text{noqub},\theta}\rvert=m]\nonumber\\
     &=&\!\!\!\text{Pr}[\lvert \Omega_{\text{noqub},\theta}\rvert \geq N\nu_\text{unf}]\nonumber\\
     &\leq &\!\!\!\!\!\sum_{l=0}^{\lfloor N(1-\nu_\text{unf}) \rfloor}\left(\begin{smallmatrix}N\\ l\end{smallmatrix}\right)(1-P_{\text{noqub},\theta})^{l}(P_{\text{noqub},\theta})^{N-l} \, .\nonumber\\
\end{eqnarray}
We can guarantee the last line to decrease exponentially with $N$ if $ P_{\text{noqub},\theta}<\nu_\text{unf}$, as given by (\ref{Z2}).

We bound the first term in (\ref{Z1}). We assume that (\ref{Z1}) holds and that $m=\lvert \Omega_{\text{noqub},\theta} \rvert <N\nu_\text{unf}$. Thus, $m<N\gamma_\text{det}\leq n$. Therefore, Alice's optimal strategy is to choose $\Lambda$ as containing $\Omega_{\text{noqub},\theta}$: $\Lambda\supset \Omega_{\text{noqub},\theta}$. Thus, we can simply replace $N$ by $n=\lvert \Lambda\rvert$ in (\ref{E10}) and keep $m=\lvert \Omega_{\text{noqub},\theta}\rvert$ to have
\begin{eqnarray}
\label{Z4}
P^{mcv}_{\text{forge}}&\leq&  \sum_{l=0}^{\lfloor n\gamma_\text{err} \rfloor}\left(\begin{smallmatrix}n-m\\ l\end{smallmatrix}\right)(1-P_{\text{bound}})^{l}(P_{\text{bound}})^{n-m-l}\nonumber\\
&=&P_\text{errors}(\lfloor n\gamma_\text{err} \rfloor,n-m,P_\text{bound} )\nonumber\\
&\leq&P_\text{errors}(\lfloor n\gamma_\text{err} \rfloor,n- \lfloor N\nu_\text{unf} \rfloor, P_\text{bound})\nonumber\\
&=&\sum_{l=0}^{\lfloor n\gamma_\text{err} \rfloor}\left(\begin{smallmatrix}n- \lfloor N\nu_\text{unf} \rfloor\\ l\end{smallmatrix}\right)(1-P_{\text{bound}})^{l}(P_{\text{bound}})^{n- \lfloor N\nu_\text{unf} \rfloor-l} \, , 
    \end{eqnarray} 
for all $m,n\in\mathbb{N}$ and $\gamma_\text{err}\in[0,1)$ satisfying $m<n(1-\gamma_\text{err})$ instead of (\ref{E9.1}), which holds from $m<N\nu_\text{unf}$, (\ref{Z0}) and (\ref{Z2}); where in the second line we used the notation $P_\text{errors}(e,T,p)$ to denote the probability of making no more than $e$ errors in $T$ independent coin tosses with success probability $p$; and in the third line we used (\ref{E12.1}) and that $m<N\nu_\text{unf}$. 
Thus, from (\ref{Z4}), we have
\begin{eqnarray}
\label{Z5}
&&\sum_{m < N\nu_\text{unf}}\text{Pr}[\lvert \Omega_{\text{noqub},\theta}\rvert=m] P^{mcv}_{\text{forge}}\nonumber\\
&&\qquad\leq\sum_{l=0}^{\lfloor n\gamma_\text{err} \rfloor}\left(\begin{smallmatrix}n- \lfloor N\nu_\text{unf} \rfloor\\ l\end{smallmatrix}\right)(1-P_{\text{bound}})^{l}(P_{\text{bound}})^{n- \lfloor N\nu_\text{unf} \rfloor-l} \, .
    \end{eqnarray} 
We can guarantee the second line to decrease exponentially with $N$ from the conditions (\ref{Z0}) and (\ref{Z2}), as these conditions straightforwardly imply that $\lfloor n\gamma_\text{err}\rfloor <(1-P_\text{bound})(n-\lfloor N\nu_\text{unf} \rfloor)$.

We note that the bounds (\ref{Z3}) and (\ref{Z5}) hold for any values of the bit $c$, the set $\Lambda$ satisfying (\ref{Z0}) and any other extra variables included in $v$. Thus, from (\ref{Z1}), (\ref{Z3}) and (\ref{Z5}), we obtain that Alice's cheating probability satisfies the bound
\begin{eqnarray}
\label{Z6}
P_{\text{forge}} &\leq& \sum_{l=0}^{\lfloor N(1-\nu_\text{unf}) \rfloor}\left(\begin{smallmatrix}N\\ l\end{smallmatrix}\right)(1-P_{\text{noqub},\theta})^{l}(P_{\text{noqub},\theta})^{N-l}\nonumber\\
&&\quad + \sum_{l=0}^{\lfloor n\gamma_\text{err} \rfloor}\left(\begin{smallmatrix}n- \lfloor N\nu_\text{unf} \rfloor\\ l\end{smallmatrix}\right)(1-P_{\text{bound}})^{l}(P_{\text{bound}})^{n- \lfloor N\nu_\text{unf} \rfloor-l} \, , 
\end{eqnarray}
where (\ref{Z0}) and (\ref{Z2}) hold. This completes the proof of unforgeability claimed in \dami{theorem 1 for the quantum token scheme $\mathcal{QT}_2$ of Ref.~\cite{KLPGR22}}.

We note that if losses are not reported, that is, if $\gamma_\text{det}=1$, the condition (\ref{Z0}) reduces to $n=N$, and the condition (\ref{Z2}) reduces to (\ref{F1.2}). In this case, the bound (\ref{Z6}) reduces to the bound (\ref{F1.1}).

\subsubsection{\dami{Proof for the quantum token scheme $\mathcal{QT}_1$}}

\dami{

The proof given above for the quantum token scheme $\mathcal{QT}_2$ applies straightforwardly for the quantum token scheme $\mathcal{QT}_1$ of Ref.~\cite{KLPGR22}. The main difference between these two schemes is that in $\mathcal{QT}_2$ Alice measures all received quantum states in the same basis, chosen randomly, whereas in $\mathcal{QT}_1$ she measures each quantum state in a random basis which is independently chosen from the measurement bases applied to other quantum states. We note that Bob's actions during the quantum token generation are the same in both schemes.

As discussed above, our unforgeabilty proof for $\mathcal{QT}_2$ reduces to showing an upperbound $P_\text{bound}$ for the maximum confidence quantum measurement on a quantum state discrimination game associated to Alice's general cheating strategy in $\mathcal{QT}_2$ that applies to each of the qubit states that Bob sends Alice, and which only depend on Bob's actions during the quantum token generation. Straightforwardly, the unforgeability proof given above also applies to $\mathcal{QT}_1$ as the bound $P_\text{bound}$ also applies to the relevant quantum state discrimination game associated to Alice's general cheating strategy in $\mathcal{QT}_1$.

\section{Extension to a global network}
\label{sec:multiple}
As discussed in Ref.~\cite{KLPGR22}, the quantum token schemes $\mathcal{QT}_1$ and $\mathcal{QT}_2$ can be straightforwardly extended to schemes $\mathcal{QT}_1^M$ and $\mathcal{QT}_2^M$ with $2^M$
presentation regions, for arbitrary $M\geq 1$. Theorem 2 in Ref.~\cite{KLPGR22} can be updated with the tighter security guarantees $\epsilon_\text{priv}$, $\epsilon_\text{rob}$, $\epsilon_\text{cor}$ and $\epsilon_\text{unf}$ proved in this paper, as given by lemmas 1, 2 and 3 and theorem 1 of the main text.

\adr{We consider \dami{an} example of $M=7$, corresponding to $2^{7} =128$ presentation regions.
We first discuss the security levels that we could achieve with our experimental setup and parameter choices above, and then note refinements.  
This is a reasonable sized sub-network for trading strategies on the global financial network.   We consider the scheme $\mathcal{QT}_2^M$ with Alice not reporting losses, as in our implementation.

We first note that this scheme would require a transmission of $NM\dami{\approx 7 \times 10^4}$ quantum states from Bob to Alice during the quantum token generation, which would take \dami{$\approx 35$} minutes, since our implementation transmitted $N=10,048$ quantum states in approximately five minutes, which we recall correspond to the transmitted heralded photon pulses. We note that our scheme has the great advantage that the quantum communication stage can be made arbitrarily in the past of the presentation regions, hence, taking $\approx 35$ minutes for completing this stage is not really problematic. Nevertheless, if needed, this time could be reduced by having several optical setups working in parallel and/or by refining the setups to achieve higher transmission frequencies.  Note that schemes in which the presentation regions must be chosen before measuring the received quantum states (e.g.,~\cite{SERGTBW23}) would unavoidably involve significant delays between choosing the presentation regions and presenting the token there in realistic applications where the number of presentation regions is large, as in this example.}

Our implementation was perfectly robust ($\epsilon_\text{rob}=0$) as Alice did not report any losses, hence, the scheme does not allow Bob to abort.  Its extension to  $2^M$ presentation regions is thus also perfectly robust.

We showed our implementation to be $\epsilon_\text{priv}-$private given the assumptions of lemma 1 in the main text, with $\epsilon_\text{priv}=10^{-5}$. As discussed previously, $\epsilon_\text{priv}$ can be made arbitrarily small by pre-processing close to random bits. Thus, we can make $\mathcal{QT}_1^M$ and $\mathcal{QT}_2^M$ $\epsilon_\text{priv}^M-$private with
\begin{equation}
    \epsilon_\text{priv}^M = \frac{1}{2^M}\Bigl[(1+2\epsilon_\text{priv})^M-1 \Bigr]
\end{equation}
arbitrarily small too.

We showed our implementation to be $\epsilon_\text{cor}-$correct with $\epsilon_\text{cor}=2.1\times 10^{-11}$. Thus, we can make $\mathcal{QT}_2^M$ $\epsilon_\text{cor}^M-$correct with
\begin{equation}
    \epsilon_\text{cor}^M = M\epsilon_\text{cor}\,.
\end{equation}
\adr{In our example of $M=7$, with our experimental setup, we can guarantee sufficiently high correctness of $\epsilon_\text{cor}^M \approx 1.5 \times 10^{-10}$.}

\adr{Presenting and verifying a token of $MN \approx 7 \times 10^4$ measurement results, using our setup, would take approximately $7 \times 1.5 ~\mu{\rm s}\approx 11~\mu{\rm s}$. From Fig.~3 \dami{and Eq.~(12) of the main text, we see that a quantum advantage could be obtained for nodes separated by $\gtrsim 2.2$} km, 
meaning a large quantum advantage would be attained compared to classical-cross checking 
around a global (or even a typical national) network.  
We attained a comparative advantage of $\approx 39~\mu{\rm s}$ in our inter-city implementation, with nodes with a direct line separation of $\approx 50~{\rm km}$, implying a significant comparative advantage for $\mathcal{QT}_2^M$ for separations larger than this. This also implies a large comparative advantage for a global or large national network.
(Note that the relevant scale for these comparisons is the diameter of the network, not the separation between individual nodes.) 
These figures could be improved still further if the $M$ schemes 
are performed in parallel, with several  FPGAs used for presentation and verification, and of course
also by faster FPGAs.}

We showed our implementation to be $\epsilon_\text{unf}-$unforgeable with $\epsilon_\text{unf}=\dami{5.52 \times 10^{-9}}$. 
A very weak bound on the unforgeability of 
 $\mathcal{QT}_2^M$ can be obtained \cite{KLPGR22}
 by noting that, to successfully forge, Alice needs to present a validated token at at least one of the $\frac{1}{2} 2^M (2^{M}-1)$ pairs of distinct presentation points, which means her forging probability is bounded by
 \begin{equation}
\frac{1}{2} 2^M (2^{M}-1) \epsilon_\text{unf} \approx 4.5 \times 10^{-5}
\end{equation}
in our example. 

This bound can likely be significantly tightened.
Security could also easily be significantly improved further by using longer tokens (i.e. increasing $N$). As the figures above illustrate, $N$ can be increased by a significant factor without significantly diminishing the quantum and comparative advantages for a global network.   Again, this could also be compensated by using several FPGAs in parallel for token presentation and verification. 

In conclusion, our example illustrates that our implementation could be extended to realistic use cases -- for example, for trading on a global financial network -- with significant quantum and comparative advantage compared to classical cross-checking, while still achieving high security.

}

\section{Acknowledgements}   We thank Mathieu Bozzio, Tobias Guggemos, Peter Schiansky and Philip
Walther for very helpful discussions relating to Section 5 of this SI.  

\bibliographystyle{unsrt}
\bibliography{tokenbiblioSI,tokenbiblio}
\end{spacing}
\end{document}